\documentclass[aps,prx,superscriptaddress,
tightenlines,
longbibliography,nofootinbib,twocolumn]{revtex4-1}
\usepackage{soul}

\newcommand{\bel}[1]{{\color{black}#1}}

\newcommand{\Q}[1]{\ensuremath \mathcal{Q}_{#1}}
\newcommand{\B}[1]{\ensuremath \mathcal{B}_{#1}}
\newcommand{\C}[1]{\ensuremath \mathcal{C}_{#1}}
\newcommand{\ReS}{\ensuremath \mathbb{R}}
\newcommand{\id}{\ensuremath \mathbb{I}}

\newcommand{\X}{\ensuremath \mathbb{X}}
\newcommand{\A}{\ensuremath \mathbb{A}}
\newcommand{\Y}{\ensuremath \mathbb{Y}}
\newcommand{\Bo}{\ensuremath \mathbb{B}}
\newcommand{\As}{\ensuremath \boldsymbol{\Sigma}}

\usepackage{rotating}
\usepackage{floatpag}
\rotfloatpagestyle{empty}

\usepackage{dsfont}
\usepackage{amsmath,amsthm,amssymb}
\usepackage{xspace,enumerate,epsfig}
\usepackage{graphicx}
\graphicspath{{.}{./figures/}}
\usepackage{marvosym}

\usepackage{tikzfig}
\usepackage{stmaryrd}
\usepackage{docmute}
\usepackage{keycommand}

\newkeycommand{\pointmap}[style=point][1]{\,\tikz{\node[style=\commandkey{style}] (x) at (0,-0.3) {$#1$};\node [style=none] (3) at (0, 1.0) {};\node [style=none] (3) at (0, -1.0) {};\draw (x)--(0,0.7);}\,}

\newkeycommand{\typedkpoint}[style=kpoint,edgestyle=][2]{%
\,\begin{tikzpicture}
  \begin{pgfonlayer}{nodelayer}
    \node [style=none] (0) at (0, 1) {};
    \node [style=\commandkey{style}] (1) at (0, -0.75) {$#2$};
    \node [style=right label] (2) at (0.25, 0.5) {$#1$};
  \end{pgfonlayer}
  \begin{pgfonlayer}{edgelayer}
    \draw [\commandkey{edgestyle}] (1) to (0.center);
  \end{pgfonlayer}
\end{tikzpicture}\,}

\newkeycommand{\typedkpointdag}[style=kpoint adjoint,edgestyle=][2]{%
\,\begin{tikzpicture}
  \begin{pgfonlayer}{nodelayer}
    \node [style=none] (0) at (0, -1) {};
    \node [style=\commandkey{style}] (1) at (0, 0.75) {$#2$};
    \node [style=right label] (2) at (0.25, -0.5) {$#1$};
  \end{pgfonlayer}
  \begin{pgfonlayer}{edgelayer}
    \draw [\commandkey{edgestyle}] (0.center) to (1);
  \end{pgfonlayer}
\end{tikzpicture}\,}

\newcommand{\bistateadj}[1]{\,\begin{tikzpicture}[yshift=-3mm]
  \begin{pgfonlayer}{nodelayer}
    \node [style=kpointadj, minimum width=1 cm, inner sep=2pt] (0) at (0, 0.5) {$\psi$};
    \node [style=none] (1) at (-0.75, 0.25) {};
    \node [style=none] (2) at (0.75, 0.25) {};
    \node [style=none] (3) at (-0.75, -0.5) {};
    \node [style=none] (4) at (0.75, -0.5) {};
  \end{pgfonlayer}
  \begin{pgfonlayer}{edgelayer}
    \draw (1.center) to (3.center);
    \draw (2.center) to (4.center);
  \end{pgfonlayer}
\end{tikzpicture}\,}

\newkeycommand{\pointketbra}[style1=copoint,style2=point][2]{\,%
\begin{tikzpicture}
  \begin{pgfonlayer}{nodelayer}
    \node [style=none] (0) at (0, -1.5) {};
    \node [style=\commandkey{style1}] (1) at (0, -0.75) {$#1$};
    \node [style=\commandkey{style2}] (2) at (0, 0.75) {$#2$};
    \node [style=none] (3) at (0, 1.5) {};
  \end{pgfonlayer}
  \begin{pgfonlayer}{edgelayer}
    \draw (0.center) to (1);
    \draw (2) to (3.center);
  \end{pgfonlayer}
\end{tikzpicture}\,}

\newkeycommand{\twocopointketbra}[style1=copoint,style2=copoint,style3=point][3]{\,%
\begin{tikzpicture}
  \begin{pgfonlayer}{nodelayer}
    \node [style=none] (0) at (-0.75, -1.5) {};
    \node [style=none] (1) at (0.75, -1.5) {};
    \node [style=\commandkey{style1}] (2) at (-0.75, -0.75) {$#1$};
    \node [style=\commandkey{style2}] (3) at (0.75, -0.75) {$#2$};
    \node [style=\commandkey{style3}] (4) at (0, 0.75) {$#3$};
    \node [style=none] (5) at (0, 1.5) {};
  \end{pgfonlayer}
  \begin{pgfonlayer}{edgelayer}
    \draw (0.center) to (2);
    \draw (1.center) to (3);
    \draw (4) to (5.center);
  \end{pgfonlayer}
\end{tikzpicture}\,}

\newkeycommand{\twopointketbra}[style1=copoint,style2=point,style3=point][3]{\,%
\begin{tikzpicture}
  \begin{pgfonlayer}{nodelayer}
    \node [style=none] (0) at (0, -1.5) {};
    \node [style=none] (1) at (-0.75, 1.5) {};
    \node [style=\commandkey{style1}] (2) at (0, -0.75) {$#1$};
    \node [style=\commandkey{style2}] (3) at (-0.75, 0.75) {$#2$};
    \node [style=\commandkey{style3}] (4) at (0.75, 0.75) {$#3$};
    \node [style=none] (5) at (0.75, 1.5) {};
  \end{pgfonlayer}
  \begin{pgfonlayer}{edgelayer}
    \draw (0.center) to (2);
    \draw (1.center) to (3);
    \draw (4) to (5.center);
  \end{pgfonlayer}
\end{tikzpicture}\,}

\newkeycommand{\pointbraket}[style1=point,style2=copoint][2]{\,%
\begin{tikzpicture}
  \begin{pgfonlayer}{nodelayer}
    \node [style=\commandkey{style1}] (0) at (0, -0.5) {$#1$};
    \node [style=\commandkey{style2}] (1) at (0, 0.5) {$#2$};
  \end{pgfonlayer}
  \begin{pgfonlayer}{edgelayer}
    \draw (0) to (1);
  \end{pgfonlayer}
\end{tikzpicture}\,}

\newkeycommand{\kpointbraket}[style1=kpoint,style2=kpoint adjoint][2]{\,%
\begin{tikzpicture}
  \begin{pgfonlayer}{nodelayer}
    \node [style=\commandkey{style1}] (0) at (0, -0.75) {$#1$};
    \node [style=\commandkey{style2}] (1) at (0, 0.75) {$#2$};
  \end{pgfonlayer}
  \begin{pgfonlayer}{edgelayer}
    \draw (0) to (1);
  \end{pgfonlayer}
\end{tikzpicture}\,}

\newkeycommand{\kpointmap}[style1=kpoint,style2=map][2]{\,%
\begin{tikzpicture}
  \begin{pgfonlayer}{nodelayer}
    \node [style=\commandkey{style1}] (0) at (0, -1.1) {$#1$};
    \node [style=\commandkey{style2}] (1) at (0, 0.2) {$#2$};
    \node [style=none] (2) at (0, 1.5) {};
  \end{pgfonlayer}
  \begin{pgfonlayer}{edgelayer}
    \draw (0) to (1);
    \draw (1) to (2.center);
  \end{pgfonlayer}
\end{tikzpicture}%
\,}

\newkeycommand{\sandwichtwo}[style1=point,style2=map,style3=map,style4=copoint][4]{\,%
\begin{tikzpicture}
  \begin{pgfonlayer}{nodelayer}
    \node [style=\commandkey{style2}] (0) at (0, -0.75) {$#2$};
    \node [style=\commandkey{style3}] (1) at (0, 0.75) {$#3$};
    \node [style=\commandkey{style1}] (2) at (0, -2) {$#1$};
    \node [style=\commandkey{style4}] (3) at (0, 2) {$#4$};
  \end{pgfonlayer}
  \begin{pgfonlayer}{edgelayer}
    \draw (2) to (0);
    \draw (0) to (1);
    \draw (1) to (3);
  \end{pgfonlayer}
\end{tikzpicture}\,}

\newkeycommand{\longbraket}[style1=point,style2=copoint][2]{\,%
\begin{tikzpicture}
  \begin{pgfonlayer}{nodelayer}
    \node [style=\commandkey{style1}] (0) at (0, -2) {$#1$};
    \node [style=\commandkey{style2}] (1) at (0, 2) {$#2$};
  \end{pgfonlayer}
  \begin{pgfonlayer}{edgelayer}
    \draw (0) to (1);
  \end{pgfonlayer}
\end{tikzpicture}\,}

\newkeycommand{\onb}[style=point]{\ensuremath{\left\{\tikz{\node[style=\commandkey{style}] (x) at (0,-0.3) {$j$};\draw (x)--(0,0.7);}\right\}}\xspace}

\newkeycommand{\copointmap}[style=copoint][1]{\,\tikz{\node[style=\commandkey{style}] (x) at (0,0.3) {$#1$};\draw (x)--(0,-0.7);}\,}

\newcommand{\bra}[1]{\ensuremath{\left\langle #1 \right|}}
\newcommand{\ket}[1]{\ensuremath{\left|  #1 \right\rangle}}

\newcommand{\ketbra}[2]{\ensuremath{\ket{#1}\!\bra{#2}}}

\tikzstyle{cdiag}=[matrix of math nodes, row sep=3em, column sep=3em, text height=1.5ex, text depth=0.25ex,inner sep=0.5em]
\tikzstyle{arrow above}=[transform canvas={yshift=0.5ex}]
\tikzstyle{arrow below}=[transform canvas={yshift=-0.5ex}]

\usetikzlibrary{shapes.multipart}

\tikzstyle{c}=[color=gray,line width = .75pt]
\tikzstyle{q}=[thin, color=black, double]
\tikzstyle{g}=[line width = 1.5pt, color=black]
\def\ladderwidth{.2mm}
\def\ladderstep{.5mm}
\tikzstyle{b}=[decorate,decoration={
      markings,
      mark=between positions {1/#1/2} and {-1/#1/2} step {1/#1} with {
        \draw[thin,black]
        (0,-\ladderwidth) -- (0,\ladderwidth)
        (\ladderstep,-\ladderwidth) -- (-\ladderstep,-\ladderwidth)
        (\ladderstep,\ladderwidth) -- (-\ladderstep,\ladderwidth);
      }
    }]

\tikzstyle{bwSpider}=[
       rectangle split,
       rectangle split parts=2,
       rectangle split part fill={black,white},
 minimum size=3.6 mm, inner sep=-2mm, draw=black,scale=0.5,rounded corners=0.8 mm
       ]
 \tikzstyle{wbSpider}=[
       rectangle split,
       rectangle split parts=2,
       rectangle split part fill={white,black},
 minimum size=3.6 mm, inner sep=-2mm, draw=black,scale=0.5,rounded corners=0.8 mm
       ]
\tikzstyle{qWire}=[line width = 1pt, color=black]
\tikzstyle{cWire}=[color=quantumgray,thin]
\tikzstyle{env}=[copoint,regular polygon rotate=0,minimum width=0.2cm, fill=black]

\tikzstyle{probs}=[shape=semicircle,fill=white,draw=black,shape border rotate=180,minimum width=1.2cm]

\tikzstyle{every picture}=[baseline=-0.25em,scale=0.5]
\tikzstyle{dotpic}=[] 
\tikzstyle{diredges}=[every to/.style={diredge}]
\tikzstyle{math matrix}=[matrix of math nodes,left delimiter=(,right delimiter=),inner sep=2pt,column sep=1em,row sep=0.5em,nodes={inner sep=0pt},text height=1.5ex, text depth=0.25ex]

\tikzstyle{inline text}=[text height=1.5ex, text depth=0.25ex,yshift=0.5mm]
\tikzstyle{label}=[font=\scriptsize,text height=1.5ex, text depth=0.25ex,yshift=0.5mm]
\tikzstyle{left label}=[label,anchor=east,xshift=1.5mm]
\tikzstyle{right label}=[label,anchor=west,xshift=-1mm]

\tikzstyle{braceedge}=[decorate,decoration={brace,amplitude=2mm,raise=-1mm}]
\tikzstyle{small braceedge}=[decorate,decoration={brace,amplitude=1mm,raise=-1mm}]

\tikzstyle{doubled}=[line width=1.6pt] 
\tikzstyle{boldedge}=[doubled,shorten <=-0.17mm,shorten >=-0.17mm]
\tikzstyle{boldedgegray}=[doubled,gray,shorten <=-0.17mm,shorten >=-0.17mm]
\tikzstyle{singleedgegray}=[gray]

\tikzstyle{semidoubled}=[line width=1.4pt] 
\tikzstyle{semiboldedgegray}=[semidoubled,gray,shorten <=-0.17mm,shorten >=-0.17mm]

\tikzstyle{boxedge}=[semiboldedgegray]

\tikzstyle{boldedgedashed}=[very thick,dashed,shorten <=-0.17mm,shorten >=-0.17mm]
\tikzstyle{vboldedgedashed}=[doubled,dashed,shorten <=-0.17mm,shorten >=-0.17mm]
\tikzstyle{left hook arrow}=[left hook-latex]
\tikzstyle{right hook arrow}=[right hook-latex]
\tikzstyle{sembracket}=[line width=0.5pt,shorten <=-0.07mm,shorten >=-0.07mm]

\tikzstyle{causal edge}=[->,thick,gray]
\tikzstyle{causal nondir}=[thick,gray]
\tikzstyle{timeline}=[thick,gray, dashed]

\tikzstyle{cedge}=[<->,thick,gray!70!white]

\tikzstyle{empty diagram}=[draw=gray!40!white,dashed,shape=rectangle,minimum width=1cm,minimum height=1cm]
\tikzstyle{empty diagram small}=[draw=gray!50!white,dashed,shape=rectangle,minimum width=0.6cm,minimum height=0.5cm]

\tikzstyle{dot}=[inner sep=0mm,minimum width=2mm,minimum height=2mm,draw,shape=circle]
\tikzstyle{leak}=[white dot, shape=regular polygon, minimum size=3.3 mm, regular polygon sides=3, outer sep=-0.2mm, regular polygon rotate=270]
\tikzstyle{proj}=[draw,fill=white,chamfered rectangle,chamfered rectangle angle=30, minimum width=2mm, minimum height= 1mm,scale=0.5, outer sep=-0.2mm]
\tikzstyle{wide proj}=[draw,fill=white,chamfered rectangle,chamfered rectangle angle=30, minimum width=15mm,minimum height=1mm,scale=0.5, outer sep=-0.2mm]
\tikzstyle{very wide proj}=[draw,fill=white,chamfered rectangle,chamfered rectangle angle=30, minimum width=25mm,minimum height=1mm,scale=0.5, outer sep=-0.2mm]
\tikzstyle{very very wide proj}=[draw,fill=white,chamfered rectangle,chamfered rectangle angle=30, minimum width=35mm,minimum height=1mm,scale=0.5, outer sep=-0.2mm]
\tikzstyle{preleak}=[proj]

\tikzstyle{split proj out}=[regular polygon,regular polygon sides=3,draw,scale=0.75,inner sep=-0.5pt,minimum width=3.3mm,fill=white,regular polygon rotate=180]
\tikzstyle{split proj in}=[regular polygon,regular polygon sides=3,draw,scale=0.75,inner sep=-0.5pt,minimum width=3.3mm,fill=white]
\tikzstyle{Vleak}=[white dot, shape=regular polygon, minimum size=3.3 mm, regular polygon sides=3, outer sep=-0.2mm, regular polygon rotate=90]
\tikzstyle{dleak}=[white dot, line width=1.6pt, shape=regular polygon, minimum size=3.3 mm, regular polygon sides=3, outer sep=-0.2mm, regular polygon rotate=270]

\tikzstyle{Wsquare}=[white dot, shape=regular polygon, rounded corners=0.8 mm, minimum size=3.3 mm, regular polygon sides=3, outer sep=-0.2mm]
\tikzstyle{Wsquareadj}=[white dot, shape=regular polygon, rounded corners=0.8 mm, minimum size=3.3 mm, regular polygon sides=3, outer sep=-0.2mm, regular polygon rotate=180]
\tikzstyle{ddot}=[inner sep=0mm, doubled, minimum width=2.5mm,minimum height=2.5mm,draw,shape=circle]

\tikzstyle{black dot}=[dot,fill=black]
\tikzstyle{white dot}=[dot,fill=white,,text depth=-0.2mm]
\tikzstyle{white Wsquare}=[Wsquare,fill=gray,,text depth=-0.2mm]
\tikzstyle{white Wsquareadj}=[Wsquareadj,fill=white,,text depth=-0.2mm]
\tikzstyle{green dot}=[white dot] 
\tikzstyle{gray dot}=[dot,fill=gray!40!white,,text depth=-0.2mm]
\tikzstyle{red dot}=[gray dot] 

\tikzstyle{black ddot}=[ddot,fill=black]
\tikzstyle{white ddot}=[ddot,fill=white]
\tikzstyle{gray ddot}=[ddot,fill=gray!40!white]

\tikzstyle{gray edge}=[gray!60!white]

\tikzstyle{small dot}=[inner sep=0.5mm,minimum width=0pt,minimum height=0pt,draw,shape=circle]

\tikzstyle{small black dot}=[small dot,fill=black]
\tikzstyle{small white dot}=[small dot,fill=white]
\tikzstyle{small gray dot}=[small dot,fill=gray!40!white]

\tikzstyle{causal dot}=[inner sep=0.4mm,minimum width=0pt,minimum height=0pt,draw=white,shape=circle,fill=gray!40!white]

\tikzstyle{phase dimensions}=[minimum size=5mm,font=\footnotesize,rectangle,rounded corners=2.5mm,inner sep=0.2mm,outer sep=-2mm]
\tikzstyle{dphase dimensions}=[minimum size=5mm,font=\footnotesize,rectangle,rounded corners=2.5mm,inner sep=0.2mm,outer sep=-2mm]

\tikzstyle{white phase dot}=[dot,fill=white,phase dimensions]
\tikzstyle{white phase ddot}=[ddot,fill=white,dphase dimensions]

\tikzstyle{white rect ddot}=[draw=black,fill=white,doubled,minimum size=5mm,font=\footnotesize,rectangle,rounded corners=2.5mm,inner sep=0.2mm]
\tikzstyle{gray rect ddot}=[draw=black,fill=gray!40!white,doubled,minimum size=6mm,font=\footnotesize,rectangle,rounded corners=3mm]

\tikzstyle{gray phase dot}=[dot,fill=gray!40!white,phase dimensions]
\tikzstyle{gray phase ddot}=[ddot,fill=gray!40!white,dphase dimensions]
\tikzstyle{grey phase dot}=[gray phase dot]
\tikzstyle{grey phase ddot}=[gray phase ddot]

\tikzstyle{small phase dimensions}=[minimum size=4mm,font=\tiny,rectangle,rounded corners=2mm,inner sep=0.2mm,outer sep=-2mm]
\tikzstyle{small dphase dimensions}=[minimum size=4mm,font=\tiny,rectangle,rounded corners=2mm,inner sep=0.2mm,outer sep=-2mm]

\tikzstyle{small gray phase dot}=[dot,fill=gray!40!white,small phase dimensions]
\tikzstyle{small gray phase ddot}=[ddot,fill=gray!40!white,small dphase dimensions]

\tikzstyle{small map}=[draw,shape=rectangle,minimum height=4mm,minimum width=4mm,fill=white]

\tikzstyle{cnot}=[fill=white,shape=circle,inner sep=-1.4pt]

\tikzstyle{asym hadamard}=[fill=white,draw,shape=NEbox,inner sep=0.6mm,font=\footnotesize,minimum height=4mm]
\tikzstyle{asym hadamard conj}=[fill=white,draw,shape=NWbox,inner sep=0.6mm,font=\footnotesize,minimum height=4mm]
\tikzstyle{asym hadamard dag}=[fill=white,draw,shape=SEbox,inner sep=0.6mm,font=\footnotesize,minimum height=4mm]

\tikzstyle{hadamard}=[fill=white,draw,inner sep=0.6mm,font=\footnotesize,minimum height=4mm,minimum width=4mm]
\tikzstyle{small hadamard}=[fill=white,draw,inner sep=0.6mm,minimum height=1.5mm,minimum width=1.5mm]
\tikzstyle{small hadamard rotate}=[small hadamard,rotate=45]
\tikzstyle{dhadamard}=[hadamard,doubled]
\tikzstyle{small dhadamard}=[small hadamard,doubled]
\tikzstyle{small dhadamard rotate}=[small hadamard rotate,doubled]
\tikzstyle{antipode}=[white dot,inner sep=0.3mm,font=\footnotesize]

\tikzstyle{scalar}=[diamond,draw,inner sep=0.5pt,font=\small]
\tikzstyle{dscalar}=[diamond,doubled, draw,inner sep=0.5pt,font=\small]

\tikzstyle{small box}=[rectangle,inline text,fill=white,draw,minimum height=5mm,yshift=-0.5mm,minimum width=5mm,font=\small]
\tikzstyle{small gray box}=[small box,fill=gray!30]
\tikzstyle{medium box}=[rectangle,inline text,fill=white,draw,minimum height=5mm,yshift=-0.5mm,minimum width=10mm,font=\small]
\tikzstyle{square box}=[small box] 
\tikzstyle{medium gray box}=[small box,fill=gray!30]
\tikzstyle{semilarge box}=[rectangle,inline text,fill=white,draw,minimum height=5mm,yshift=-0.5mm,minimum width=12.5mm,font=\small]
\tikzstyle{large box}=[rectangle,inline text,fill=white,draw,minimum height=5mm,yshift=-0.5mm,minimum width=15mm,font=\small]
\tikzstyle{large gray box}=[small box,fill=gray!30]

\tikzstyle{Bayes box}=[rectangle,fill=black,draw, minimum height=3mm, minimum width=3mm]

\tikzstyle{gray square point}=[small box,fill=gray!50]

\tikzstyle{dphase box white}=[dhadamard]
\tikzstyle{dphase box gray}=[dhadamard,fill=gray!50!white]
\tikzstyle{phase box white}=[hadamard]
\tikzstyle{phase box gray}=[hadamard,fill=gray!50!white]

\tikzstyle{point}=[regular polygon,regular polygon sides=3,draw,scale=0.75,inner sep=-0.5pt,minimum width=9mm,fill=white,regular polygon rotate=180]
\tikzstyle{point nosep}=[regular polygon,regular polygon sides=3,draw,scale=0.75,inner sep=-2pt,minimum width=9mm,fill=white,regular polygon rotate=180]
\tikzstyle{copoint}=[regular polygon,regular polygon sides=3,draw,scale=0.75,inner sep=-0.5pt,minimum width=9mm,fill=white]
\tikzstyle{dpoint}=[point,doubled]
\tikzstyle{dcopoint}=[copoint,doubled]

\tikzstyle{pointgrow}=[shape=cornerpoint,kpoint common,scale=0.75,inner sep=3pt]
\tikzstyle{pointgrow dag}=[shape=cornercopoint,kpoint common,scale=0.75,inner sep=3pt]

\tikzstyle{wide copoint}=[fill=white,draw,shape=isosceles triangle,shape border rotate=90,isosceles triangle stretches=true,inner sep=0pt,minimum width=1.5cm,minimum height=6.12mm]
\tikzstyle{wide point}=[fill=white,draw,shape=isosceles triangle,shape border rotate=-90,isosceles triangle stretches=true,inner sep=0pt,minimum width=1.5cm,minimum height=6.12mm,yshift=-0.0mm]
\tikzstyle{wide point plus}=[fill=white,draw,shape=isosceles triangle,shape border rotate=-90,isosceles triangle stretches=true,inner sep=0pt,minimum width=1.74cm,minimum height=7mm,yshift=-0.0mm]

\tikzstyle{wide dpoint}=[fill=white,doubled,draw,shape=isosceles triangle,shape border rotate=-90,isosceles triangle stretches=true,inner sep=0pt,minimum width=1.5cm,minimum height=6.12mm,yshift=-0.0mm]

\tikzstyle{tinypoint}=[regular polygon,regular polygon sides=3,draw,scale=0.55,inner sep=-0.15pt,minimum width=6mm,fill=white,regular polygon rotate=180]

\tikzstyle{white point}=[point]
\tikzstyle{white dpoint}=[dpoint]
\tikzstyle{green point}=[white point] 
\tikzstyle{white copoint}=[copoint]
\tikzstyle{gray point}=[point,fill=gray!40!white]
\tikzstyle{gray dpoint}=[gray point,doubled]
\tikzstyle{red point}=[gray point] 
\tikzstyle{gray copoint}=[copoint,fill=gray!40!white]
\tikzstyle{gray dcopoint}=[gray copoint,doubled]

\tikzstyle{white point guide}=[regular polygon,regular polygon sides=3,font=\scriptsize,draw,scale=0.65,inner sep=-0.5pt,minimum width=9mm,fill=white,regular polygon rotate=180]

\tikzstyle{black point}=[point,fill=black,font=\color{white}]
\tikzstyle{black copoint}=[copoint,fill=black,font=\color{white}]

\tikzstyle{tiny gray point}=[tinypoint,fill=gray!40!white]

\tikzstyle{diredge}=[->]
\tikzstyle{ddiredge}=[<->]
\tikzstyle{rdiredge}=[<-]
\tikzstyle{thickdiredge}=[->, very thick]
\tikzstyle{pointer edge}=[->,very thick,gray]
\tikzstyle{pointer edge part}=[very thick,gray]
\tikzstyle{dashed edge}=[dashed]
\tikzstyle{thick dashed edge}=[very thick,dashed]
\tikzstyle{thick gray dashed edge}=[thick dashed edge,gray!40]
\tikzstyle{thick map edge}=[very thick,|->]

\makeatletter
\newcommand{\boxshape}[3]{%
\pgfdeclareshape{#1}{
\inheritsavedanchors[from=rectangle] 
\inheritanchorborder[from=rectangle]
\inheritanchor[from=rectangle]{center}
\inheritanchor[from=rectangle]{north}
\inheritanchor[from=rectangle]{south}
\inheritanchor[from=rectangle]{west}
\inheritanchor[from=rectangle]{east}
\backgroundpath{
\southwest \pgf@xa=\pgf@x \pgf@ya=\pgf@y
\northeast \pgf@xb=\pgf@x \pgf@yb=\pgf@y

\@tempdima=#2
\@tempdimb=#3

\pgfpathmoveto{\pgfpoint{\pgf@xa - 5pt + \@tempdima}{\pgf@ya}}
\pgfpathlineto{\pgfpoint{\pgf@xa - 5pt - \@tempdima}{\pgf@yb}}
\pgfpathlineto{\pgfpoint{\pgf@xb + 5pt + \@tempdimb}{\pgf@yb}}
\pgfpathlineto{\pgfpoint{\pgf@xb + 5pt - \@tempdimb}{\pgf@ya}}
\pgfpathlineto{\pgfpoint{\pgf@xa - 5pt + \@tempdima}{\pgf@ya}}
\pgfpathclose
}
}}

\boxshape{NEbox}{0pt}{3pt}
\boxshape{SEbox}{0pt}{-3pt}
\boxshape{NWbox}{5pt}{0pt}
\boxshape{SWbox}{-5pt}{0pt}
\boxshape{EBox}{-3pt}{3pt}
\boxshape{WBox}{3pt}{-3pt}
\makeatother

\tikzstyle{cloud}=[shape=cloud,draw,minimum width=1.5cm,minimum height=1.5cm]

\tikzstyle{map}=[draw,shape=NEbox,inner sep=1pt,minimum height=4mm,fill=white]
\tikzstyle{dashedmap}=[draw,dashed,shape=NEbox,inner sep=2pt,minimum height=6mm,fill=white]
\tikzstyle{mapdag}=[draw,shape=SEbox,inner sep=1pt,minimum height=4mm,fill=white]
\tikzstyle{mapadj}=[draw,shape=SEbox,inner sep=2pt,minimum height=6mm,fill=white]
\tikzstyle{maptrans}=[draw,shape=SWbox,inner sep=2pt,minimum height=6mm,fill=white]
\tikzstyle{mapconj}=[draw,shape=NWbox,inner sep=2pt,minimum height=6mm,fill=white]

\tikzstyle{medium map}=[draw,shape=NEbox,inner sep=2pt,minimum height=6mm,fill=white,minimum width=7mm]
\tikzstyle{medium map dag}=[draw,shape=SEbox,inner sep=2pt,minimum height=6mm,fill=white,minimum width=7mm]
\tikzstyle{medium map adj}=[draw,shape=SEbox,inner sep=2pt,minimum height=6mm,fill=white,minimum width=7mm]
\tikzstyle{medium map trans}=[draw,shape=SWbox,inner sep=2pt,minimum height=6mm,fill=white,minimum width=7mm]
\tikzstyle{medium map conj}=[draw,shape=NWbox,inner sep=2pt,minimum height=6mm,fill=white,minimum width=7mm]
\tikzstyle{semilarge map}=[draw,shape=NEbox,inner sep=2pt,minimum height=6mm,fill=white,minimum width=9.5mm]
\tikzstyle{semilarge map trans}=[draw,shape=SWbox,inner sep=2pt,minimum height=6mm,fill=white,minimum width=9.5mm]
\tikzstyle{semilarge map adj}=[draw,shape=SEbox,inner sep=2pt,minimum height=6mm,fill=white,minimum width=9.5mm]
\tikzstyle{semilarge map dag}=[draw,shape=SEbox,inner sep=2pt,minimum height=6mm,fill=white,minimum width=9.5mm]
\tikzstyle{semilarge map conj}=[draw,shape=NWbox,inner sep=2pt,minimum height=6mm,fill=white,minimum width=9.5mm]
\tikzstyle{large map}=[draw,shape=NEbox,inner sep=2pt,minimum height=6mm,fill=white,minimum width=12mm]
\tikzstyle{large map conj}=[draw,shape=NWbox,inner sep=2pt,minimum height=6mm,fill=white,minimum width=12mm]
\tikzstyle{very large map}=[draw,shape=NEbox,inner sep=2pt,minimum height=6mm,fill=white,minimum width=17mm]

\tikzstyle{medium dmap}=[draw,doubled,shape=NEbox,inner sep=2pt,minimum height=6mm,fill=white,minimum width=7mm]
\tikzstyle{medium dmap dag}=[draw,doubled,shape=SEbox,inner sep=2pt,minimum height=6mm,fill=white,minimum width=7mm]
\tikzstyle{medium dmap adj}=[draw,doubled,shape=SEbox,inner sep=2pt,minimum height=6mm,fill=white,minimum width=7mm]
\tikzstyle{medium dmap trans}=[draw,doubled,shape=SWbox,inner sep=2pt,minimum height=6mm,fill=white,minimum width=7mm]
\tikzstyle{medium dmap conj}=[draw,doubled,shape=NWbox,inner sep=2pt,minimum height=6mm,fill=white,minimum width=7mm]
\tikzstyle{semilarge dmap}=[draw,doubled,shape=NEbox,inner sep=2pt,minimum height=6mm,fill=white,minimum width=9.5mm]
\tikzstyle{semilarge dmap trans}=[draw,doubled,shape=SWbox,inner sep=2pt,minimum height=6mm,fill=white,minimum width=9.5mm]
\tikzstyle{semilarge dmap adj}=[draw,doubled,shape=SEbox,inner sep=2pt,minimum height=6mm,fill=white,minimum width=9.5mm]
\tikzstyle{semilarge dmap dag}=[draw,doubled,shape=SEbox,inner sep=2pt,minimum height=6mm,fill=white,minimum width=9.5mm]
\tikzstyle{semilarge dmap conj}=[draw,doubled,shape=NWbox,inner sep=2pt,minimum height=6mm,fill=white,minimum width=9.5mm]
\tikzstyle{large dmap}=[draw,doubled,shape=NEbox,inner sep=2pt,minimum height=6mm,fill=white,minimum width=12mm]
\tikzstyle{large dmap conj}=[draw,doubled,shape=NWbox,inner sep=2pt,minimum height=6mm,fill=white,minimum width=12mm]
\tikzstyle{large dmap trans}=[draw,doubled,shape=SWbox,inner sep=2pt,minimum height=6mm,fill=white,minimum width=12mm]
\tikzstyle{large dmap adj}=[draw,doubled,shape=SEbox,inner sep=2pt,minimum height=6mm,fill=white,minimum width=12mm]
\tikzstyle{large dmap dag}=[draw,doubled,shape=SEbox,inner sep=2pt,minimum height=6mm,fill=white,minimum width=12mm]
\tikzstyle{very large dmap}=[draw,doubled,shape=NEbox,inner sep=2pt,minimum height=6mm,fill=white,minimum width=19.5mm]

\tikzstyle{muxbox}=[draw,shape=rectangle,minimum height=3mm,minimum width=3mm,fill=white]
\tikzstyle{dmuxbox}=[muxbox,doubled]

\tikzstyle{box}=[draw,shape=rectangle,inner sep=2pt,minimum height=6mm,minimum width=6mm,fill=white]
\tikzstyle{dbox}=[draw,doubled,shape=rectangle,inner sep=2pt,minimum height=6mm,minimum width=6mm,fill=white]
\tikzstyle{dmap}=[draw,doubled,shape=NEbox,inner sep=2pt,minimum height=6mm,fill=white]
\tikzstyle{dmapdag}=[draw,doubled,shape=SEbox,inner sep=2pt,minimum height=6mm,fill=white]
\tikzstyle{dmapadj}=[draw,doubled,shape=SEbox,inner sep=2pt,minimum height=6mm,fill=white]
\tikzstyle{dmaptrans}=[draw,doubled,shape=SWbox,inner sep=2pt,minimum height=6mm,fill=white]
\tikzstyle{dmapconj}=[draw,doubled,shape=NWbox,inner sep=2pt,minimum height=6mm,fill=white]

\tikzstyle{ddmap}=[draw,doubled,dashed,shape=NEbox,inner sep=2pt,minimum height=6mm,fill=white]
\tikzstyle{ddmapdag}=[draw,doubled,dashed,shape=SEbox,inner sep=2pt,minimum height=6mm,fill=white]
\tikzstyle{ddmapadj}=[draw,doubled,dashed,shape=SEbox,inner sep=2pt,minimum height=6mm,fill=white]
\tikzstyle{ddmaptrans}=[draw,doubled,dashed,shape=SWbox,inner sep=2pt,minimum height=6mm,fill=white]
\tikzstyle{ddmapconj}=[draw,doubled,dashed,shape=NWbox,inner sep=2pt,minimum height=6mm,fill=white]

\boxshape{sNEbox}{0pt}{3pt}
\boxshape{sSEbox}{0pt}{-3pt}
\boxshape{sNWbox}{3pt}{0pt}
\boxshape{sSWbox}{-3pt}{0pt}
\tikzstyle{smap}=[draw,shape=sNEbox,fill=white]
\tikzstyle{smapdag}=[draw,shape=sSEbox,fill=white]
\tikzstyle{smapadj}=[draw,shape=sSEbox,fill=white]
\tikzstyle{smaptrans}=[draw,shape=sSWbox,fill=white]
\tikzstyle{smapconj}=[draw,shape=sNWbox,fill=white]

\tikzstyle{dsmap}=[draw,dashed,shape=sNEbox,fill=white]
\tikzstyle{dsmapdag}=[draw,dashed,shape=sSEbox,fill=white]
\tikzstyle{dsmaptrans}=[draw,dashed,shape=sSWbox,fill=white]
\tikzstyle{dsmapconj}=[draw,dashed,shape=sNWbox,fill=white]

\boxshape{mNEbox}{0pt}{10pt}
\boxshape{mSEbox}{0pt}{-10pt}
\boxshape{mNWbox}{10pt}{0pt}
\boxshape{mSWbox}{-10pt}{0pt}
\tikzstyle{mmap}=[draw,shape=mNEbox]
\tikzstyle{mmapdag}=[draw,shape=mSEbox]
\tikzstyle{mmaptrans}=[draw,shape=mSWbox]
\tikzstyle{mmapconj}=[draw,shape=mNWbox]

\tikzstyle{mmapgray}=[draw,fill=gray!40!white,shape=mNEbox]
\tikzstyle{smapgray}=[draw,fill=gray!40!white,shape=sNEbox]

\makeatletter

\pgfdeclareshape{cornerpoint}{
\inheritsavedanchors[from=rectangle] 
\inheritanchorborder[from=rectangle]
\inheritanchor[from=rectangle]{center}
\inheritanchor[from=rectangle]{north}
\inheritanchor[from=rectangle]{south}
\inheritanchor[from=rectangle]{west}
\inheritanchor[from=rectangle]{east}
\backgroundpath{
\southwest \pgf@xa=\pgf@x \pgf@ya=\pgf@y
\northeast \pgf@xb=\pgf@x \pgf@yb=\pgf@y

\pgfmathsetmacro{\pgf@shorten@left}{\pgfkeysvalueof{/tikz/shorten left}}
\pgfmathsetmacro{\pgf@shorten@right}{\pgfkeysvalueof{/tikz/shorten right}}

\pgfpathmoveto{\pgfpoint{0.5 * (\pgf@xa + \pgf@xb)}{\pgf@ya - 5pt}}
\pgfpathlineto{\pgfpoint{\pgf@xa - 8pt + \pgf@shorten@left}{\pgf@yb - 1.5 * \pgf@shorten@left}}
\pgfpathlineto{\pgfpoint{\pgf@xa - 8pt + \pgf@shorten@left}{\pgf@yb}}
\pgfpathlineto{\pgfpoint{\pgf@xb + 8pt - \pgf@shorten@right}{\pgf@yb}}
\pgfpathlineto{\pgfpoint{\pgf@xb + 8pt - \pgf@shorten@right}{\pgf@yb - 1.5 * \pgf@shorten@right}}
\pgfpathclose
}
}

\pgfdeclareshape{cornercopoint}{
\inheritsavedanchors[from=rectangle] 
\inheritanchorborder[from=rectangle]
\inheritanchor[from=rectangle]{center}
\inheritanchor[from=rectangle]{north}
\inheritanchor[from=rectangle]{south}
\inheritanchor[from=rectangle]{west}
\inheritanchor[from=rectangle]{east}
\backgroundpath{
\southwest \pgf@xa=\pgf@x \pgf@ya=\pgf@y
\northeast \pgf@xb=\pgf@x \pgf@yb=\pgf@y

\pgfmathsetmacro{\pgf@shorten@left}{\pgfkeysvalueof{/tikz/shorten left}}
\pgfmathsetmacro{\pgf@shorten@right}{\pgfkeysvalueof{/tikz/shorten right}}

\pgfpathmoveto{\pgfpoint{0.5 * (\pgf@xa + \pgf@xb)}{\pgf@yb + 5pt}}
\pgfpathlineto{\pgfpoint{\pgf@xa - 8pt + \pgf@shorten@left}{\pgf@ya + 1.5 * \pgf@shorten@left}}
\pgfpathlineto{\pgfpoint{\pgf@xa - 8pt + \pgf@shorten@left}{\pgf@ya}}
\pgfpathlineto{\pgfpoint{\pgf@xb + 8pt - \pgf@shorten@right}{\pgf@ya}}
\pgfpathlineto{\pgfpoint{\pgf@xb + 8pt - \pgf@shorten@right}{\pgf@ya + 1.5 * \pgf@shorten@right}}
\pgfpathclose
}
}

\makeatother

\pgfkeyssetvalue{/tikz/shorten left}{0pt}
\pgfkeyssetvalue{/tikz/shorten right}{0pt}

\tikzstyle{kpoint common}=[draw,fill=white,inner sep=1pt,minimum height=4mm]
\tikzstyle{kpoint sc}=[shape=cornerpoint,kpoint common]
\tikzstyle{kpoint adjoint sc}=[shape=cornercopoint,kpoint common]
\tikzstyle{kpoint}=[shape=cornerpoint,shorten left=5pt,kpoint common]
\tikzstyle{kpoint adjoint}=[shape=cornercopoint,shorten left=5pt,kpoint common]
\tikzstyle{kpoint conjugate}=[shape=cornerpoint,shorten right=5pt,kpoint common]
\tikzstyle{kpoint transpose}=[shape=cornercopoint,shorten right=5pt,kpoint common]
\tikzstyle{kpoint symm}=[shape=cornerpoint,shorten left=5pt,shorten right=5pt,kpoint common]

\tikzstyle{wide kpoint sc}=[shape=cornerpoint,kpoint common, minimum width=1 cm]
\tikzstyle{wide kpointdag sc}=[shape=cornercopoint,kpoint common, minimum width=1 cm]

\tikzstyle{black kpoint}=[shape=cornerpoint,shorten left=5pt,kpoint common,fill=black,font=\color{white}]

\tikzstyle{black kpoint sm}=[shape=cornerpoint,shorten left=5pt,kpoint common,fill=black,font=\color{white},scale=0.75]

\tikzstyle{black kpoint adjoint}=[shape=cornercopoint,shorten left=5pt,kpoint common,fill=black,font=\color{white}]
\tikzstyle{black kpointadj}=[shape=cornercopoint,shorten left=5pt,kpoint common,fill=black,font=\color{white}]

\tikzstyle{black kpointadj sm}=[shape=cornercopoint,shorten left=5pt,kpoint common,fill=black,font=\color{white},scale=0.75]

\tikzstyle{black dkpoint}=[shape=cornerpoint,shorten left=5pt,kpoint common,fill=black, doubled,font=\color{white}]
\tikzstyle{black dkpoint adjoint}=[shape=cornercopoint,shorten left=5pt,kpoint common,fill=black, doubled,font=\color{white}]
\tikzstyle{black dkpointadj}=[shape=cornercopoint,shorten left=5pt,kpoint common,fill=black, doubled,font=\color{white}]

\tikzstyle{black dkpoint sm}=[shape=cornerpoint,shorten left=5pt,kpoint common,fill=black, doubled,font=\color{white},scale=0.75]
\tikzstyle{black dkpointadj sm}=[shape=cornercopoint,shorten left=5pt,kpoint common,fill=black, doubled,font=\color{white},scale=0.75]

\tikzstyle{kpointdag}=[kpoint adjoint]
\tikzstyle{kpointadj}=[kpoint adjoint]
\tikzstyle{kpointconj}=[kpoint conjugate]
\tikzstyle{kpointtrans}=[kpoint transpose]

\tikzstyle{big kpoint}=[kpoint, minimum width=1.2 cm, minimum height=8mm, inner sep=4pt, text depth=3mm]

\tikzstyle{wide kpoint}=[kpoint, minimum width=1 cm, inner sep=2pt]
\tikzstyle{wide kpointdag}=[kpointdag, minimum width=1 cm, inner sep=2pt]
\tikzstyle{wide kpointconj}=[kpointconj, minimum width=1 cm, inner sep=2pt]
\tikzstyle{wide kpointtrans}=[kpointtrans, minimum width=1 cm, inner sep=2pt]

\tikzstyle{wider kpoint}=[kpoint, minimum width=1.25 cm, inner sep=2pt]
\tikzstyle{wider kpointdag}=[kpointdag, minimum width=1.25 cm, inner sep=2pt]
\tikzstyle{wider kpointconj}=[kpointconj, minimum width=1.25 cm, inner sep=2pt]
\tikzstyle{wider kpointtrans}=[kpointtrans, minimum width=1.25 cm, inner sep=2pt]

\tikzstyle{gray kpoint}=[kpoint,fill=gray!50!white]
\tikzstyle{gray kpointdag}=[kpointdag,fill=gray!50!white]
\tikzstyle{gray kpointadj}=[kpointadj,fill=gray!50!white]
\tikzstyle{gray kpointconj}=[kpointconj,fill=gray!50!white]
\tikzstyle{gray kpointtrans}=[kpointtrans,fill=gray!50!white]

\tikzstyle{gray dkpoint}=[kpoint,fill=gray!50!white,doubled]
\tikzstyle{gray dkpointdag}=[kpointdag,fill=gray!50!white,doubled]
\tikzstyle{gray dkpointadj}=[kpointadj,fill=gray!50!white,doubled]
\tikzstyle{gray dkpointconj}=[kpointconj,fill=gray!50!white,doubled]
\tikzstyle{gray dkpointtrans}=[kpointtrans,fill=gray!50!white,doubled]

\tikzstyle{white label}=[draw,fill=white,rectangle,inner sep=0.7 mm]
\tikzstyle{gray label}=[draw,fill=gray!50!white,rectangle,inner sep=0.7 mm]
\tikzstyle{black label}=[draw,fill=black,rectangle,inner sep=0.7 mm]

\tikzstyle{dkpoint}=[kpoint,doubled]
\tikzstyle{wide dkpoint}=[wide kpoint,doubled]
\tikzstyle{dkpointdag}=[kpoint adjoint,doubled]
\tikzstyle{wide dkpointdag}=[wide kpointdag,doubled]
\tikzstyle{dkcopoint}=[kpoint adjoint,doubled]
\tikzstyle{dkpointadj}=[kpoint adjoint,doubled]
\tikzstyle{dkpointconj}=[kpoint conjugate,doubled]
\tikzstyle{dkpointtrans}=[kpoint transpose,doubled]

\tikzstyle{kscalar}=[kpoint common, shape=EBox, inner xsep=-1pt, inner ysep=3pt,font=\small]
\tikzstyle{kscalarconj}=[kpoint common, shape=WBox, inner xsep=-1pt, inner ysep=3pt,font=\small]

\tikzstyle{spekpoint}=[kpoint sc,minimum height=5mm,inner sep=3pt]
\tikzstyle{spekcopoint}=[kpoint adjoint sc,minimum height=5mm,inner sep=3pt]

\tikzstyle{dspekpoint}=[spekpoint,doubled]
\tikzstyle{dspekcopoint}=[spekcopoint,doubled]

 \tikzstyle{upground}=[circuit ee IEC,thick,ground,rotate=90,scale=2.5]
 \tikzstyle{downground}=[circuit ee IEC,thick,ground,rotate=-90,scale=2.5]
 \tikzstyle{bigground}=[regular polygon,regular polygon sides=3,draw=gray,scale=0.50,inner sep=-0.5pt,minimum width=10mm,fill=gray]

\tikzstyle{arrs}=[-latex,font=\small,auto]
\tikzstyle{arrow plain}=[arrs]
\tikzstyle{arrow dashed}=[dashed,arrs]
\tikzstyle{arrow bold}=[very thick,arrs]
\tikzstyle{arrow hide}=[draw=white!0,-]
\tikzstyle{arrow reverse}=[latex-]
\tikzstyle{cdnode}=[]

\DeclareMathOperator{\tr}{tr}

\let\olddagger\dagger
\renewcommand{\dagger}{\ensuremath{\olddagger}\xspace}

\theoremstyle{definition}
\newtheorem{theorem}{Theorem}[section]

\newtheorem{lemma}[]{Lemma}
\newtheorem{proposition}[theorem]{Proposition}

\newtheorem{example*}[theorem]{Example*}
\newtheorem{examples*}[theorem]{Examples*}

\newtheorem{remark*}[theorem]{Remark*}

\theoremstyle{plain}

\usepackage{color}
\def\bR{\begin{color}{red}}
\def\bB{\begin{color}{blue}}
\def\bM{\begin{color}{magenta}}
\def\bC{\begin{color}{cyan}}
\def\bW{\begin{color}{white}}
\def\bBl{\begin{color}{black}}
\def\bG{\begin{color}{green}}
\def\bY{\begin{color}{yellow}}
\def\e{\end{color}\xspace}
\newcommand{\bit}{\begin{itemize}}
\newcommand{\eit}{\end{itemize}\par\noindent}
\newcommand{\ben}{\begin{enumerate}}
\newcommand{\een}{\end{enumerate}\par\noindent}
\newcommand{\beq}{\begin{equation}}
\newcommand{\eeq}{\end{equation}\par\noindent}
\newcommand{\beqa}{\begin{eqnarray*}}
\newcommand{\eeqa}{\end{eqnarray*}\par\noindent}
\newcommand{\beqn}{\begin{eqnarray}}
\newcommand{\eeqn}{\end{eqnarray}\par\noindent}

\def\jR{\begin{color}{DarkRed}}
\def\jB{\begin{color}{MidnightBlue}}
\def\jM{\begin{color}{magenta}}
\def\jC{\begin{color}{cyan}}
\def\jW{\begin{color}{white}}
\def\jBl{\begin{color}{black}}
\def\jG{\begin{color}{green}}
\def\jY{\begin{color}{yellow}}

\def\cR{\begin{color}{Crimson}}
\def\cB{\begin{color}{cyan}}
\def\cM{\begin{color}{magenta}}
\def\cC{\begin{color}{cyan}}
\def\cW{\begin{color}{white}}
\def\cBl{\begin{color}{black}}
\def\cG{\begin{color}{green}}
\def\cY{\begin{color}{yellow}}

\usepackage{bm}
\usepackage{upgreek}
\usepackage[linktocpage=true, colorlinks=true, linkcolor=blue, urlcolor=blue, citecolor=blue]{hyperref}
\usepackage{paralist}
\usepackage{soul}
\usepackage{changes}

\newtheorem{defn}[theorem]{Definition}

\newtheorem{thm}[theorem]{Theorem}

\definecolor{darkgreen}{rgb}{0,.5,0}
\definecolor{darkblue}{rgb}{0,0,0.5}
\definecolor{darkred}{rgb}{0.5,0,0}

\usepackage{xr}

\allowdisplaybreaks

\begin{document}

\title{\bel{Post-quantum steering is a stronger-than-quantum resource for information processing}}

\author{Paulo J. Cavalcanti}
\email{paulojcvf@gmail.com}
\affiliation{International Centre for Theory of Quantum Technologies, University of Gda\'nsk, 80-309 Gda\'nsk, Poland}
\author{John H.~Selby}
\email{john.h.selby@gmail.com}
\affiliation{International Centre for Theory of Quantum Technologies, University of Gda\'nsk, 80-309 Gda\'nsk, Poland}
\author{Jamie Sikora}
\email{jamiesikora@gmail.com} 
\affiliation{Virginia Polytechnic Institute and State University, Blacksburg, VA 24061, USA}
\author{Thomas D.~Galley}
\email{tgalley1@perimeterinstitute.ca}
\affiliation{Perimeter Institute for Theoretical Physics, 31 Caroline St. N, Waterloo, Ontario, N2L 2Y5, Canada}
\author{Ana Bel\'en Sainz}
\email{sainz.ab@gmail.com}
\affiliation{International Centre for Theory of Quantum Technologies, University of Gda\'nsk, 80-309 Gda\'nsk, Poland}

\date{\today} 

\begin{abstract}
\bel{We present the first instance where post-quantum steering is a stronger-than-quantum resource for information processing -- remote state preparation. In addition, we show that the phenomenon of post-quantum steering is not just a mere mathematical curiosity allowed by the no-signalling principle, but it may arise within compositional theories beyond quantum theory, hence making its study fundamentally relevant. We show these results by formulating a new compositional general probabilistic theory -- which we call Witworld -- with strong post-quantum features, which proves to be a intuitive and useful tool for exploring steering and its applications beyond the quantum realm.} 
\end{abstract}

\maketitle

\section{Introduction}

A striking property of nature is that it is non-classical. Entanglement \cite{EPR35,schrodinger1935gegenwartige}, Bell nonlocality \cite{Bell64}, and steering \cite{schrodinger36,SteeWiseman07,SteeWiseman207} are examples of quantum phenomena that can be observed experimentally \cite{BellExp1,BellExp2,BellExp3,SteeExp1,SteeExp2,SteeExp3} and which cannot be explained by classical physics. Besides their foundational relevance, with the advent of quantum information theory we learned a valuable lesson: these seemingly bizarre quantum features can be exploited to process information more efficiently, even in ways that could never be possible with classical resources alone \cite{BellRev,SteeRevSDP,SteeRev}. 

\bel{A ubiquitous framework in which the scope of quantum advantage in information processing is explored is the so-called {device-independent {framework}}, where the parties executing protocols only rely on the classical inputs and outputs with which they operate their shared (and possibly quantum) devices. Such {a} framework is particularly well suited to the necessarily paranoid perspective on cryptographic tasks \cite{BellRev}, and is almost ubiquitously underpinned by a Bell nonlocality setup. A special milestone in the research of non-classical resources for device-independent information processing was the realisation that there exist correlations beyond what is quantumly admissible (a.k.a., {post-quantum} {correlations}), but which nonetheless are consistent with special relativity \cite{popescu1994quantum}. These so-called no-signalling correlations, moreover, were shown to be consistent with alternative theories of nature, confirming the necessity of their study. Exploring these general no-signalling correlations enabled, for example, the design of quantum cryptographic protocols that are robust against powerful adversaries that are not bounded by the laws of quantum theory \cite{QKD1,QKD2}, and also the formulation of physical principles that a quantum world must satisfy \cite{popescu1994quantum,brassard2006limit,pawlowski2009information,navascues2010glance,LO,LO2,afls,MNC}. These post-quantum correlations are hence studied beyond philosophical {motivations}, and from the perspective of the resources they provide for operational tasks.}

\bel{Device-independent frameworks for information processing, however, are substantially demanding to implement experimentally. {Indeed}, for practical purposes, even if cryptographically secure, device-independent protocols are yet to move beyond `proof of principle' applications into scalable and easily-accessible technologies. There are situations, however, where one may argue that the quantum description of some of the parties involved can be leveraged  in the protocols: in the simplest case where two parties are involved {and a single party is assigned a quantum description} this is usually referred to as a one-way device-independent framework \cite{SteeRev,SteeRevSDP}. In such scenarios, the non-classical phenomenon providing quantum advantage is steering rather than Bell nonlocality. Recently, it has been shown that steering beyond that {which} quantum theory allows, whilst still consistent with special relativity, may exist \cite{pqs,bpqs}, which opens a new plethora of questions, such as (i) can post-quantum steering provide an advantage beyond what is possible with quantum theory for some information processing task?; (ii) is post-quantum steering just a mathematical curiosity, or may it emerge within alternative physical theories?}

In this work we tackle those two questions. First, we show that there are alternative theories beyond quantum which feature post-quantum steering, making the phenomenon physically relevant for post-quantum information processing and motivating its exploration. Second, we find a task for which post-quantum steering is a stronger-than-quantum resource: remote state preparation. Remote state preparation (RSP) is a task similar in spirit to teleportation: the goal is to transmit quantum states from one party to another distant party using only shared entanglement and classical communication. Unlike teleportation, though, in RSP the sender has knowledge of the transmitted state, which makes RSP protocols more economical than teleportation in terms of resources (e.g., classical communication) needed to succeed at the task \cite{RSPlo,RSPdebbie}. In addition, RSP protocols do not necessarily require the ability to experimentally implement Bell (entangling) measurements, which makes them potentially {more feasible \bel{experimentally}  \cite{RSPexp}. The kind of RSP protocols that we focus on are so-called {oblivious} -- \bel{namely} those where no information about the state is leaked to the receiver, apart from the state itself, something that is relevant for certain applications such as blind quantum computation \cite{childs2005blind}. RSP is indeed an insightful task to explore from both a fundamental and applied viewpoint.} 

\bel{In order to prove  our results we define a generalised probabilistic theory (GPT) that we name {Witworld}, given its strong connection to entanglement witnesses. Witworld combines system types of three well-known GPTs (classical, quantum, and Boxworld) in a simple mathematical way, via the so-called max tensor product. Remarkably, even though Witworld cannot reproduce all the phenomenology of quantum theory, it does realise all quantum predictions for Bell and steering experiments. Hence, we can learn about the limitations of quantum advantage in one-sided and fully device-independent protocols by exploring the performance of Witworld. Despite its simplicity, Witworld displays powerful post-quantum features: not only can it realise all non-signalling correlations in Bell experiments, but it also displays post-quantum steering. \bel{ As Witworld is a fully compositional theory, it comes equipped with an intuitive diagrammatic calculus \cite{chiribella2010probabilistic,hardy2011reformulating}. This provides a convenient toolkit} for exploring other applications of post-quantum phenomena for information processing. }

The paper is organised as follows. \bel{In what remains of this section we present a brief introduction to the three main topics of this paper: generalised probabilistic theories, Bell nonlocality, and steering.} Section \ref{se:ww2} presents the definition of Witworld, assuming basic knowledge of GPTs --  the reader who is not familiar with them may consult Sec.~A of the Supplemental Material. Section \ref{se:pqps} discusses the post-quantum properties of Witworld with a focus on Bell and steering experiments, while section \ref{se:pqarsp} discusses how Witworld outperforms quantum theory at certain information processing tasks.  Section \ref{se:pqarsp} also presents the first task where post-quantum steering outperforms quantum steering as a resource. Technicalities, as well as a brief review on steering, are included in the appendices.

\noindent \textbf{Generalised probabilistic theories .--}
The framework of {generalised probabilistic theories} \cite{barrett2007gpt,hardy2001quantum} provides tools with which to explore the operational features of candidate theories in a unified fashion. Classical theory, as well as quantum theory, may be recast within the language of GPTs \cite{hardy2001quantum,barrett2007gpt}, which enables their unified and comparative study. 
The GPT framework has been proven useful not only from a foundational perspective (e.g., for developing axiomatic reconstructions of quantum theory  \cite{hardy2001quantum,
dakic2009quantum,chiribella2011informational,
hardy2011reformulating,
barnum2014higher}), but also when exploring the quantum information capabilities of post-quantum theories, such as their computational power \cite{barrett2019computational,lee2016deriving,lee2016generalised,barnum2018oracles,lee2015computation,krumm2019quantum,garner2018interferometric} or cryptographic security \cite{barrett2007gpt,barnum2011information,lami2018ultimate,sikora2020impossibility,selby2018make,sikora2020impossibility}.

\noindent \textbf{Bell nonlocality.--}
One example of a non-classical phenomenon of foundational and applied relevance is Bell nonlocality. Bell experiments are ubiquitous in the fields of quantum foundations and quantum information processing. On the one hand, Bell's Theorem \cite{Bell64} established a precise sense in which quantum theory requires a departure from a classical worldview, and violations of Bell inequalities provide a means for certifying the nonclassicality of nature. On the other hand, the correlations observed in a Bell test have become a resource for certain tasks \cite{BellRev}, and \bel{the} violation of so-called Bell inequalities \bel{by these correlations} \bel{has} become \bel{a} standard certification tool for security in cryptographic protocols \cite{QKD1,QKD2,QKD3}. In brief, a Bell scenario consists of a set of distant parties that perform space-like separated actions on their \bel{part} of a physical system, and the objects of study are the correlations they observe among their measurement outcomes. In the case of a bipartite scenario,  let $x \in \X$ and $a \in \A$ denote the classical variables that label the measurement choices and produced outcomes, respectively, corresponding to the first party (hereon, Alice), and, respectively,  $y \in \Y$ and $b \in \Bo$ those for the second party (hereon, Bob). The correlations observed in this bipartite Bell experiment are captured by the conditional probability distribution $\{p(ab|xy)\}$. It is therefore natural to ask ourselves which possible $\{p(ab|xy)\}$ may be \bel{generated}, and at what cost. 
Given the space-like separation constraints, the largest set of correlations observable in a Bell scenario corresponds to those that satisfy the No-Signalling Principle, and it is known that correlations allowed by quantum theory \bel{are} a strict subset of those correlations. Notably, a GPT colloquially referred to as Boxworld \bel{\cite{Short_2010,gross2010all}} has been defined \cite{barrett2007gpt}, which can realise all the correlations compatible with the no-signalling principle via its bipartite states and local measurements.

\noindent \textbf{Steering.--}
Steering is another non-classical phenomenon of foundational and practical relevance, which was identified back in the 1930s \cite{schrodinger36} but, unlike Bell nonlocality, only recently caught the attention of the quantum information community \cite{SteeWiseman07,SteeWiseman207}. Steering captures the idea that Alice seemingly remotely `steers' the state of a distant Bob, in a way which has no classical explanation. A main feature of a steering experiment is the asymmetric role that the parties play, which makes it particularly suitable as a resource for certain asymmetric information processing tasks \cite{SteeRevSDP,SteeRev}. In brief,  the simplest steering experiment consist of two distant parties -- Alice and Bob -- which perform local actions on their \bel{part} of a physical system. Unlike in a Bell experiment, though, the parties here perform different types of transformations in their labs: Alice performs a measurement, labelled by $x \in \X$, on her system, and obtains a classical outcome $a \in \A$, whereas Bob performs full tomography of the quantum system and so describes it via a density matrix $\rho^\mathrm{B}_{a|x}$ that is effectively prepared in his lab after Alice's actions. In this way, the object of study in these experiments are the ensembles of ensembles (a.k.a.~\textit{assemblages} \cite{pusey2013negativity}) given by $\{\{\sigma_{a|x}\}_{a\in \A}\}_{x \in \X}$, where $\tr{(\sigma_{a|x})}{} = p(a|x)$ and $\sigma_{a|x} = p(a|x) \, \rho^\mathrm{B}_{a|x}$. While nonclassical properties of steering within quantum theory have been considerably explored, not much is known about steering beyond quantum theory \cite{pqs,pqsc,pqso,bpqs}. One main obstacle for this is the complexity of capturing fundamentally what could be post-quantum about an assemblage of quantum states. An operational recast of the steering phenomenon has been recently put forward \cite{pqs,bpqs}, which facilitates a way to articulate the concept of post-quantum assemblages. The study of post-quantum steering has only just begun, and, unlike for Bell nonlocality, \bel{important} fundamental and practical questions are yet to be answered. \bel{One such question is:} does there exist a GPT that realises all these post-quantum assemblages?

\section{Results}
\subsection{Witworld}\label{se:ww2}

In this section we provide a simple and concise introduction to Witworld, which should enable the understanding of the subsequent results. 
We moreover provide a detailed formal definition in Sec.~B of the Supplemental Material.

In Witworld, there are three types of basic  systems, which can be composed to construct more general system types.
The basic systems are classical systems,  quantum systems,
and Boxworld systems \cite{barrett2007gpt}.
(One could easily modify the theory to allow for further system types. However, it is not clear that this will provide any further benefit to the study of steering).
Systems that are of one of those three types are called \emph{atomic}. Witworld features a composition rule (which we define shortly) by which these simple system types can form new ones that are neither classical, quantum, nor Boxworld. We denote the \bel{atomic} types diagrammatically with different types of wires by:
\begin{equation}\label{eq:thewireswit}
	\begin{tikzpicture}
	\begin{pgfonlayer}{nodelayer}
		\node [style=none] (0) at (0, 0.75) {};
		\node [style=none] (1) at (0, -0.75) {};
		\node [style={right label}] (2) at (0, -0.5) {$\C{v}$};
	\end{pgfonlayer}
	\begin{pgfonlayer}{edgelayer}
		\draw [c] (0.center) to (1.center);
	\end{pgfonlayer}
\end{tikzpicture},\quad \begin{tikzpicture}
	\begin{pgfonlayer}{nodelayer}
		\node [style=none] (0) at (0, 0.75) {};
		\node [style=none] (1) at (0, -0.75) {};
		\node [style={right label}] (2) at (0, -0.5) {$\Q{d}$};
	\end{pgfonlayer}
	\begin{pgfonlayer}{edgelayer}
		\draw [q] (0.center) to (1.center);
	\end{pgfonlayer}
\end{tikzpicture}, \quad \begin{tikzpicture}
	\begin{pgfonlayer}{nodelayer}
		\node [style=none] (0) at (0, 0.75) {};
		\node [style=none] (1) at (0, -0.75) {};
		\node [style={right label}] (2) at (0, -0.5) {$\B{n,k}$};
	\end{pgfonlayer}
	\begin{pgfonlayer}{edgelayer}
		\draw [b=12] (0.center) to (1.center);
	\end{pgfonlayer}
\end{tikzpicture}
,
\end{equation}
where $ \mathcal{C}_{v} $ denotes a classical system of dimension $v$, $ \mathcal{Q}_{d}$ denotes a quantum system of dimension $d$, and $ \mathcal{B}_{n,k}$ denotes a Boxworld system of dimension $(n,k)$ (These two integers relate to the input/output cardinality of the correlations in Bell scenarios that the system is tailored at \cite{barrett2007gpt}).
Moreover, when we need to use a generic system type (which can be either simple or composite), we denote this by 
\begin{equation}\label{}
	\begin{tikzpicture}
	\begin{pgfonlayer}{nodelayer}
		\node [style=none] (0) at (0, 0.75) {};
		\node [style=none] (1) at (0, -0.75) {};
		\node [style={right label}] (2) at (0, -0.5) {$S$};
	\end{pgfonlayer}
	\begin{pgfonlayer}{edgelayer}
		\draw [g] (0.center) to (1.center);
	\end{pgfonlayer}
\end{tikzpicture}.
\end{equation}

We can also explicitly denote the components of a composite system by using parallel wires, for example:
\beq
	\begin{tikzpicture}
	\begin{pgfonlayer}{nodelayer}
		\node [style=none] (0) at (-1, 0.75) {};
		\node [style=none] (1) at (-1, -0.75) {};
		\node [style=none] (2) at (0.25, -0.75) {};
		\node [style=none] (3) at (0.25, 0.75) {};
		\node [style=none] (4) at (1.5, -0.75) {};
		\node [style=none] (5) at (1.5, 0.75) {};
		\node [style=right label] (6) at (-1, -0.25) {$\Q{2}$};
		\node [style=right label] (7) at (0.25, -0.25) {$\B{2,2}$};
		\node [style=right label] (8) at (1.5, -0.25) {$\Q{3}$};
		\node [style=none] (9) at (2.75, 0.75) {};
		\node [style=none] (10) at (2.75, -0.75) {};
		\node [style=right label] (11) at (2.75, -0.25) {$\C{5}$};
	\end{pgfonlayer}
	\begin{pgfonlayer}{edgelayer}
		\draw [q] (0.center) to (1.center);
		\draw [b=12] (2.center) to (3.center);
		\draw [q] (4.center) to (5.center);
		\draw [c] (10.center) to (9.center);
	\end{pgfonlayer}
\end{tikzpicture}
\eeq
corresponds to a system composed of a qubit, a $(2,2)$ Boxworld system, a qutrit, and a classical system of dimension $5$.

The state space of a given system type is represented by a convex set living inside some real vector space.
For instance, an atomic quantum system $ \mathcal{Q}_{2}$ has states living inside a Bloch sphere in a 4-dimensional real vector space, an atomic classical state $ \mathcal{C}_{3}$ has states living inside a triangle in a 3-dimensional vector space, and an atomic Boxworld system $ \mathcal{B}_{2,2}$ has states inside a square in a 3-dimensional vector space. These examples are depicted in Figure~\ref{geometry}. Diagrammatically we denote a state $\sigma$ of a system $S$ by
\beq
\begin{tikzpicture}
	\begin{pgfonlayer}{nodelayer}
		\node [style=none] (0) at (-1, 0.75) {};
		\node [style=point] (1) at (-1, -0.5) {$\sigma$};
		\node [style=right label] (6) at (-1, 0.25) {$S$};
	\end{pgfonlayer}
	\begin{pgfonlayer}{edgelayer}
		\draw [g] (0.center) to (1);
	\end{pgfonlayer}
\end{tikzpicture}
\ .
\eeq

Regarding the effects, Witworld includes all the elements of the dual of the vector space which evaluate to valid probabilities for every state.
That is, Witworld satisfies the no-restriction hypothesis (NRH) \cite{chiribella2010probabilistic}.
For example, for a system of the type $ \mathcal{Q}_{2}$, the effects correspond to POVM elements and are represented as
a particular region of ${(\ReS^4)}^*$.
This region can be defined as the intersection of the cone of linear functionals which evaluate to positive reals on the set of state vectors with the set of linear functionals which evaluate to less than $1$ for all state vectors.
For $ \mathcal{C}_{3}$, the effects live in a cube, and for $ \mathcal{B}_{2,2}$, the effects live in an octahedron.
 A pictorial representation of these can be seen in Figure \ref{geometry}. An effect $e$ for a system $S$ is diagrammatically denoted by
\beq
\begin{tikzpicture}
	\begin{pgfonlayer}{nodelayer}
		\node [style=none] (0) at (-1, -0.75) {};
		\node [style=copoint] (1) at (-1, 0.5) {$e$};
		\node [style=right label] (6) at (-1, -0.25) {$S$};
	\end{pgfonlayer}
	\begin{pgfonlayer}{edgelayer}
		\draw [g] (0.center) to (1);
	\end{pgfonlayer}
\end{tikzpicture}
\ .
\eeq
Since effects belong to the dual vector space, when we compose them with a state we obtain a real number, which, by assumption, must give a valid probability. That is, for all states $\sigma^S$ and effects $e^S$ we have $e^S(\sigma^S) \in [0,1]$. Diagrammatically this is written as
\beq
\begin{tikzpicture}
	\begin{pgfonlayer}{nodelayer}
		\node [style=copoint] (0) at (-1, 0.75) {$e$};
		\node [style=point] (1) at (-1, -0.75) {$\sigma$};
		\node [style=right label] (6) at (-1, 0) {$S$};
	\end{pgfonlayer}
	\begin{pgfonlayer}{edgelayer}
		\draw [g] (0) to (1);
	\end{pgfonlayer}
\end{tikzpicture}
 \ \ \in \ \ [0,1]\ .
\eeq


As mentioned previously, we define Witworld to satisfy the NRH.
This, however, is not the only simplifying assumption that we make in this construction.
Additionally, we define the composition of systems to be via the so-called max tensor product \cite{barnum2011information} and, hence, that the theory is locally tomographic \cite{hardy2001quantum}. Moreover, we demand that Witworld
satisfies the generalised no-restriction hypothesis (GNRH).
Intuitively, the GNRH is the NRH together with the requirement that every transformation that takes every element of a valid state space to an element of another valid state space is a valid transformation in the theory, that is, every completely positive transformation is considered valid.

The max tensor product \bel{(see Def.~A.8 in the Supplemental Material)}  is a composition  rule that assigns as valid states of a composite system $A\cdot B$ any vector in the product vector space ($ V^{A\cdot B} = V^{A} \otimes V^{B}$) that is consistent with the separable effects. Formally, $ \rho^{AB}$ is a valid state of the composite system AB if and only if, for every effect $ e^{A}$ of A and $ {e'}^{B}$ of B, $ e^{A} \otimes {e'}^{B} ( \rho^{AB} ) \in [0,1]$. 
Diagrammatically we express this condition as:
\beq
\forall e, e' \ \ \begin{tikzpicture}
	\begin{pgfonlayer}{nodelayer}
		\node [style=none] (0) at (-1, -0.75) {};
		\node [style=copoint] (1) at (-1, 0.5) {$e$};
		\node [style=none] (2) at (1, -0.75) {};
		\node [style=copoint] (3) at (1, 0.5) {$e'$};
		\node [style=none] (4) at (-1.75, -0.75) {};
		\node [style=none] (5) at (0, -2) {};
		\node [style=none] (6) at (1.75, -0.75) {};
		\node [style=none] (7) at (0, -1.25) {$\rho$};
		\node [style=none] (8) at (-1, -0.25) {};
		\node [style=right label] (9) at (-1, -0.25) {$A$};
		\node [style=right label] (10) at (1, -0.25) {$B$};
	\end{pgfonlayer}
	\begin{pgfonlayer}{edgelayer}
		\draw [g] (0.center) to (1);
		\draw [g] (2.center) to (3);
		\draw (4.center) to (5.center);
		\draw (6.center) to (5.center);
		\draw (4.center) to (6.center);
	\end{pgfonlayer}
\end{tikzpicture}
 \ \in [0,1]\,.
\eeq
\bel{From an intuitive point of view, the max tensor product gives rise to a GPT that somehow maximises the set of states that the system can be prepared in, whilst strongly restricting the set of measurements that one may perform on it. As a matter of fact, even though Witworld might appear to be a more general theory than quantum theory, these two are actually incomparable: quantum theory allows for measurements that Witworld systems cannot be acted upon with (with the latter having a more restricted set of measurements on collections of quantum atomic systems types), whilst Witworld allows for more states on which the composition of quantum atomic system types can be prepared (Witworld allows for two qubits to be prepared on a state mathematically corresponding to a quantum entanglement witness, whereas in quantum theory this is not an allowed state of a two-qubit system).}

The fact that we have defined composition via the max tensor product and are demanding that the theory satisfies the GNRH, means that when we define the atomic states, we define the whole theory\bel{, since} from the atomic states and max tensor product every possible state is defined, and from the states and GNRH the effects and transformations are also determined.

Finally, as mentioned above, the max tensor product is tomographically local, that is to say that 
its states can be uniquely determined by the information obtained from performing local measurements on its parts. Using the example above, this means that $ \rho^{AB}$ is completely determined by a set of values $e_{i}^{A}\otimes e_{j}^{B} ( \rho^{AB})$.

At this point it is worth mentioning some further 
consequences of our definitions.
The first one is that the use of the max tensor product to compose systems implies that every effect in Witworld is separable (see Lem.~A.14).  
Therefore, an important \bel{feature of} Witworld is that it does not contain entangling effects.

A second important fact about Witworld is that the combination of two atomic quantum systems yields systems whose states spaces are larger than \bel{the joint state space obtained from the standard quantum composition rule} \bel{(see Thms.~B.5 and B.6 in the Supplemental Material)}. For example, in the bipartite case, we have that the composite of $ \mathcal{Q}_{n}$ and $ \mathcal{Q}_{m}$, denoted by $ \mathcal{Q}_{n} \cdot \mathcal{Q}_{m}$, has as its state space the set of entanglement witnesses (including density matrices \cite{toniOform})
for the quantum bipartite states, which strictly contains the set of bipartite quantum states $\Q{nm}$. Therefore, whilst Witworld does contain arbitrary quantum systems, quantum theory is not a compositional subtheory within Witworld.
Note that if one were to allow quantum systems in Witworld (where these quantum systems do not have the additional dynamics of Witworld quantum systems) to be composable both according to the standard quantum rule as well as to the Witworld composition rule, one could construct a protocol giving negative probabilities~\cite{barnum2005influence}. As such one cannot extend Witworld in such a manner as to contain quantum theory as a compositional subtheory.

A third important feature is that Boxworld-type (resp.~classical-type) systems in Witworld compose exactly as they do in Boxworld (resp.~classical theory). Therefore, both Boxworld and classical theory are indeed full  subtheories of Witworld. 
Here, by \textit{full subtheory} we mean that you can recover Boxworld or classical theory from Witworld by suitably restricting it to a particular collection of system types. This restriction recovers all and only the states, effects, transformations, and the composition rules of Boxworld or classical theory.

Finally, another important feature of Witworld is that because of the combination of GNRH and max tensor product, there is no difference between positive and completely positive maps (see Lem.~A.15). Of course, for systems that are not quantum, a more general notion (relative to that of positive operators in quantum theory) of positivity must be used in order to make that statement (see Def.~A.6 in the Supplemental Material). Now, in the case of atomic quantum systems, this means that the valid Witworld transformations correspond to positive, but not necessarily completely positive, quantum transformations.
 Hence,
this means that in Witworld there are more transformations available to local agents (i.e., to Alice and Bob) than would be available in quantum theory.

\subsection{Post-quantum phenomena: Bell non-classicality and steering}\label{se:pqps}

In this section we explore the non-classical features that Witworld displays, starting with the case of Bell scenarios. One can readily see that Witworld can realise all non-signalling correlations in arbitrary Bell scenarios (see Fig.~\ref{fig:NSpol}), since Boxworld is a full subtheory of Witworld. Therefore, one can leverage the Boxworld realisations of any non-signalling correlation, and translate them straightforwardly to a realisation within Witworld. For example, take the case of Popescu-Rohrlich (PR) correlations, which read $p_{PR}(ab|xy) = \frac{1}{2} \delta_{a\oplus b = xy}$ with $a,b,x,y \in \{0,1\}$ \bel{and $\oplus$ denoting modulo-2 addition;} these correlations can be realised in Witworld as follows: 
\begin{equation}\label{eq:PRbox}
p_{\mathrm{PR}}(ab|xy) = \begin{tikzpicture}
	\begin{pgfonlayer}{nodelayer}
		\node [style=none] (0) at (-1.25, -1.5) {};
		\node [style=none] (1) at (1.25, -1.5) {};
		\node [style=none] (2) at (0, -2.75) {};
		\node [style=none] (3) at (0, -2) {$s_{\text{PR}}$};
		\node [style=none] (4) at (-0.5, -1.5) {};
		\node [style=none] (5) at (0.5, -1.5) {};
		\node [style=none] (14) at (0.5, -1.5) {};
		\node [style=none] (15) at (1.5, 0.25) {};
		\node [style=none] (16) at (3, 0.25) {};
		\node [style=none] (17) at (1, 0.25) {};
		\node [style=none] (18) at (1.5, 1.25) {};
		\node [style=none] (19) at (3, 1.25) {};
		\node [style=none] (20) at (2.25, 0.75) {$M_{\mathrm{PR}}^\prime$};
		\node [style=none] (21) at (2.25, 0.25) {};
		\node [style=none] (22) at (2.25, 1.25) {};
		\node [style=copoint] (23) at (2.5, 2.25) {$b$};
		\node [style=none] (24) at (0.5, -1.5) {};
		\node [style=point] (25) at (3.25, -1) {$y$};
		\node [style=none] (26) at (-0.5, -1.5) {};
		\node [style=none] (27) at (-1.5, 0.25) {};
		\node [style=none] (28) at (-3, 0.25) {};
		\node [style=none] (29) at (-1, 0.25) {};
		\node [style=none] (30) at (-1.5, 1.25) {};
		\node [style=none] (31) at (-3, 1.25) {};
		\node [style=none] (32) at (-2.25, 0.75) {$M_{\mathrm{PR}}$};
		\node [style=none] (33) at (-2.25, 0.25) {};
		\node [style=none] (34) at (-2.25, 1.25) {};
		\node [style=copoint] (35) at (-2.5, 2.25) {$a$};
		\node [style=none] (36) at (-0.5, -1.5) {};
		\node [style=point] (37) at (-3.25, -1) {$x$};
		\node [style=right label] (38) at (-3, -0.5) {$\mathds{X}$};
		\node [style=right label] (39) at (3.25, -0.5) {$\mathds{Y}$};
		\node [style=right label] (40) at (2.5, 1.5) {$\mathds{B}$};
		\node [style=right label] (41) at (-2.25, 1.5) {$\mathds{A}$};
	\end{pgfonlayer}
	\begin{pgfonlayer}{edgelayer}
		\draw (2.center) to (1.center);
		\draw (0.center) to (2.center);
		\draw (0.center) to (4.center);
		\draw (4.center) to (5.center);
		\draw (5.center) to (1.center);
		\draw [b=12, in=90, out=-90] (15.center) to (14.center);
		\draw (19.center) to (18.center);
		\draw (18.center) to (17.center);
		\draw (17.center) to (16.center);
		\draw (16.center) to (19.center);
		\draw [c, in=90, out=-90] (23) to (22.center);
		\draw [c, in=90, out=-90] (21.center) to (25);
		\draw [b=12, in=90, out=-90] (27.center) to (26.center);
		\draw (31.center) to (30.center);
		\draw (30.center) to (29.center);
		\draw (29.center) to (28.center);
		\draw (28.center) to (31.center);
		\draw [c, in=90, out=-90] (35) to (34.center);
		\draw [c, in=90, out=-90] (33.center) to (37);
	\end{pgfonlayer}
\end{tikzpicture}
,
\end{equation}
with the state $s_\mathrm{PR}$  and controlled measurements $M_{\mathrm{PR}}$  and $M_{\mathrm{PR}}^{\prime}$  as introduced in Ref.~\cite{barrett2007gpt}, whose explicit form we present in Eqs.~C8, C3, and C6 in the Supplemental Material. Note that for simplicity of notation we will often label the classical systems by an outcome or setting variable such as $\mathds{X}$, in this case the relevant GPT system is $\C{|\mathds{X}|}$.

The situation slightly changes when we instead focus on the non-classical phenomenon of steering (see Sec.~D of the Supplemental Material for a comprehensive introduction). In brief, a traditional bipartite steering experiment consists of two distant parties, Alice and Bob, who share a physical system, perform space-like separated actions, and, unlike in Bell experiments, play asymmetric roles in the experiment. On the one hand, Alice (sometimes referred to as the black-box party in the steering literature) chooses a measurement labelled by $x \in \X$ to perform on her share of the system, and obtains a classical outcome $a \in \A$ with probability $p(a|x)$. Bob, on the other hand, merely characterises the quantum state $\rho^B_{a|x}$ to which his subsystem is steered. The information collected in this experiment (Alice's probabilities and Bob's conditional states) is expressed concisely as an assemblage \cite{pusey2013negativity}:
\begin{align}
\As_{\A|\X}= \left\{ \sigma^B_{a|x} \right\}_{a\in\A,\,x\in\X} \ \text{where} \quad \sigma^B_{a|x} := p(a|x) \, \rho^B_{a|x}.
\end{align}
Note that in Ref.~\cite{pqsc} it is shown how assemblages can be equivalently represented by so-called `causal' channels. With a slight abuse of notation we therefore  diagrammatically represent the assemblage by the causal classical-quantum channel:
\beq
\begin{tikzpicture}
	\begin{pgfonlayer}{nodelayer}
		\node [style=none] (0) at (-1, 0.5) {};
		\node [style=none] (1) at (1.5, 0.5) {};
		\node [style=none] (2) at (0.75, -0.5) {};
		\node [style=none] (3) at (-1, -0.5) {};
		\node [style=none] (4) at (0, 0) {$\Sigma$};
		\node [style=none] (5) at (-0.5, 0.5) {};
		\node [style=none] (6) at (-0.5, 1.25) {};
		\node [style=none] (7) at (0.75, 0.5) {};
		\node [style=none] (8) at (0.75, 1.25) {};
		\node [style=none] (9) at (-0.5, -0.5) {};
		\node [style=none] (10) at (-0.5, -1.25) {};
		\node [style=right label] (11) at (-0.5, -1.25) {$\mathds{X}$};
		\node [style=right label] (12) at (-0.5, 1) {$\mathds{A}$};
	\end{pgfonlayer}
	\begin{pgfonlayer}{edgelayer}
		\draw (0.center) to (3.center);
		\draw (3.center) to (2.center);
		\draw (2.center) to (1.center);
		\draw (1.center) to (0.center);
		\draw [c] (9.center) to (10.center);
		\draw [c] (5.center) to (6.center);
		\draw [q] (7.center) to (8.center);
	\end{pgfonlayer}
\end{tikzpicture}
\ .
\eeq
To see that this is indeed a good representation, note that we can extract the elements of the assemblage, i.e., the subnormalised steered states as:
\begin{equation}
\begin{tikzpicture}
	\begin{pgfonlayer}{nodelayer}
		\node [style=none] (38) at (4.5, 0) {};
		\node [style=none] (39) at (6.5, 0) {};
		\node [style=none] (40) at (5.5, -1.25) {};
		\node [style=none] (41) at (5.5, 0) {};
		\node [style=none] (42) at (5.5, 1) {};
		\node [style=none] (43) at (5.5, -0.5) {$\sigma_{a|x}$};
	\end{pgfonlayer}
	\begin{pgfonlayer}{edgelayer}
		\draw (38.center) to (41.center);
		\draw (41.center) to (39.center);
		\draw (39.center) to (40.center);
		\draw (40.center) to (38.center);
		\draw [q] (41.center) to (42.center);
	\end{pgfonlayer}
\end{tikzpicture}
 \ \ =	\ \ \begin{tikzpicture}
	\begin{pgfonlayer}{nodelayer}
		\node [style=none] (0) at (-1, 0.5) {};
		\node [style=none] (1) at (1.5, 0.5) {};
		\node [style=none] (2) at (0.75, -0.5) {};
		\node [style=none] (3) at (-1, -0.5) {};
		\node [style=none] (4) at (0, 0) {$\Sigma$};
		\node [style=none] (5) at (-0.5, 0.5) {};
		\node [style=copoint] (6) at (-0.5, 1.5) {$a$};
		\node [style=none] (7) at (0.75, 0.5) {};
		\node [style=none] (8) at (0.75, 2) {};
		\node [style=none] (9) at (-0.5, -0.5) {};
		\node [style=point] (10) at (-0.5, -1.75) {$x$};
		\node [style=right label] (11) at (-0.5, -1) {$\mathds{X}$};
		\node [style=right label] (12) at (-0.5, 0.75) {$\mathds{A}$};
	\end{pgfonlayer}
	\begin{pgfonlayer}{edgelayer}
		\draw (0.center) to (3.center);
		\draw (3.center) to (2.center);
		\draw (2.center) to (1.center);
		\draw (1.center) to (0.center);
		\draw [c] (9.center) to (10);
		\draw [c] (5.center) to (6);
		\draw [q] (7.center) to (8.center);
	\end{pgfonlayer}
\end{tikzpicture}
,
\end{equation}
and then the probabilities as:
\begin{equation}
p(a|x)\ \ = \ \ \begin{tikzpicture}
	\begin{pgfonlayer}{nodelayer}
		\node [style=none] (38) at (-1, -0.25) {};
		\node [style=none] (39) at (1, -0.25) {};
		\node [style=none] (40) at (0, -1.5) {};
		\node [style=none] (41) at (0, -0.25) {};
		\node [style=none] (42) at (0, 0.5) {};
		\node [style=none] (43) at (0, -0.75) {$\sigma_{a|x}$};
		\node [style=upground] (44) at (0, 0.75) {};
	\end{pgfonlayer}
	\begin{pgfonlayer}{edgelayer}
		\draw (38.center) to (41.center);
		\draw (41.center) to (39.center);
		\draw (39.center) to (40.center);
		\draw (40.center) to (38.center);
		\draw [q] (41.center) to (42.center);
	\end{pgfonlayer}
\end{tikzpicture}
 \ \ =	\ \ \begin{tikzpicture}
	\begin{pgfonlayer}{nodelayer}
		\node [style=none] (0) at (-1, 0.5) {};
		\node [style=none] (1) at (1.5, 0.5) {};
		\node [style=none] (2) at (0.75, -0.5) {};
		\node [style=none] (3) at (-1, -0.5) {};
		\node [style=none] (4) at (0, 0) {$\Sigma$};
		\node [style=none] (5) at (-0.5, 0.5) {};
		\node [style=copoint] (6) at (-0.5, 1.5) {$a$};
		\node [style=none] (7) at (0.75, 0.5) {};
		\node [style=none] (8) at (0.75, 1.25) {};
		\node [style=none] (9) at (-0.5, -0.5) {};
		\node [style=point] (10) at (-0.5, -1.75) {$x$};
		\node [style=right label] (11) at (-0.5, -1) {$\mathds{X}$};
		\node [style=right label] (12) at (-0.5, 0.75) {$\mathds{A}$};
		\node [style=upground] (13) at (0.75, 1.5) {};
	\end{pgfonlayer}
	\begin{pgfonlayer}{edgelayer}
		\draw (0.center) to (3.center);
		\draw (3.center) to (2.center);
		\draw (2.center) to (1.center);
		\draw (1.center) to (0.center);
		\draw [c] (9.center) to (10);
		\draw [c] (5.center) to (6);
		\draw [q] (7.center) to (8.center);
	\end{pgfonlayer}
\end{tikzpicture}
\,,
\end{equation}
\bel{where $\begin{tikzpicture}
	\begin{pgfonlayer}{nodelayer}
		\node [style=none] (0) at (0, -0.75) {};
		\node [style=none] (1) at (0, 0.25) {};
		\node [style=upground] (3) at (0, 0.5) {};
	\end{pgfonlayer}
	\begin{pgfonlayer}{edgelayer}
		\draw [double, in=90, out=-90] (1.center) to (0.center);
	\end{pgfonlayer}
\end{tikzpicture}
$ denotes the so-called unit effect (see Eq.~A11 and its precedding paragraph in the Supplemental Material), which in quantum theory corresponds to the \bel{partial trace of the relevant subsystem}.}
Analogously, we can view non-signalling correlations as particular channels, in this case channels with classical inputs and outputs which correspond to stochastic maps. Using this we can rewrite Eq.~\eqref{eq:PRbox} as
\beq\label{eq:PRBoxChannel}
\begin{tikzpicture}
	\begin{pgfonlayer}{nodelayer}
		\node [style=none] (16) at (1.25, -0.5) {};
		\node [style=none] (21) at (0.75, -0.5) {};
		\node [style=none] (22) at (0.75, 0.5) {};
		\node [style=none] (23) at (0.75, 1.5) {};
		\node [style=none] (25) at (0.75, -1.5) {};
		\node [style=none] (28) at (-1.25, -0.5) {};
		\node [style=none] (30) at (1.25, 0.5) {};
		\node [style=none] (31) at (-1.25, 0.5) {};
		\node [style=none] (33) at (-0.75, -0.5) {};
		\node [style=none] (34) at (-0.75, 0.5) {};
		\node [style=none] (35) at (-0.75, 1.5) {};
		\node [style=none] (37) at (-0.75, -1.5) {};
		\node [style=right label] (44) at (-0.75, -1.25) {$\X$};
		\node [style=right label] (45) at (0.75, -1.25) {$\Y$};
		\node [style=right label] (46) at (-0.75, 1.25) {$\A$};
		\node [style=right label] (47) at (0.75, 1.25) {$\mathds{B}$};
		\node [style=none] (48) at (0, 0) {$\mathrm{PR}$};
	\end{pgfonlayer}
	\begin{pgfonlayer}{edgelayer}
		\draw [c, in=90, out=-90] (23.center) to (22.center);
		\draw [c, in=90, out=-90] (21.center) to (25.center);
		\draw (31.center) to (30.center);
		\draw (28.center) to (31.center);
		\draw [c, in=90, out=-90] (35.center) to (34.center);
		\draw [c, in=90, out=-90] (33.center) to (37.center);
		\draw (28.center) to (16.center);
		\draw (16.center) to (30.center);
	\end{pgfonlayer}
\end{tikzpicture}
 \ \ = \ \ \begin{tikzpicture}
	\begin{pgfonlayer}{nodelayer}
		\node [style=none] (0) at (-1.25, -0.75) {};
		\node [style=none] (1) at (1.25, -0.75) {};
		\node [style=none] (2) at (0, -2) {};
		\node [style=none] (3) at (0, -1.25) {$s_{\mathrm{PR}}$};
		\node [style=none] (4) at (-0.5, -0.75) {};
		\node [style=none] (5) at (0.5, -0.75) {};
		\node [style=none] (14) at (0.5, -0.75) {};
		\node [style=none] (15) at (1.5, 0.25) {};
		\node [style=none] (16) at (3, 0.25) {};
		\node [style=none] (17) at (1, 0.25) {};
		\node [style=none] (18) at (1.5, 1.25) {};
		\node [style=none] (19) at (3, 1.25) {};
		\node [style=none] (20) at (2.25, 0.75) {$M_{\mathrm{PR}}^\prime$};
		\node [style=none] (21) at (2.25, 0.25) {};
		\node [style=none] (22) at (2.25, 1.25) {};
		\node [style=none] (23) at (2.25, 2.25) {};
		\node [style=none] (24) at (0.5, -0.75) {};
		\node [style=none] (25) at (2.25, -2.25) {};
		\node [style=none] (26) at (-0.5, -0.75) {};
		\node [style=none] (27) at (-1.5, 0.25) {};
		\node [style=none] (28) at (-3, 0.25) {};
		\node [style=none] (29) at (-1, 0.25) {};
		\node [style=none] (30) at (-1.5, 1.25) {};
		\node [style=none] (31) at (-3, 1.25) {};
		\node [style=none] (32) at (-2.25, 0.75) {$M_{\mathrm{PR}}$};
		\node [style=none] (33) at (-2.25, 0.25) {};
		\node [style=none] (34) at (-2.25, 1.25) {};
		\node [style=none] (35) at (-2.25, 2.25) {};
		\node [style=none] (36) at (-0.5, -0.75) {};
		\node [style=none] (37) at (-2.25, -2.25) {};
		\node [style=right label] (44) at (-2.25, -2) {$\X$};
		\node [style=right label] (45) at (2.25, -2) {$\Y$};
		\node [style=right label] (46) at (-2.25, 2) {$\A$};
		\node [style=right label] (47) at (2.25, 2) {$\mathds{B}$};
	\end{pgfonlayer}
	\begin{pgfonlayer}{edgelayer}
		\draw (2.center) to (1.center);
		\draw (0.center) to (2.center);
		\draw (0.center) to (4.center);
		\draw (4.center) to (5.center);
		\draw (5.center) to (1.center);
		\draw [b=12, in=90, out=-90] (15.center) to (14.center);
		\draw (19.center) to (18.center);
		\draw (18.center) to (17.center);
		\draw (17.center) to (16.center);
		\draw (16.center) to (19.center);
		\draw [c, in=90, out=-90] (23.center) to (22.center);
		\draw [c, in=90, out=-90] (21.center) to (25.center);
		\draw [b=12, in=90, out=-90] (27.center) to (26.center);
		\draw (31.center) to (30.center);
		\draw (30.center) to (29.center);
		\draw (29.center) to (28.center);
		\draw (28.center) to (31.center);
		\draw [c, in=90, out=-90] (35.center) to (34.center);
		\draw [c, in=90, out=-90] (33.center) to (37.center);
	\end{pgfonlayer}
\end{tikzpicture}
.
\eeq

Beyond the traditional scenario, one may have steering experiments with more black-box parties also in a space-like separated configuration \cite{cavalcanti2015detection,pqs}, or even situations where Bob may influence the state preparation of his system by choosing a classical variable $y$ (Bob-with-input scenarios) \cite{bpqs}.

In a similar fashion to Bell non-classicality, one can define what ``classical'' (a.k.a.~LHS), quantum, and non-signalling assemblages are \cite{pqs,bpqs}. Notice that the differences in all these kinds of steering are  not related to the type of system prepared in Bob's lab, but rather to the types of shared resources that are used to prepare those quantum systems in Bob's lab. From the point of view of Witworld, then, an assemblage in a steering experiment is produced by the parties performing local operations in a shared arbitrary composite multipartite system, which may include classical, quantum, and Boxworld systems. One fascinating property of Witworld is that it not only features all LHS and quantum assemblages (see Defs.~D.5 and D.6, respectively, in the Supplemental Material), but may also realise post-quantum assemblages. That is, Witworld features post-quantum steering. In this section we present a few key  examples of this. Whether Witworld can realise \textit{all} non-signalling assemblages is still an open question (see Fig.~\ref{fig:SteeSet}).

The first example we present is in a tripartite steering scenario, \bel{since in traditional bipartite steering scenarios post-quantum steering is forbidden by the Gisin \cite{gisin1989stochastic} and Hughston, Josza and Wootters \cite{hughston1993complete} theorems. In a tripartite scenario, it is enough to consider the simplest setup} with two black-box parties choosing among two dichotomic measurements each, so $\X=\A=\{0,1\}$, and where Bob's subsystem is a qubit. The particular assemblage \bel{we present} is the PR-box assemblage, defined by 
\begin{align}
&\As^{\mathrm{PR}}_{\A\A|\X\X} = \left\{ \sigma^{*B}_{a_1a_2|x_1x_2}\right\}_{a_j \in \A, x_j \in \X, j \in \{1,2\}}\,, \\
& \text{with} \quad \sigma^{*B}_{a_1a_2|x_1x_2} = p_{\mathrm{PR}}(a_1a_2|x_1x_2) \, \frac{\id}{2} \,.
\end{align}
This assemblage cannot be realised by the three parties sharing quantum resources \cite{pqs}, i.e., it is post-quantum. 
$\As^{\mathrm{PR}}_{\A\A|\X\X} $ can however be realised within Witworld when the parties share the following mutipartite system: a bipartite Boxworld system of dimension $(2,2)$ on a PR state shared by the black-box parties, composed in parallel with a quantum state $\rho^{*B}=\frac{\id}{2}$ for Bob. Leveraging the realisation of PR-box correlations as in Eq.~\eqref{eq:PRbox}, the  assemblage $\As^{\mathrm{PR}}_{\A\A|\X\X} $ can be realised by: 
\begin{equation}\label{pqprcor}
\begin{tikzpicture}
	\begin{pgfonlayer}{nodelayer}
		\node [style=none] (38) at (-2, 0) {};
		\node [style=none] (39) at (2, 0) {};
		\node [style=none] (41) at (1, 0) {};
		\node [style=none] (42) at (1, 1) {};
		\node [style=none] (43) at (0, -0.5) {$\Sigma^{\mathrm{PR}}$};
		\node [style=none] (44) at (-2, -1) {};
		\node [style=none] (45) at (1, -1) {};
		\node [style=none] (46) at (-1.5, 0) {};
		\node [style=none] (47) at (-1.5, 1) {};
		\node [style=none] (48) at (-0.5, 0) {};
		\node [style=none] (49) at (-0.5, 1) {};
		\node [style=none] (50) at (-1.5, -2) {};
		\node [style=none] (51) at (-1.5, -1) {};
		\node [style=none] (52) at (-0.5, -2) {};
		\node [style=none] (53) at (-0.5, -1) {};
		\node [style=right label] (54) at (-1.5, -1.75) {$\X$};
		\node [style=right label] (55) at (-0.5, -1.75) {$\X$};
		\node [style=right label] (56) at (-1.5, 0.75) {$\A$};
		\node [style=right label] (57) at (-0.5, 0.75) {$\A$};
	\end{pgfonlayer}
	\begin{pgfonlayer}{edgelayer}
		\draw (38.center) to (41.center);
		\draw (41.center) to (39.center);
		\draw [q] (41.center) to (42.center);
		\draw (38.center) to (44.center);
		\draw (44.center) to (45.center);
		\draw (45.center) to (39.center);
		\draw [c] (46.center) to (47.center);
		\draw [c] (48.center) to (49.center);
		\draw [c] (50.center) to (51.center);
		\draw [c] (52.center) to (53.center);
	\end{pgfonlayer}
\end{tikzpicture}
 \ \  = \ \ \begin{tikzpicture}
	\begin{pgfonlayer}{nodelayer}
		\node [style=none] (0) at (-1.25, -0.75) {};
		\node [style=none] (1) at (1.25, -0.75) {};
		\node [style=none] (2) at (0, -2) {};
		\node [style=none] (3) at (0, -1.25) {$s_{\mathrm{PR}}$};
		\node [style=none] (4) at (-0.5, -0.75) {};
		\node [style=none] (5) at (0.5, -0.75) {};
		\node [style=none] (14) at (0.5, -0.75) {};
		\node [style=none] (15) at (1.5, 0.25) {};
		\node [style=none] (16) at (3, 0.25) {};
		\node [style=none] (17) at (1, 0.25) {};
		\node [style=none] (18) at (1.5, 1.25) {};
		\node [style=none] (19) at (3, 1.25) {};
		\node [style=none] (20) at (2.25, 0.75) {$M_{\mathrm{PR}}^\prime$};
		\node [style=none] (21) at (2.25, 0.25) {};
		\node [style=none] (22) at (2.25, 1.25) {};
		\node [style=none] (23) at (2.25, 2.25) {};
		\node [style=none] (24) at (0.5, -0.75) {};
		\node [style=none] (25) at (2.25, -2.25) {};
		\node [style=none] (26) at (-0.5, -0.75) {};
		\node [style=none] (27) at (-1.5, 0.25) {};
		\node [style=none] (28) at (-3, 0.25) {};
		\node [style=none] (29) at (-1, 0.25) {};
		\node [style=none] (30) at (-1.5, 1.25) {};
		\node [style=none] (31) at (-3, 1.25) {};
		\node [style=none] (32) at (-2.25, 0.75) {$M_{\mathrm{PR}}$};
		\node [style=none] (33) at (-2.25, 0.25) {};
		\node [style=none] (34) at (-2.25, 1.25) {};
		\node [style=none] (35) at (-2.25, 2.25) {};
		\node [style=none] (36) at (-0.5, -0.75) {};
		\node [style=none] (37) at (-2.25, -2.25) {};
		\node [style=none] (38) at (4, 0.25) {};
		\node [style=none] (39) at (6, 0.25) {};
		\node [style=none] (40) at (5, -1) {};
		\node [style=none] (41) at (5, 0.25) {};
		\node [style=none] (42) at (5, 2.25) {};
		\node [style=none] (43) at (5, -0.25) {$\rho^*$};
		\node [style=right label] (44) at (-2.25, -2) {$\X$};
		\node [style=right label] (45) at (2.25, -2) {$\X$};
		\node [style=right label] (46) at (-2.25, 2) {$\A$};
		\node [style=right label] (47) at (2.25, 2) {$\A$};
	\end{pgfonlayer}
	\begin{pgfonlayer}{edgelayer}
		\draw (2.center) to (1.center);
		\draw (0.center) to (2.center);
		\draw (0.center) to (4.center);
		\draw (4.center) to (5.center);
		\draw (5.center) to (1.center);
		\draw [b=12, in=90, out=-90] (15.center) to (14.center);
		\draw (19.center) to (18.center);
		\draw (18.center) to (17.center);
		\draw (17.center) to (16.center);
		\draw (16.center) to (19.center);
		\draw [c, in=90, out=-90] (23.center) to (22.center);
		\draw [c, in=90, out=-90] (21.center) to (25.center);
		\draw [b=12, in=90, out=-90] (27.center) to (26.center);
		\draw (31.center) to (30.center);
		\draw (30.center) to (29.center);
		\draw (29.center) to (28.center);
		\draw (28.center) to (31.center);
		\draw [c, in=90, out=-90] (35.center) to (34.center);
		\draw [c, in=90, out=-90] (33.center) to (37.center);
		\draw (38.center) to (41.center);
		\draw (41.center) to (39.center);
		\draw (39.center) to (40.center);
		\draw (40.center) to (38.center);
		\draw [q] (41.center) to (42.center);
	\end{pgfonlayer}
\end{tikzpicture}
.
\end{equation}

The second example we present is in a bipartite Bob-with-input steering scenario, where Alice has $\X=\A=\{0,1\}$, Bob's subsystem is a qubit, and Bob's input is dichotomic (i.e., $y \in \Y=\{0,1\}$). The particular post-quantum assemblage $\As^*_{\A|\X\Y}$ we consider has elements defined by $\sigma^{*B}_{a|xy} = \frac{1}{2} ( \ket{a}\bra{a} \delta_{xy = 0} + \ket{a \oplus 1} \bra{a \oplus 1} \delta_{xy = 1})$ \cite{bpqs}. This assemblage can be realised in Witworld by Alice and Bob sharing a bipartite Boxworld system of dimension $(2,2)$ prepared in a PR state, and implementing the following protocol. On the one hand, Alice performs the measurement $M_{\mathrm{PR}}$ of Eq.~\eqref{eq:PRbox} controlled on her classical input $x$, and obtains the output $a$. On the other hand, here the state preparation of Bob's system further depends on Bob's choice of a classical variable $y$ which he inputs in a device. In this protocol, this device has a two-stage process: first it implements the measurement $M^{\prime}_{\mathrm{PR}}$ of Eq.~\eqref{eq:PRbox} on the Boxworld system, conditioned on $y$; second, there is a controlled state preparation $P$ which prepares the quantum state $\ket{b}\bra{b}$ conditioned on $b$, the classical output of $M^{\prime}_{\mathrm{PR}}$. Diagrammatically, the whole protocol reads: 
\begin{equation}\label{pqprprep}
\begin{tikzpicture}
	\begin{pgfonlayer}{nodelayer}
		\node [style=none] (38) at (-1.25, 0.5) {};
		\node [style=none] (39) at (1.25, 0.5) {};
		\node [style=none] (41) at (0.75, 0.5) {};
		\node [style=none] (42) at (0.75, 1.5) {};
		\node [style=none] (43) at (0, 0) {$\Sigma^{*}$};
		\node [style=none] (44) at (-1.25, -0.5) {};
		\node [style=none] (45) at (1.25, -0.5) {};
		\node [style=none] (46) at (-0.75, 0.5) {};
		\node [style=none] (47) at (-0.75, 1.5) {};
		\node [style=none] (50) at (-0.75, -1.5) {};
		\node [style=none] (51) at (-0.75, -0.5) {};
		\node [style=none] (52) at (0.75, -1.5) {};
		\node [style=none] (53) at (0.75, -0.5) {};
		\node [style=right label] (54) at (-0.75, -1.25) {$\X$};
		\node [style=right label] (55) at (0.75, -1.25) {$\Y$};
		\node [style=right label] (56) at (-0.75, 1.25) {$\A$};
	\end{pgfonlayer}
	\begin{pgfonlayer}{edgelayer}
		\draw (38.center) to (41.center);
		\draw (41.center) to (39.center);
		\draw [q] (41.center) to (42.center);
		\draw (38.center) to (44.center);
		\draw (44.center) to (45.center);
		\draw (45.center) to (39.center);
		\draw [c] (46.center) to (47.center);
		\draw [c] (50.center) to (51.center);
		\draw [c] (52.center) to (53.center);
	\end{pgfonlayer}
\end{tikzpicture}
 \ \  = \ \  \begin{tikzpicture}
	\begin{pgfonlayer}{nodelayer}
		\node [style=none] (20) at (2.25, -1) {};
		\node [style=none] (21) at (2.25, -4.25) {};
		\node [style=none] (22) at (0.5, -3) {};
		\node [style=none] (23) at (1.5, -1) {};
		\node [style=none] (24) at (3, -1) {};
		\node [style=none] (25) at (1, -1) {};
		\node [style=none] (26) at (1.5, 0) {};
		\node [style=none] (27) at (3, 0) {};
		\node [style=none] (28) at (2.25, -0.5) {$M^\prime_{\mathrm{PR}}$};
		\node [style=none] (29) at (2.25, -1) {};
		\node [style=none] (30) at (2.25, 0) {};
		\node [style=none] (31) at (2.25, 1) {};
		\node [style=none] (33) at (1.5, 1) {};
		\node [style=none] (34) at (3, 1) {};
		\node [style=none] (35) at (3, 2) {};
		\node [style=none] (36) at (1.5, 2) {};
		\node [style=none] (37) at (2.25, 1.5) {P};
		\node [style=none] (38) at (2.25, 2) {};
		\node [style=none] (39) at (2.25, 3) {};
		\node [style=none] (40) at (-2.25, -1) {};
		\node [style=none] (41) at (-2.25, -4.25) {};
		\node [style=none] (42) at (-0.5, -3) {};
		\node [style=none] (43) at (-1.5, -1) {};
		\node [style=none] (44) at (-3, -1) {};
		\node [style=none] (45) at (-1, -1) {};
		\node [style=none] (46) at (-1.5, 0) {};
		\node [style=none] (47) at (-3, 0) {};
		\node [style=none] (48) at (-2.25, -0.5) {$M_{\mathrm{PR}}$};
		\node [style=none] (49) at (-2.25, -1) {};
		\node [style=none] (50) at (-2.25, 0) {};
		\node [style=none] (51) at (-2.25, 3) {};
		\node [style=none] (52) at (-1.25, -3) {};
		\node [style=none] (53) at (1.25, -3) {};
		\node [style=none] (54) at (0, -4.25) {};
		\node [style=none] (55) at (0, -3.5) {$s_{\mathrm{PR}}$};
		\node [style=none] (56) at (-0.5, -3) {};
		\node [style=none] (57) at (0.5, -3) {};
		\node [style=right label] (58) at (-2.25, 2.75) {$\A$};
		\node [style=right label] (59) at (-2.25, -4) {$\X$};
		\node [style=right label] (60) at (2.25, -4) {$\Y$};
	\end{pgfonlayer}
	\begin{pgfonlayer}{edgelayer}
		\draw [c, in=90, out=-90] (20.center) to (21.center);
		\draw [b=12, in=90, out=-90] (23.center) to (22.center);
		\draw (27.center) to (26.center);
		\draw (26.center) to (25.center);
		\draw (25.center) to (24.center);
		\draw (24.center) to (27.center);
		\draw [c, in=90, out=-90] (31.center) to (30.center);
		\draw (34.center) to (31.center);
		\draw (31.center) to (33.center);
		\draw (33.center) to (36.center);
		\draw (36.center) to (38.center);
		\draw (38.center) to (35.center);
		\draw (35.center) to (34.center);
		\draw [q] (38.center) to (39.center);
		\draw [c, in=90, out=-90] (40.center) to (41.center);
		\draw [b=12, in=90, out=-90] (43.center) to (42.center);
		\draw (47.center) to (46.center);
		\draw (46.center) to (45.center);
		\draw (45.center) to (44.center);
		\draw (44.center) to (47.center);
		\draw [c, in=90, out=-90] (51.center) to (50.center);
		\draw (54.center) to (53.center);
		\draw (52.center) to (54.center);
		\draw (52.center) to (56.center);
		\draw (56.center) to (57.center);
		\draw (57.center) to (53.center);
	\end{pgfonlayer}
\end{tikzpicture}
\,.
\end{equation}
One can readily see that Eq.~\eqref{pqprprep} indeed holds, since the assemblage elements of $\As^*_{\A|\X\Y}$ can be rewritten as  $\sigma^{*B}_{a|xy} = \frac{1}{2} \ket{a \oplus xy}\bra{a \oplus xy}$, and PR-box correlations satisfy $b=a\oplus xy$ and $\frac{1}{2} = \sum_b p_{\mathrm{PR}}(ab|xy)$. 

The third example we present is that of {Gleason assemblages} \cite{pqso}. In short, Gleason assemblages are those that can be mathematically expressed in the language of quantum theory by having the parties measure a shared system whose state preparation is represented by a normalised quantum entanglement witness. Gleason assemblages are particularly useful, since there are constructions that yield provably post-quantum Gleason assemblages. More importantly, the post-quantumness of some Gleason assemblages is not implied by post-quantum Bell non-locality, which renders post-quantum steering as a genuinely new effect \cite{pqso}. Witworld readily provides realisations of any Gleason assemblage, by noticing two facts: (i) any quantum entanglement witness is a valid state of composite quantum-type systems in Witworld (Thm.~B.6), and (ii) in Witworld, any local quantum measurement is a valid Witworld measurement  (Lem.~8.14). The explicit diagram for a Witworld realisation of a generic Gleason assemblage $\As^\mathrm{G}_{\A_1\A_2|\X_1\X_2}$ in a tripartite scenario with two black-box parties is:
\begin{equation}
\begin{tikzpicture}
	\begin{pgfonlayer}{nodelayer}
		\node [style=none] (38) at (-1.25, 0.5) {};
		\node [style=none] (39) at (2.25, 0.5) {};
		\node [style=none] (41) at (1.75, 0.5) {};
		\node [style=none] (42) at (1.75, 1.5) {};
		\node [style=none] (43) at (0.5, 0) {$\Sigma^{G}$};
		\node [style=none] (44) at (-1.25, -0.5) {};
		\node [style=none] (45) at (1.75, -0.5) {};
		\node [style=none] (46) at (-0.75, 0.5) {};
		\node [style=none] (47) at (-0.75, 1.5) {};
		\node [style=none] (50) at (-0.75, -1.5) {};
		\node [style=none] (51) at (-0.75, -0.5) {};
		\node [style=none] (52) at (0.25, -1.5) {};
		\node [style=none] (53) at (0.25, -0.5) {};
		\node [style=right label] (54) at (-0.75, -1.25) {$\X_1$};
		\node [style=right label] (55) at (0.25, -1.25) {$\X_2$};
		\node [style=right label] (56) at (-0.75, 1.25) {$\A_1$};
		\node [style=none] (57) at (0.25, 0.5) {};
		\node [style=none] (58) at (0.25, 1.5) {};
		\node [style=right label] (59) at (0.25, 1.25) {$\A_2$};
	\end{pgfonlayer}
	\begin{pgfonlayer}{edgelayer}
		\draw (38.center) to (41.center);
		\draw (41.center) to (39.center);
		\draw [q] (41.center) to (42.center);
		\draw (38.center) to (44.center);
		\draw (44.center) to (45.center);
		\draw (45.center) to (39.center);
		\draw [c] (46.center) to (47.center);
		\draw [c] (50.center) to (51.center);
		\draw [c] (52.center) to (53.center);
		\draw [c] (57.center) to (58.center);
	\end{pgfonlayer}
\end{tikzpicture}
 \ \  = \ \ \begin{tikzpicture}
	\begin{pgfonlayer}{nodelayer}
		\node [style=none] (0) at (-3, -1.5) {};
		\node [style=none] (1) at (7, -1.5) {};
		\node [style=none] (2) at (2, -2.75) {};
		\node [style=none] (3) at (2, -2) {Witness};
		\node [style=none] (4) at (-0.5, -1.5) {};
		\node [style=none] (5) at (0.5, -1.5) {};
		\node [style=none] (14) at (1.75, -1.5) {};
		\node [style=none] (15) at (1.25, 0.25) {};
		\node [style=none] (16) at (-0.25, 0.25) {};
		\node [style=none] (17) at (1.75, 0.25) {};
		\node [style=none] (18) at (1.25, 1.25) {};
		\node [style=none] (19) at (-0.25, 1.25) {};
		\node [style=none] (20) at (0.5, 0.75) {$N^\prime$};
		\node [style=none] (21) at (0.5, 0.25) {};
		\node [style=none] (22) at (0.5, 1.25) {};
		\node [style=none] (23) at (0.5, 2.25) {};
		\node [style=none] (24) at (1.75, -1.5) {};
		\node [style=none] (25) at (0.5, -3.5) {};
		\node [style=none] (26) at (-0.5, -1.5) {};
		\node [style=none] (27) at (-1.5, 0.25) {};
		\node [style=none] (28) at (-3, 0.25) {};
		\node [style=none] (29) at (-1, 0.25) {};
		\node [style=none] (30) at (-1.5, 1.25) {};
		\node [style=none] (31) at (-3, 1.25) {};
		\node [style=none] (32) at (-2.25, 0.75) {$N$};
		\node [style=none] (33) at (-2.25, 0.25) {};
		\node [style=none] (34) at (-2.25, 1.25) {};
		\node [style=none] (35) at (-2.25, 2.25) {};
		\node [style=none] (36) at (-0.5, -1.5) {};
		\node [style=none] (37) at (-2.25, -3.5) {};
		\node [style=none] (38) at (5.75, -1.5) {};
		\node [style=none] (39) at (5.75, 2.25) {};
		\node [style=none] (40) at (5.75, -1.5) {};
		\node [style=right label] (41) at (-2.25, -3.25) {$\X_1$};
		\node [style=right label] (42) at (0.5, -3.25) {$\X_2$};
		\node [style=right label] (43) at (-2.25, 2) {$\A_1$};
		\node [style=right label] (44) at (0.5, 2) {$\A_2$};
	\end{pgfonlayer}
	\begin{pgfonlayer}{edgelayer}
		\draw (2.center) to (1.center);
		\draw (0.center) to (2.center);
		\draw (0.center) to (4.center);
		\draw (4.center) to (5.center);
		\draw (5.center) to (1.center);
		\draw [q, in=90, out=-90] (15.center) to (14.center);
		\draw (19.center) to (18.center);
		\draw (18.center) to (17.center);
		\draw (17.center) to (16.center);
		\draw (16.center) to (19.center);
		\draw [c, in=90, out=-90] (23.center) to (22.center);
		\draw [c, in=90, out=-90] (21.center) to (25.center);
		\draw [q, in=90, out=-90] (27.center) to (26.center);
		\draw (31.center) to (30.center);
		\draw (30.center) to (29.center);
		\draw (29.center) to (28.center);
		\draw (28.center) to (31.center);
		\draw [c, in=90, out=-90] (35.center) to (34.center);
		\draw [c, in=90, out=-90] (33.center) to (37.center);
		\draw [q, in=90, out=-90] (39.center) to (38.center);
	\end{pgfonlayer}
\end{tikzpicture}
.
\end{equation} 

The fourth example that we present is in a bipartite Bob-with-input steering scenario, with $\A=\Y=\{0,1\}$ and $\X = \{1,2,3\}$. The particular post-quantum assemblage $\As^{**}_{\A|\X\Y}$ here has elements defined by \cite{bpqs}: 
\begin{align}\label{eq:BWIAssemblage}
\sigma^{**B}_{a|xy} = \frac{1}{4} \left( \id + (-1)^{a+\delta_{x,2}\delta_{y,1}} \Sigma_x \right) \,,
\end{align}
where $(\Sigma_1,\Sigma_2,\Sigma_3)$ are the Pauli X, Y, and Z operators, respectively. $\As^{**}_{\A|\X\Y}$ is the first assemblage found in the Bob-with-input scenario whose post-quantumness cannot be proven directly from leveraging post-quantum Bell non-locality, which renders this type of post-quantum steering as a genuinely new effect. To see that Witworld can realise this assemblage, first notice that its elements can be mathematically written as $\sigma^{**B}_{a|xy} = (\overline{\sigma}^{B}_{a|xy})^{\top_y}$, where $\overline{\sigma}^{B}_{a|xy} = \frac{1}{4} \left( \id + (-1)^{a} \Sigma_x \right)$ are the elements of a quantum assemblage (see Def.~D.6 in the Supplemental Material), and $\top_y$ is the identity operator for $y=0$ and the Transpose operator (denoted $\top$) for $y=1$. The final step is to observe that all of these mathematical objects are acceptable physical operations in Witworld: (i) the maximally entangled quantum state that realises $\{\overline{\sigma}^{B}_{a|xy}\}$ is a valid Witworld state preparation for Alice and Bob by Thm.~B.6, (ii) the quantum Pauli measurements for Alice that realise $\{\overline{\sigma}^{B}_{a|xy}\}$ are valid Witworld measurements by Lem.~A.14, where we denote the classically controlled Pauli measurement by $\textsc{Pauli}$, and (iii) the operations $\{\top_y\}$, which are positive quantum maps, are valid Witworld transformations by Thm.~A.15, and, hence, so too is the controlled transformation $c\top$ which implements $\top$ conditioned on a classical input system. Diagrammatically, this is represented as:

\begin{equation}\label{eq:BWIinWitworld}
\begin{tikzpicture}
	\begin{pgfonlayer}{nodelayer}
		\node [style=none] (38) at (-1.25, 0.5) {};
		\node [style=none] (39) at (1.25, 0.5) {};
		\node [style=none] (41) at (0.75, 0.5) {};
		\node [style=none] (42) at (0.75, 1.5) {};
		\node [style=none] (43) at (0, 0) {$\Sigma^{**}$};
		\node [style=none] (44) at (-1.25, -0.5) {};
		\node [style=none] (45) at (1.25, -0.5) {};
		\node [style=none] (46) at (-0.75, 0.5) {};
		\node [style=none] (47) at (-0.75, 1.5) {};
		\node [style=none] (50) at (-0.75, -1.5) {};
		\node [style=none] (51) at (-0.75, -0.5) {};
		\node [style=none] (52) at (0.75, -1.5) {};
		\node [style=none] (53) at (0.75, -0.5) {};
		\node [style=right label] (54) at (-0.75, -1.25) {$\X$};
		\node [style=right label] (55) at (0.75, -1.25) {$\Y$};
		\node [style=right label] (56) at (-0.75, 1.25) {$\A$};
	\end{pgfonlayer}
	\begin{pgfonlayer}{edgelayer}
		\draw (38.center) to (41.center);
		\draw (41.center) to (39.center);
		\draw [q] (41.center) to (42.center);
		\draw (38.center) to (44.center);
		\draw (44.center) to (45.center);
		\draw (45.center) to (39.center);
		\draw [c] (46.center) to (47.center);
		\draw [c] (50.center) to (51.center);
		\draw [c] (52.center) to (53.center);
	\end{pgfonlayer}
\end{tikzpicture}
 \ \ = \ \ \begin{tikzpicture}
	\begin{pgfonlayer}{nodelayer}
		\node [style=none] (21) at (2.75, 1) {};
		\node [style=none] (23) at (1, -1) {};
		\node [style=none] (24) at (2, 1) {};
		\node [style=none] (25) at (3.5, 1) {};
		\node [style=none] (26) at (1.5, 1) {};
		\node [style=none] (27) at (2, 2) {};
		\node [style=none] (28) at (3.5, 2) {};
		\node [style=none] (29) at (2.75, 1.5) {${c\top}$};
		\node [style=none] (30) at (2.75, 1) {};
		\node [style=none] (31) at (2.75, 2) {};
		\node [style=none] (32) at (2.75, 3) {};
		\node [style=none] (33) at (0, -1.5) {$\Phi^+$};
		\node [style=none] (34) at (-1.75, -1) {};
		\node [style=none] (35) at (0, -2.25) {};
		\node [style=none] (36) at (1.75, -1) {};
		\node [style=none] (37) at (-1, -1) {};
		\node [style=none] (38) at (1, -1) {};
		\node [style=none] (40) at (1, -1) {};
		\node [style=none] (42) at (-2.75, 1) {};
		\node [style=none] (43) at (-2.75, -2.5) {};
		\node [style=none] (44) at (-1, -1) {};
		\node [style=none] (45) at (-1.75, 1) {};
		\node [style=none] (46) at (-3.75, 1) {};
		\node [style=none] (47) at (-1.25, 1) {};
		\node [style=none] (48) at (-1.75, 2) {};
		\node [style=none] (49) at (-3.75, 2) {};
		\node [style=none] (50) at (-2.75, 1.5) {${\textsc{Pauli}}$};
		\node [style=none] (51) at (-2.75, 1) {};
		\node [style=none] (52) at (-2.75, 2) {};
		\node [style=none] (53) at (-2.75, 3) {};
		\node [style=none] (54) at (2.75, -2.5) {};
		\node [style=none] (55) at (2.75, 1) {};
		\node [style=right label] (56) at (-2.75, 2.75) {$\A$};
		\node [style=right label] (57) at (-2.75, -2.25) {$\X$};
		\node [style=right label] (58) at (2.75, -2.25) {$\Y$};
	\end{pgfonlayer}
	\begin{pgfonlayer}{edgelayer}
		\draw [q, in=90, out=-90] (24.center) to (23.center);
		\draw (28.center) to (27.center);
		\draw (27.center) to (26.center);
		\draw (26.center) to (25.center);
		\draw (25.center) to (28.center);
		\draw [q, in=90, out=-90] (32.center) to (31.center);
		\draw (35.center) to (36.center);
		\draw (36.center) to (34.center);
		\draw (34.center) to (35.center);
		\draw [c, in=90, out=-90] (42.center) to (43.center);
		\draw [q, in=90, out=-90] (45.center) to (44.center);
		\draw (49.center) to (48.center);
		\draw (48.center) to (47.center);
		\draw (47.center) to (46.center);
		\draw (46.center) to (49.center);
		\draw [c, in=90, out=-90] (53.center) to (52.center);
		\draw [c, in=-90, out=90] (54.center) to (55.center);
	\end{pgfonlayer}
\end{tikzpicture}
\,.
\end{equation}

The final example that we consider is   steering in the instrumental scenario \cite{bpqs}. This can be seen as an adaptation of the Bob-with-input scenario, in which Alice's output $a$ determines the setting $y$ for Bob. The particular example of post-quantum steering in this scenario that we present here is given by modifying our previous example, by wiring Alice's output to Bob's input. That is, it can be shown that the assemblage 
\begin{align}
\sigma^{IB}_{a|x} = \frac{1}{4} \left( \id + (-1)^{a+\delta_{x,2}\delta_{a,1}} \Sigma_x \right) \,,
\end{align}
which is obtained by setting $y=a$ in Eq.~\eqref{eq:BWIAssemblage}, is post-quantum \cite{bpqs}. It is then a simple modification of Eq.~\eqref{eq:BWIinWitworld} to see that this too can be realised in Witworld:
\beq
\begin{tikzpicture}
	\begin{pgfonlayer}{nodelayer}
		\node [style=none] (38) at (-1.25, 0.5) {};
		\node [style=none] (39) at (1.25, 0.5) {};
		\node [style=none] (41) at (0.75, 0.5) {};
		\node [style=none] (42) at (0.75, 1.5) {};
		\node [style=none] (43) at (0, 0) {$\Sigma^{I}$};
		\node [style=none] (44) at (-1.25, -0.5) {};
		\node [style=none] (45) at (0.75, -0.5) {};
		\node [style=none] (46) at (-0.75, 0.5) {};
		\node [style=none] (47) at (-0.75, 1.5) {};
		\node [style=none] (50) at (-0.75, -1.5) {};
		\node [style=none] (51) at (-0.75, -0.5) {};
		\node [style=right label] (54) at (-0.75, -1.25) {$\X$};
		\node [style=right label] (56) at (-0.75, 1.25) {$\A$};
	\end{pgfonlayer}
	\begin{pgfonlayer}{edgelayer}
		\draw (38.center) to (41.center);
		\draw (41.center) to (39.center);
		\draw [q] (41.center) to (42.center);
		\draw (38.center) to (44.center);
		\draw (44.center) to (45.center);
		\draw (45.center) to (39.center);
		\draw [c] (46.center) to (47.center);
		\draw [c] (50.center) to (51.center);
	\end{pgfonlayer}
\end{tikzpicture}
 \ \ = \ \ \begin{tikzpicture}
	\begin{pgfonlayer}{nodelayer}
		\node [style=none] (23) at (1, -2.5) {};
		\node [style=none] (24) at (1.75, 3) {};
		\node [style=none] (25) at (3.25, 3) {};
		\node [style=none] (26) at (1.25, 3) {};
		\node [style=none] (27) at (1.75, 4) {};
		\node [style=none] (28) at (3.25, 4) {};
		\node [style=none] (29) at (2.5, 3.5) {${c\top}$};
		\node [style=none] (31) at (2.5, 4) {};
		\node [style=none] (32) at (2.5, 5) {};
		\node [style=none] (33) at (0, -3) {$\Phi^+$};
		\node [style=none] (34) at (-1.75, -2.5) {};
		\node [style=none] (35) at (0, -3.75) {};
		\node [style=none] (36) at (1.75, -2.5) {};
		\node [style=none] (37) at (-1, -2.5) {};
		\node [style=none] (38) at (1, -2.5) {};
		\node [style=none] (40) at (1, -2.5) {};
		\node [style=none] (42) at (-2.75, -0.5) {};
		\node [style=none] (43) at (-2.75, -4) {};
		\node [style=none] (44) at (-1, -2.5) {};
		\node [style=none] (45) at (-1.75, -0.5) {};
		\node [style=none] (46) at (-3.75, -0.5) {};
		\node [style=none] (47) at (-1.25, -0.5) {};
		\node [style=none] (48) at (-1.75, 0.5) {};
		\node [style=none] (49) at (-3.75, 0.5) {};
		\node [style=none] (50) at (-2.75, 0) {$\textsc{Pauli}$};
		\node [style=none] (51) at (-2.75, -0.5) {};
		\node [style=none] (53) at (-1.75, 3.25) {};
		\node [style=none] (54) at (-2.75, 0.5) {};
		\node [style=white dot] (56) at (-0.5, 1.75) {};
		\node [style=none] (57) at (-0.5, 1.75) {};
		\node [style=none] (58) at (-2.75, -4) {};
		\node [style=right label] (59) at (-2.75, -3.75) {$\X$};
		\node [style=right label] (60) at (-1.75, 4.75) {$\A$};
		\node [style=none] (61) at (-1.75, 5) {};
		\node [style=none] (62) at (2.75, 3) {};
	\end{pgfonlayer}
	\begin{pgfonlayer}{edgelayer}
		\draw [q, in=90, out=-90] (24.center) to (23.center);
		\draw (28.center) to (27.center);
		\draw (27.center) to (26.center);
		\draw (26.center) to (25.center);
		\draw (25.center) to (28.center);
		\draw [q, in=90, out=-90] (32.center) to (31.center);
		\draw (35.center) to (36.center);
		\draw (36.center) to (34.center);
		\draw (34.center) to (35.center);
		\draw [c, in=90, out=-90] (42.center) to (43.center);
		\draw [q, in=90, out=-90] (45.center) to (44.center);
		\draw (49.center) to (48.center);
		\draw (48.center) to (47.center);
		\draw (47.center) to (46.center);
		\draw (46.center) to (49.center);
		\draw [c, in=-105, out=45, looseness=0.75] (56) to (57.center);
		\draw [c, in=150, out=-90] (53.center) to (56);
		\draw [c, in=270, out=90] (54.center) to (56);
		\draw [c] (53.center) to (61.center);
		\draw [c, in=60, out=-90] (62.center) to (57.center);
	\end{pgfonlayer}
\end{tikzpicture}
\,,
\eeq
where the small white circle splitting the classical system is the copy operation.

With this we see that Witworld features a variety of non-classical and post-quantum properties, both in Bell and steering scenarios, and hence is the first GPT that has been shown to display post-quantum steering.

\subsection{Post-quantum advantage for information processing}\label{se:pqarsp}

Post-quantum resources may outperform quantum ones for information processing tasks \cite{QKD1,QKD2,QKD3,RAND1,RAND2,BellRev}. A natural question then is whether the post-quantum features of Witworld enable this theory to be more powerful than quantum theory in this respect. First, one can focus on device-independent information processing tasks, such as quantum cryptography \cite{QKD1,QKD2,QKD3}, which rely on the use of correlations in Bell scenarios. Here, it is known that Boxworld may outperform quantum theory, since it realises any non-signalling correlation. Since Boxworld is a subtheory of Witworld, then, the latter inherits these properties; that is, Witworld outperforms quantum theory in those device-independent information processing tasks. A more relevant question then is whether such advantage persists when moving on from device-independent tasks. Hence in this section we investigate whether Witworld provides an advantage for tasks that go beyond the processing of Bell-type correlations.

There are two features of Witworld that go beyond Boxworld which are noteworthy when looking for an information task where Witworld is resourceful. One is the fact that Witworld has quantum systems as \bel{atomic} system types, and the other is the fact that positive (but not necessarily completely positive) quantum operations are allowed physical operations in Witworld. \bel{Using these two facts} we first show that Witworld outperforms Quantum Theory in the task of Remote State Preparation, and then we show that the resource underlying this advantage is post-quantum steering.

Remote State Preparation (RSP) \cite{RSPlo,RSPdebbie} is a protocol with a similar flavour to state teleportation.
A main difference between teleportation and RSP is that in the former, Alice can send to Bob a state she knows nothing about, whereas in the latter she may require a complete classical description of $\ket{\psi}$. We denote this complete classical description by $\bm{\uppsi}$.
In both cases, the main goal is \bel{for} Alice \bel{to} deterministically prepare a state $ | \psi \rangle $ in Bob's lab, \bel{such that} he gets no additional information \bel{about} $ | \psi \rangle $. In RSP (see Fig.~\ref{rspfig}), however, Alice does not need to perform experimentally challenging entangling measurements (\bel{as in a} full Bell-state analysis) \cite{RSPexp}. Instead, she can directly encode the information about the state she wishes to send onto her share of an entangled state shared with Bob. When Alice and Bob use quantum resources, the minimum amount of classical information that she needs to send him for the protocol to succeed is $2 \log{d}$ bits of information, where $d$ is the dimension of the Hilbert space containing $ | \psi \rangle $ \cite{RSPdebbie}. Here we present a protocol using Witworld resources which may prepare an arbitrary qubit state in Bob's lab using only $1$ bit (instead of $2$) of classical communication.

Consider the following protocol in Witworld.  Alice and Bob share the two qubit state $ | \Phi^{s}\rangle = (| 01 \rangle - | 10 \rangle )/\sqrt{2}$. Alice performs the unitary $ U_{\bm{\uppsi}} = | 0 \rangle \langle \psi^{\perp} | + | 1 \rangle \langle \psi | $ on her qubit, which encodes the state $ | \psi \rangle $ to be sent. 
This effectively applies $U_{\bm{\uppsi}}^{\dagger}$ to Bob's half of the state (the singlet state $|\Phi^{s}\rangle$ transforms trivially under $U \otimes U$; implying that $( U_{\bm{\uppsi}} \otimes \mathbb{I} ) |\Phi^{s}\rangle =   (\mathbb{I} \otimes U_{\bm{\uppsi}}^{\dagger}) |\Phi^{s}\rangle$).
Next, she performs the measurement given by $ B = \{ | 0 \rangle \langle 0 | , | 1 \rangle \langle 1 |   \} $, whose outcome consists of one classical bit $a$ 
which indicates exactly whether Bob now has the post measured state $-\ket{\psi}$ (if $a = 0$) or $\ket{\psi^{\perp}}$ (if $a = 1$).
Then, Alice sends $a$ to Bob, who now knows whether he is holding $-\ket{\psi}$ or $\ket{\psi^{\perp}}$.  
The task can be completed if Bob has access to a {universal-\textup{NOT} operation}, which maps an arbitrary input $\ket{\phi}$ into an orthogonal state to it (which is unique up to global phases for qubits). 
The universal-NOT operation is not valid in quantum theory since it is a positive transformation, but not a completely positive transformation. 
However, in Witworld, this is an allowable transformation. 
Thus in Witworld Bob can apply the universal-NOT gate when $a=1$, leaving him with a perfect copy of $\ket{\psi}$ (up to a physically irrelevant global phase).
Diagrammatically, this protocol is represented as follows: 
\begin{equation}\label{eq:RSPwit}
	\begin{tikzpicture}
	\begin{pgfonlayer}{nodelayer}
		\node [style=none] (0) at (-2.25, -2.25) {};
		\node [style=none] (1) at (-0.25, -3.5) {};
		\node [style=none] (2) at (1.75, -2.25) {};
		\node [style=none] (3) at (-0.25, -2.75) {$\Phi^s$};
		\node [style=none] (4) at (-1.25, -2.25) {};
		\node [style=none] (5) at (0.75, -2.25) {};
		\node [style=small box] (6) at (-1.25, -0.75) {$U_{\bm{\uppsi}}$};
		\node [style=none] (7) at (0.75, 3) {};
		\node [style=right label] (8) at (-1.25, -1.75) {$\Q{2}$};
		\node [style=right label] (9) at (0.75, -1.75) {$\Q{2}$};
		\node [style=small box] (10) at (-1.25, 1) {$B$};
		\node [style=right label] (11) at (-1.25, 0) {$\Q{2}$};
		\node [style=none] (12) at (-0.25, 3) {};
		\node [style=right label] (13) at (-0.5, 2) {$\C{2}$};
		\node [style=none] (14) at (-1, 3) {};
		\node [style=none] (15) at (-0.5, 4) {};
		\node [style=none] (16) at (2, 4) {};
		\node [style=none] (17) at (2, 3) {};
		\node [style=none] (18) at (0.75, 3.5) {$\text{cUNOT}$};
		\node [style=none] (19) at (0.75, 4) {};
		\node [style=none] (20) at (0.75, 5.25) {};
		\node [style=right label] (21) at (0.75, 4.75) {$\Q{2}$};
	\end{pgfonlayer}
	\begin{pgfonlayer}{edgelayer}
		\draw (0.center) to (2.center);
		\draw (2.center) to (1.center);
		\draw (1.center) to (0.center);
		\draw [q] (4.center) to (6);
		\draw [q] (7.center) to (5.center);
		\draw [q] (10) to (6);
		\draw [c, in=90, out=-90] (12.center) to (10);
		\draw [in=180, out=0] (14.center) to (17.center);
		\draw (17.center) to (16.center);
		\draw (16.center) to (15.center);
		\draw (15.center) to (14.center);
		\draw [q] (20.center) to (19.center);
	\end{pgfonlayer}
\end{tikzpicture}
 = \begin{tikzpicture}
	\begin{pgfonlayer}{nodelayer}
		\node [style=point] (5) at (0.75, -0.75) {$\psi$};
		\node [style=none] (7) at (0.75, 1) {};
		\node [style=right label] (9) at (0.75, 0) {$\Q{2}$};
	\end{pgfonlayer}
	\begin{pgfonlayer}{edgelayer}
		\draw [q] (7.center) to (5);
	\end{pgfonlayer}
\end{tikzpicture}
\end{equation} 
where cUNOT is the controlled-universal-NOT operation.
The diagrammatic manipulations that prove that Eq.~\eqref{eq:RSPwit} holds are presented in Sec.~E of the Supplemental Material.   
Through this protocol, Witworld performs RSP of a qubit deterministically with the transmission of only one classical bit from Alice to Bob, outperforming quantum theory at the task. 

We now move on to unveiling what the critical resource is underlying the success of RSP in Witworld. For this, it is convenient to rewrite the diagram in the left hand side of Eq.~\eqref{eq:RSPwit} as: 

\begin{equation}\label{}
	\begin{aligned}
	\begin{tikzpicture}
	\begin{pgfonlayer}{nodelayer}
		\node [style=none] (0) at (-2.25, -2.25) {};
		\node [style=none] (1) at (-0.25, -3.5) {};
		\node [style=none] (2) at (1.75, -2.25) {};
		\node [style=none] (3) at (-0.25, -2.75) {$\Phi^s$};
		\node [style=none] (4) at (-1.25, -2.25) {};
		\node [style=none] (5) at (0.75, -2.25) {};
		\node [style=small box] (6) at (-1.25, -0.75) {$U_{\bm{\uppsi}}$};
		\node [style=none] (7) at (0.75, 3) {};
		\node [style=right label] (8) at (-1.25, -1.75) {$\Q{2}$};
		\node [style=right label] (9) at (0.75, -1.75) {$\Q{2}$};
		\node [style=small box] (10) at (-1.25, 1) {$B$};
		\node [style=right label] (11) at (-1.25, 0) {$\Q{2}$};
		\node [style=none] (12) at (-0.25, 3) {};
		\node [style=right label] (13) at (-0.5, 2) {$\C{2}$};
		\node [style=none] (14) at (-1, 3) {};
		\node [style=none] (15) at (-0.5, 4) {};
		\node [style=none] (16) at (2, 4) {};
		\node [style=none] (17) at (2, 3) {};
		\node [style=none] (18) at (0.75, 3.5) {$\text{cUNOT}$};
		\node [style=none] (19) at (0.75, 4) {};
		\node [style=none] (20) at (0.75, 5.25) {};
		\node [style=right label] (21) at (0.75, 4.75) {$\Q{2}$};
	\end{pgfonlayer}
	\begin{pgfonlayer}{edgelayer}
		\draw (0.center) to (2.center);
		\draw (2.center) to (1.center);
		\draw (1.center) to (0.center);
		\draw [q] (4.center) to (6);
		\draw [q] (7.center) to (5.center);
		\draw [q] (10) to (6);
		\draw [c, in=90, out=-90] (12.center) to (10);
		\draw [in=180, out=0] (14.center) to (17.center);
		\draw (17.center) to (16.center);
		\draw (16.center) to (15.center);
		\draw (15.center) to (14.center);
		\draw [q] (20.center) to (19.center);
	\end{pgfonlayer}
\end{tikzpicture}
 &= \begin{tikzpicture}
	\begin{pgfonlayer}{nodelayer}
		\node [style=none] (0) at (-2.25, -2.5) {};
		\node [style=none] (1) at (-0.25, -3.75) {};
		\node [style=none] (2) at (1.75, -2.5) {};
		\node [style=none] (3) at (-0.25, -3) {$\Phi^s$};
		\node [style=none] (4) at (-1.25, -2.5) {};
		\node [style=none] (5) at (0.75, -2.5) {};
		\node [style=none] (7) at (0.75, 3) {};
		\node [style=right label] (8) at (-1.25, -2) {$\Q{2}$};
		\node [style=right label] (9) at (0.75, -2) {$\Q{2}$};
		\node [style=small box] (10) at (-1.25, 1) {$B$};
		\node [style=white dot] (12) at (-1.25, 2) {};
		\node [style=right label] (13) at (-0.5, 2.5) {$\C{2}$};
		\node [style=none] (14) at (-1, 3) {};
		\node [style=none] (15) at (-0.5, 4) {};
		\node [style=none] (16) at (2, 4) {};
		\node [style=none] (17) at (2, 3) {};
		\node [style=none] (18) at (0.75, 3.5) {$\text{cUNOT}$};
		\node [style=none] (19) at (0.75, 4) {};
		\node [style=none] (20) at (0.75, 5.25) {};
		\node [style=right label] (21) at (0.75, 4.75) {$\Q{2}$};
		\node [style=point] (22) at (-3.5, -2.5) {$\bm{\uppsi}$};
		\node [style=none] (23) at (-0.25, 3) {};
		\node [style=none] (24) at (-2.5, 3) {};
		\node [style=none] (25) at (-2.5, 3) {};
		\node [style=upground] (26) at (-2.5, 3.25) {};
		\node [style=none] (27) at (-2.25, -0.25) {};
		\node [style=none] (28) at (-2.75, -1.25) {};
		\node [style=none] (29) at (-0.25, -1.25) {};
		\node [style=none] (30) at (-0.25, -0.25) {};
		\node [style=none] (31) at (-1.25, -0.25) {};
		\node [style=none] (32) at (-1.25, -1.25) {};
		\node [style=none] (33) at (-2, -1.25) {};
		\node [style=right label] (34) at (-1.25, 0) {$\Q{2}$};
		\node [style=none] (35) at (-1.25, -0.75) {$\text{c}U$};
	\end{pgfonlayer}
	\begin{pgfonlayer}{edgelayer}
		\draw (0.center) to (2.center);
		\draw (2.center) to (1.center);
		\draw (1.center) to (0.center);
		\draw [q] (7.center) to (5.center);
		\draw [c, in=90, out=-90] (12) to (10);
		\draw [in=180, out=0] (14.center) to (17.center);
		\draw (17.center) to (16.center);
		\draw (16.center) to (15.center);
		\draw (15.center) to (14.center);
		\draw [q] (20.center) to (19.center);
		\draw [c, in=-90, out=60] (12) to (23.center);
		\draw [c, in=120, out=-90] (24.center) to (12);
		\draw (28.center) to (29.center);
		\draw (29.center) to (30.center);
		\draw (30.center) to (27.center);
		\draw (27.center) to (28.center);
		\draw [q, in=270, out=90] (4.center) to (32.center);
		\draw [c,in=-90, out=90, looseness=0.50] (22) to (33.center);
		\draw [q, in=270, out=90] (31.center) to (10);
	\end{pgfonlayer}
\end{tikzpicture}
\\
		&= \sum_{a} \begin{tikzpicture}
	\begin{pgfonlayer}{nodelayer}
		\node [style=none] (0) at (-2.25, -2.5) {};
		\node [style=none] (1) at (-0.25, -3.75) {};
		\node [style=none] (2) at (1.75, -2.5) {};
		\node [style=none] (3) at (-0.25, -3) {$\Phi^s$};
		\node [style=none] (4) at (-1.25, -2.5) {};
		\node [style=none] (5) at (0.75, -2.5) {};
		\node [style=none] (7) at (0.75, 2.25) {};
		\node [style=right label] (8) at (-1.25, -2) {$\Q{2}$};
		\node [style=right label] (9) at (0.75, -2) {$\Q{2}$};
		\node [style=none] (10) at (-1.25, 0) {};
		\node [style=white dot] (12) at (-1.25, 1) {};
		\node [style=right label] (13) at (-1, 0.5) {$\C{2}$};
		\node [style=none] (14) at (-1, 2.25) {};
		\node [style=none] (15) at (-0.5, 3.25) {};
		\node [style=none] (16) at (2, 3.25) {};
		\node [style=none] (17) at (2, 2.25) {};
		\node [style=none] (18) at (0.75, 2.75) {$\text{cUNOT}$};
		\node [style=none] (19) at (0.75, 3.25) {};
		\node [style=none] (20) at (0.75, 4.5) {};
		\node [style=right label] (21) at (0.75, 4) {$\Q{2}$};
		\node [style=point] (22) at (-2.75, -1.75) {$\bm{\uppsi}$};
		\node [style=none] (23) at (-0.25, 2.25) {};
		\node [style=copoint] (24) at (-2.5, 2.25) {$a$};
		\node [style=none] (27) at (-2.5, 0) {};
		\node [style=none] (28) at (-3, -1) {};
		\node [style=none] (29) at (0, -1) {};
		\node [style=none] (30) at (0, 0) {};
		\node [style=none] (31) at (-1.25, 0) {};
		\node [style=none] (32) at (-1.25, -1) {};
		\node [style=none] (33) at (-2.25, -1) {};
		\node [style=none] (35) at (-1.25, -0.5) {$B\circ \text{c}U$};
	\end{pgfonlayer}
	\begin{pgfonlayer}{edgelayer}
		\draw (0.center) to (2.center);
		\draw (2.center) to (1.center);
		\draw (1.center) to (0.center);
		\draw [q] (7.center) to (5.center);
		\draw [c, in=90, out=-90] (12) to (10.center);
		\draw [in=180, out=0] (14.center) to (17.center);
		\draw (17.center) to (16.center);
		\draw (16.center) to (15.center);
		\draw (15.center) to (14.center);
		\draw [q] (20.center) to (19.center);
		\draw [c, in=-90, out=60] (12) to (23.center);
		\draw [c, in=120, out=-90, looseness=0.75] (24) to (12);
		\draw (28.center) to (29.center);
		\draw (29.center) to (30.center);
		\draw (30.center) to (27.center);
		\draw (27.center) to (28.center);
		\draw [q, in=270, out=90] (4.center) to (32.center);
		\draw [c, in=-90, out=90, looseness=0.50] (22) to (33.center);
		\draw [q, in=270, out=90] (31.center) to (10.center);
	\end{pgfonlayer}
\end{tikzpicture}
, \label{eq:thesum}
	\end{aligned}
\end{equation}
where
\begin{equation}\label{}
	\begin{tikzpicture}
	\begin{pgfonlayer}{nodelayer}
		\node [style=point] (0) at (0, -0.75) {$\bm{\uppsi}$};
		\node [style=none] (1) at (0, 1) {};
	\end{pgfonlayer}
	\begin{pgfonlayer}{edgelayer}
		\draw [c] (0) to (1.center);
	\end{pgfonlayer}
\end{tikzpicture}
\end{equation}
is not the state $ |\psi \rangle $, but simply a classical label corresponding to it, used to determine the unitary $U_{\psi}$ that the transformation $\text{c}U$ implements. In addition, $B\circ\text{c}U$ is the process that first implements the controlled unitary $\text{c}U$ and then the measurement $B$.

The crucial step here is to notice that each term in the sum in Eq.~\eqref{eq:thesum} can be identified with an element of an assemblage $\{\sigma_{a| \bm{\uppsi}}\}$ in an instrumental steering scenario (see Def.~D.4 in the Supplemental Material) as follows:
\beq
\begin{tikzpicture}
	\begin{pgfonlayer}{nodelayer}
		\node [style=none] (38) at (-1.5, 0) {};
		\node [style=none] (39) at (1.5, 0) {};
		\node [style=none] (40) at (0, -1.75) {};
		\node [style=none] (41) at (0, 0) {};
		\node [style=none] (42) at (0, 1) {};
		\node [style=none] (43) at (0, -0.5) {\small{$\sigma_{a|\bm{\uppsi}}$}};
	\end{pgfonlayer}
	\begin{pgfonlayer}{edgelayer}
		\draw (38.center) to (41.center);
		\draw (41.center) to (39.center);
		\draw (39.center) to (40.center);
		\draw [in=315, out=135] (40.center) to (38.center);
		\draw [q] (41.center) to (42.center);
	\end{pgfonlayer}
\end{tikzpicture}
 \ \ := \ \ \begin{tikzpicture}
	\begin{pgfonlayer}{nodelayer}
		\node [style=none] (0) at (-2.25, -2.5) {};
		\node [style=none] (1) at (-0.25, -3.75) {};
		\node [style=none] (2) at (1.75, -2.5) {};
		\node [style=none] (3) at (-0.25, -3) {$\Phi^s$};
		\node [style=none] (4) at (-1.25, -2.5) {};
		\node [style=none] (5) at (0.75, -2.5) {};
		\node [style=none] (7) at (0.75, 2.25) {};
		\node [style=right label] (8) at (-1.25, -2) {$\Q{2}$};
		\node [style=right label] (9) at (0.75, -2) {$\Q{2}$};
		\node [style=none] (10) at (-1.25, 0) {};
		\node [style=white dot] (12) at (-1.25, 1) {};
		\node [style=right label] (13) at (-1, 0.5) {$\C{2}$};
		\node [style=none] (14) at (-1, 2.25) {};
		\node [style=none] (15) at (-0.5, 3.25) {};
		\node [style=none] (16) at (2, 3.25) {};
		\node [style=none] (17) at (2, 2.25) {};
		\node [style=none] (18) at (0.75, 2.75) {$\text{cUNOT}$};
		\node [style=none] (19) at (0.75, 3.25) {};
		\node [style=none] (20) at (0.75, 4.5) {};
		\node [style=right label] (21) at (0.75, 4) {$\Q{2}$};
		\node [style=point] (22) at (-2.75, -1.75) {$\bm{\uppsi}$};
		\node [style=none] (23) at (-0.25, 2.25) {};
		\node [style=copoint] (24) at (-2.5, 2.25) {$a$};
		\node [style=none] (27) at (-2.5, 0) {};
		\node [style=none] (28) at (-3, -1) {};
		\node [style=none] (29) at (0, -1) {};
		\node [style=none] (30) at (0, 0) {};
		\node [style=none] (31) at (-1.25, 0) {};
		\node [style=none] (32) at (-1.25, -1) {};
		\node [style=none] (33) at (-2.25, -1) {};
		\node [style=none] (35) at (-1.25, -0.5) {$B\circ \text{c}U$};
	\end{pgfonlayer}
	\begin{pgfonlayer}{edgelayer}
		\draw (0.center) to (2.center);
		\draw (2.center) to (1.center);
		\draw (1.center) to (0.center);
		\draw [q] (7.center) to (5.center);
		\draw [c, in=90, out=-90] (12) to (10.center);
		\draw [in=180, out=0] (14.center) to (17.center);
		\draw (17.center) to (16.center);
		\draw (16.center) to (15.center);
		\draw (15.center) to (14.center);
		\draw [q] (20.center) to (19.center);
		\draw [c, in=-90, out=60] (12) to (23.center);
		\draw [c, in=120, out=-90, looseness=0.75] (24) to (12);
		\draw (28.center) to (29.center);
		\draw (29.center) to (30.center);
		\draw (30.center) to (27.center);
		\draw (27.center) to (28.center);
		\draw [q, in=270, out=90] (4.center) to (32.center);
		\draw [c, in=-90, out=90, looseness=0.50] (22) to (33.center);
		\draw [q, in=270, out=90] (31.center) to (10.center);
	\end{pgfonlayer}
\end{tikzpicture}
,
\eeq
where $a$ denotes Alice's dichotomic outcome, and $\bm{\uppsi}$ is the classical variable that denotes her measurement choice. That is, RSP is ultimately an instance of an instrumental steering scenario, and the possible assemblages that Alice can prepare dictates whether RSP is possible for the given cardinality of $a$. 
For the particular RSP protocol discussed above, 
\beq
\begin{tikzpicture}
	\begin{pgfonlayer}{nodelayer}
		\node [style=none] (38) at (-1.5, 0) {};
		\node [style=none] (39) at (1.5, 0) {};
		\node [style=none] (40) at (0, -1.75) {};
		\node [style=none] (41) at (0, 0) {};
		\node [style=none] (42) at (0, 1) {};
		\node [style=none] (43) at (0, -0.5) {\small{$\sigma_{a|\bm{\uppsi}}$}};
	\end{pgfonlayer}
	\begin{pgfonlayer}{edgelayer}
		\draw (38.center) to (41.center);
		\draw (41.center) to (39.center);
		\draw (39.center) to (40.center);
		\draw [in=315, out=135] (40.center) to (38.center);
		\draw [q] (41.center) to (42.center);
	\end{pgfonlayer}
\end{tikzpicture}
 \ \ := \ \ \frac{1}{2} \, \begin{tikzpicture}
	\begin{pgfonlayer}{nodelayer}
		\node [style=point] (5) at (0.75, -0.75) {$\psi$};
		\node [style=none] (7) at (0.75, 1) {};
		\node [style=right label] (9) at (0.75, 0) {$\Q{2}$};
	\end{pgfonlayer}
	\begin{pgfonlayer}{edgelayer}
		\draw [q] (7.center) to (5);
	\end{pgfonlayer}
\end{tikzpicture}
\,.
\eeq
It is readily seen that the assemblage $\{\sigma_{a| \bm{\uppsi}}\}$ has no quantum realisation: if this was instead the case, this assemblage would provide a quantum RSP protocol that succeeds deterministically with 1 bit of communication, which is fundamentally impossible. 

We see therefore how instrumental steering powers RSP, and how the post-quantum steering featured in Witworld makes this theory more efficient than quantum theory at the task of Remote State Preparation.

Let us observe that quantum theory restricted to the reals \cite{hardy2012limited}, which has mixed states given  by symmetric matrices (a subset of quantum states), also requires a single bit of communication for RSP. A rebit (2 dimensional real quantum system) has mixed states given by the $X-Z$ plane of the Bloch sphere (a disk). The universal NOT is just rotation by $\pi$ around the $Y$ axis, and is completely positive. Since the singlet state $\ketbra{\Phi^{s}}{\Phi^{s}}$ is a real valued density operator (i.e. it is a symmetric matrix) it follows that it is a valid entangled state of two rebits. Hence the protocol outlined above in Witworld can be applied to real quantum theory as well, to give RSP with a single bit of communication.

\section{Discussion}

\bel{In this work we explored the scope of post-quantum steering as a stronger-than-quantum resource for information processing. We particularly focused the search on tasks beyond device-independent ones or those that ultimately rely on Bell correlations (such as random access codes \cite{RAC1,RAC2} or device-independent quantum key distribution): we aimed at finding tasks that intrinsically leveraged quantum systems and non-classical steering. We discovered that remote state preparation of qubits systems provides a friendly proof-of-principle of a general phenomenology: steering assemblages in the instrumental scenario serve as a resource for the task, and post-quantum assemblages perform better than quantum ones at it. This is the first time that post-quantum steering  -- as opposed to post-quantum Bell non-classicality -- has been identified as a resource powering information processing which can provably outperform quantum theory. }

In order to prove our claims, we defined a  generalised probabilistic theory, that we call Witworld, by combining classical, quantum, and Boxworld systems in a simple mathematical way, via the max tensor product. The intuitive formulation of Witworld allowed us to present its  powerful post-quantum features in an accessible way: one can readily see how post-quantum Bell nonlocality, post-quantum steering, and post-quantum states emerge within Witworld. \bel{The task of remote state preparation can be studied diagrammatically within Witworld, and by doing so we showed how the post-quantum assemblages allowed by the theory makes Witworld perform better at it than quantum theory does. }

\bel{A feature of Witworld is that, even though it is built in part on quantum systems, it does not contain quantum theory as a subtheory: there are tasks, such as quantum teleportation, that quantum theory can perform whilst Witworld cannot. The reason for this is the choice of composition rule: Witworld composes via the max tensor product, and hence no entangling measurements are allowed in the theory. Nonetheless, Witworld remarkably succeeds at reproducing all the quantum entangled states, quantum steering assemblages, and quantum correlations in Bell scenarios. That is, for the non-classical phenomena usually leveraged in quantum information tasks, Witworld is \bel{at least} as good as quantum theory at \bel{manifesting} them.    }

If we turn our attention to a particular subtheory of Witworld by restricting the system types to classical and quantum only -- that is, by removing Boxworld from the theory -- we find that this subtheory still features post-quantum properties, such as Bell non-classicality in multipartite scenarios (for example, by utilizing the results of Ref.~\cite{toniOform}), as well as post-quantum steering and post-quantum states even in bipartite scenarios. \bel{Remarkably, the post-quantum advantage for remote state preparation is also featured by this subtheory of Witworld, since the post-quantum advantage provided by it stems from the enlarged set of operations allowed on local quantum systems. We leave it as an open question whether other previously defined GPTs (e.g., Refs.~\cite{barnum2020composites,hefford2020hyper}) may provide such an advantage for this information processing task. }

\bel{It is worth noticing that Witworld's simple formulation \bel{does not make} the theory intrinsically groundbreaking from the perspective of generalised probabilitic theories, however its relevance \bel{is not grounded in its appeal as a standalone GPT}. \bel{Rather}, Witworld shows that there exists a compositional theory that could underpin post-quantum effects such as post-quantum steering. This shows that the latter phenomenon in not \bel{in principle} \bel{ un-realisable}, and hence its study should not be simply dismissed.  }

Looking ahead, there \bel{are} a variety of open questions that can be studied, \bel{ especially about the extent to which post-quantum steering compatible with special relativity \bel{can} be underpinned by some generalised theory.} We know that Witworld may display post-quantum steering but, unlike the case of Bell non-classicality, it is still unknown whether any no-signalling assemblage may have a realisation within Witworld. Any answer to this question would be of interest: if Witworld can realise all no-signalling assemblages, then this theory becomes the first GPT to accommodate steering fully in a common-cause resource theoretic framework \cite{schmid2020postquantum}; otherwise,  understanding the reason behind the gap between no-signalling realisable and Witworld realisable assemblages \bel{may lead to}  an operational principle that could shed light on the characterisation of quantum phenomena.

Finally, the exploration of the information processing power of \bel{steering (quantum and beyond)} has only just begun. Since Witworld is formulated in an intuitive way leveraging a diagrammatic representation \cite{chiribella2010probabilistic,
hardy2011reformulating,coecke2018picturing,gogioso2017categorical,selby2018reconstructing}, there is plenty of scope for investigating other post-quantum advantages of this theory, and of post-quantum steering, for information processing. 

\bigskip
\section*{Data Availability}

Data sharing not applicable to this article as no datasets were generated or analysed during the current study.

\section*{Code Availability}

Code sharing not applicable to this article as no codes were used during the current study.

\section*{Acknowledgments}

We thank an anonymous referee for their useful feedback on our manuscript. 
PJC, JHS, and ABS acknowledge support by the Foundation for Polish Science (IRAP project, ICTQT, contract no.2018/MAB/5, co-financed by EU within Smart Growth Operational Programme). 
This research was supported by Perimeter Institute for Theoretical Physics. Research at Perimeter Institute is supported by the Government of Canada through the Department of Innovation, Science and Economic Development Canada and by the Province of Ontario through the Ministry of Research, Innovation and Science.
All of the diagrams within this manuscript were prepared using TikZit.

\section*{Author Contributions}

All authors contributed equally to the writing and editing of the manuscript. TDG, ABS, JHS, JS conceived of this project during the PIMan workshop at Chapman University in 2018. The bulk of the results were formally proven by PJC under the supervision of ABS and JHS and in collaboration with TDG and JS. Diagrams were created by PJC and JHS using TikZiT. Figures were created by PJC and ABS.

\section*{Competing interests}

The authors declare no competing financial or non-financial interests.

\section*{Figure Captions}

\begin{figure}[ht]
	\centering
	\includegraphics[width=0.5\columnwidth]{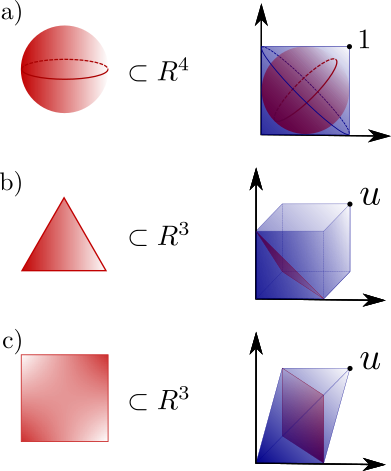}
	\caption{The geometry of the set of states (to the left) and effects (to the right) for a) atomic quantum systems of dimension $d=2$, b) atomic classical systems of dimension $v=3$, and c) atomic Boxworld systems of dimension $(n,k)=(2,2)$.}
	\label{geometry}
\end{figure}

\begin{figure}
\begin{center}
			\includegraphics[width=\columnwidth,trim={2cm 0 2cm 0}, clip]{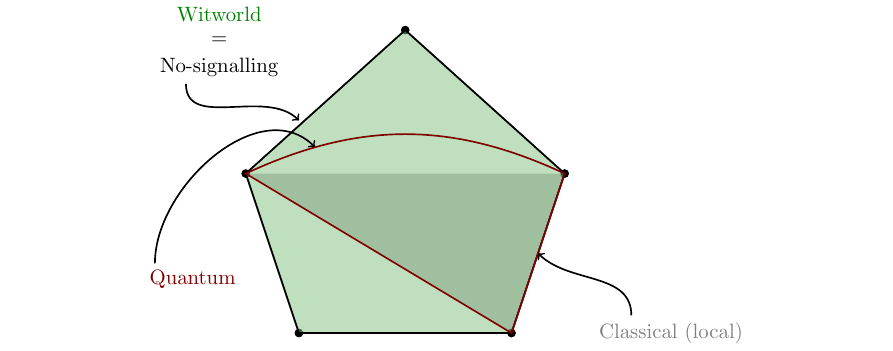}
\end{center}
\caption{Pictorial representation of the sets of correlations in Bell scenarios. The polytope of classical (Local Hidden State) correlations is depicted by the gray triangle. The convex set of quantum correlations is depicted by the red-bordered convex region. The set of Witworld correlations is depicted by the green transparent region. The polytope of no-signalling correlations is depicted by the black-bordered pentagon.  The sets of no-signalling and Witworld correlations are equivalent. Witworld correlations strictly contain the quantum set, which strictly contains the classical set.}
\label{fig:NSpol}
\end{figure}

\begin{figure}
\begin{center}
			\includegraphics[width=\columnwidth,trim={2cm 0 2cm 0}, clip]{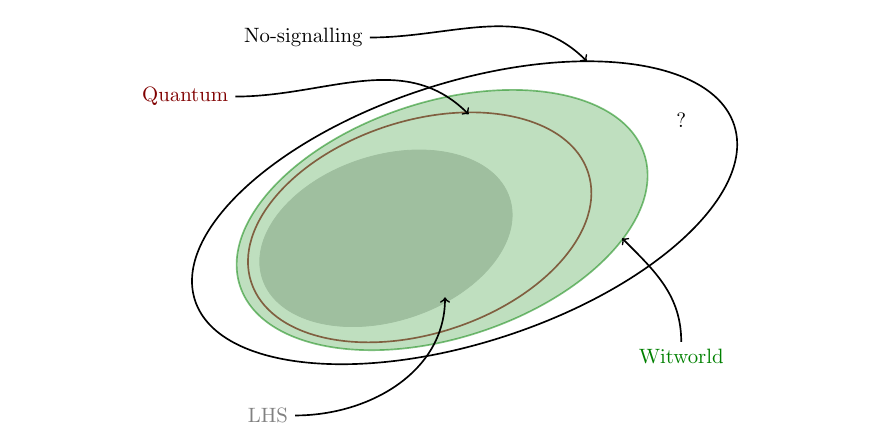}
\end{center}
\caption{Pictorial representation of the sets of assemblages in steering scenarios. The set of classical (Local Hidden State) assemblages is depicted by the gray ellipse. The set of quantum assemblages is depicted by the red-bordered region. The set of Witworld assemblages is depicted by the green transparent region. The set of no-signalling assemblages is depicted by the black-bordered ellipse. The set of Witworld assemblages strictly contains the quantum set, which strictly contains the LHS set. The set of no-signalling assemblages contains the Witworld set, and an open question  \bel{(pictorially depicted by the question mark)} is whether this inclusion is strict. }
\label{fig:SteeSet}
\end{figure}

\begin{figure}[]
	\centering
	\includegraphics[width=\columnwidth]{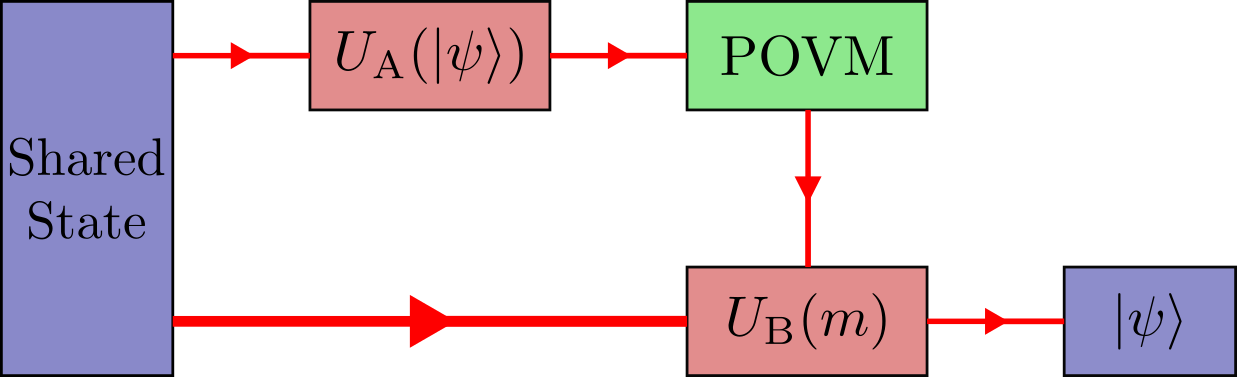}
	\caption{Pictorial representation of the remote state preparation protocol. Alice performs $ U_{A} ( | \psi \rangle  ) $ on her share of a physical system --  a unitary operation that depends on $\ket{\psi}$ -- and then a  measurement (POVM). She sends a classical message $m$ to Bob, who, in turn, performs an unitary transformation (which depends on $m$) on his share of the system. The outcome of the protocol is a quantum state on state $\ket{\psi}$ on Bob's lab.}
	\label{rspfig}
\end{figure}

\bibliography{bibliography}

\appendix

\begin{center}
\vskip 1cm
\Large{Supplemental Material}
\end{center}

\section{Generalised Probabilistic Theories}
\label{ap:GPT}
\subsection{The basics}\label{se:thebasics}

In this section, we provide a description of generalised probabilistic theories (GPTs) that have the property of being locally tomographic \cite{hardy2001quantum}. Roughly speaking, the GPT framework is a general framework to formulate and describe theories (including quantum theory), that allows for the calculation of probabilities of measurement outcomes when system preparations have states associated to them. In this appendix, we aim for a description of GPTs which connects the diagrammatic \cite{chiribella2010probabilistic,
hardy2011reformulating} and the linear algebraic notations, hoping to make it useful to a broader audience. Although a more general kind of GPT could be defined, for the purpose of defining Witworld, restricting the present discussion to locally tomographic GPTs significantly simplifies the task in hand. The interested reader can find a more general definition of GPTs in, for example, Ref.~\cite{chiribella2010probabilistic}. 

In general, a GPT consists of collections of states that are associated to different system types, a rule for combining state spaces
of simple systems into state spaces of composite systems, collections of allowed transformations between these states, and effects  -- i.e., functions that associate probabilities to each outcome of each measurement for each state in the theory.
Here, as mentioned above, we focus on locally tomographic GPTs, which are those where the states of composite systems can be uniquely determined by the information given by local measurements on its parts. Each of these ingredients are defined and compared to quantum theory in what follows.

We start with the states. For each system $A$ of a GPT, there is a vector space $V^{A}$ associated to it. A convex subset $ \Omega^{A} \subset V^{A}$, called the state space,
 defines the allowed states of the system $A$. This subset has dimension $\mathsf{dim}( \Omega^{A} ) = \mathsf{dim} (V^{A} ) - 1$. We require every state in $ \Omega^{A}$ to be normalised in a sense to be defined later in this section when we introduce effects. The convexity property means that, if $ \sigma \in \Omega^{A} $ and $ \rho \in \Omega^{A}$, then $ p \sigma + (1-p) \rho  \in \Omega^{A}$ for any $p \in [0,1]$. We require convexity so that the GPT accommodates statistical mixtures of state preparations in a natural way. Diagrammatically, the system $A$ and its associated vector space $ V^{A}$ are represented by a labelled wire: 
\begin{equation}\label{eq:thewire}
\begin{tikzpicture}
	\begin{pgfonlayer}{nodelayer}
		\node [style=none] (0) at (0, 0.75) {};
		\node [style=none] (1) at (0, -0.75) {};
		\node [style={right label}] (2) at (0, -0.5) {$A$};
	\end{pgfonlayer}
	\begin{pgfonlayer}{edgelayer}
		\draw (0.center) to (1.center);
	\end{pgfonlayer}
\end{tikzpicture}}.
\end{equation}
\bel{A remark on notation is in order: throughout this section (i.e., Sec.~\ref{ap:GPT}), we review definitions and properties of a general class of GPTs,which  Witworld belongs to, but we do not restrict the presentation to the latter. Hence, the wire type in Eq.~\eqref{eq:thewire} should here be understood as a generic system rather than as a classical system in Witworld (as per Eq.~\eqref{eq:thewireswit}). From Sec.~\ref{app:WW} we will shift the focus back to Witworld, and hence the notation from Eq.~\eqref{eq:thewireswit} will take precedence again.} 
In the case of quantum theory, \bel{hence}, the wires \bel{in Eq.~\eqref{eq:thewire}} represent real vector spaces of Hermitian operators on Hilbert spaces.
For instance, if $A$ is a qubit system, the wire labelled by $A$ represents the vector space $ V^{A} =\{ O \in \mathcal{L} ( \mathds{C}^{2}, \mathds{C}^{2} ) : O = O^{\dagger} \}$, where $ \mathcal{L} ( \mathds{C}^{2}, \mathds{C}^{2} )$ is the space of linear operators on $ \mathds{C}^{2}$. Then, $ \Omega^{A}$ is the set of positive operators whose trace is $1$, that is, $ \Omega^{A} = \{ \rho \in V^{A} : \rho \geq 0, \tr{\rho} = 1 \}$, which is indeed a convex set as required. It is sometimes useful, moreover, to include within the GPT formulation of quantum theory some wires that represent classical variables which store the results of measurements, see, for example, Refs.~\cite{coecke2016categorical,gogioso2017categorical,
selby2018reconstructing}. This is done in section \ref{se:ww2} of the present work.

As mentioned previously, one can construct composite systems by the combination of simpler systems. We denote by $ A \cdot B$ the system composed by a system $A$  and a system $B$. Its states belong to the set $ \Omega^{A \cdot B} \subset V^{A \cdot B} $, which is represented diagrammatically by multiple wires side by side:
\begin{equation}\label{}
\begin{tikzpicture}
	\begin{pgfonlayer}{nodelayer}
		\node [style=none] (0) at (0, 0.75) {};
		\node [style=none] (1) at (0, -0.75) {};
		\node [style={right label}] (2) at (0, -0.5) {$A\cdot B$};
	\end{pgfonlayer}
	\begin{pgfonlayer}{edgelayer}
		\draw (0.center) to (1.center);
	\end{pgfonlayer}
\end{tikzpicture}} = %
\begin{tikzpicture}
	\begin{pgfonlayer}{nodelayer}
		\node [style=none] (0) at (0, 0.75) {};
		\node [style=none] (1) at (0, -0.75) {};
		\node [style={right label}] (2) at (0, -0.5) {$A$};
	\end{pgfonlayer}
	\begin{pgfonlayer}{edgelayer}
		\draw (0.center) to (1.center);
	\end{pgfonlayer}
\end{tikzpicture}} %
\begin{tikzpicture}
	\begin{pgfonlayer}{nodelayer}
		\node [style=none] (0) at (0, 0.75) {};
		\node [style=none] (1) at (0, -0.75) {};
		\node [style={right label}] (2) at (0, -0.5) {$B$};
	\end{pgfonlayer}
	\begin{pgfonlayer}{edgelayer}
		\draw (0.center) to (1.center);
	\end{pgfonlayer}
\end{tikzpicture}}.
\end{equation}
In the locally tomographic GPTs that we consider here, such as quantum theory, $ V^{A \cdot B} = V^{A} \otimes V^{B}$. For convenience,we omit the label when the exact system being discussed is not relevant or it is clear from the context, or we use different kinds of wires to highlight the distinctions, as we do in Sec.~\ref{se:ww2}. If we want to refer diagrammatically to a specific state of $A$, that is, some element $s$ of $ \Omega^{A}$, we use a box (usually, but not necessarily, a triangular box) with an output wire $A$ connected to its top:
\begin{equation}\label{}
	s \ \ {\equiv} \ \ %
\begin{tikzpicture}
	\begin{pgfonlayer}{nodelayer}
		\node [style=none] (0) at (0, 0.75) {};
		\node [style=point] (1) at (0, -0.75) {$s$};
		\node [style={right label}] (2) at (0, 0.5) {$A$};
	\end{pgfonlayer}
	\begin{pgfonlayer}{edgelayer}
		\draw (0.center) to (1);
	\end{pgfonlayer}
\end{tikzpicture}
}.
\end{equation}
The transformations in a (tomographically local) GPT are linear functions from the vector space $V^A$ associated to a system of some type $A$ to the vector space $V^B$ associated to some system of type $B$. Hence, the set of transformations of type $A \to B$, denoted by $ \mathcal{T}^{A \to B}$ is a subset of $ \mathcal{L} ( V^{A} , V^{B} )$, the set of linear transformations from $V^{A}$ to $V^{B}$. Diagrammatically, a particular transformation $T \in \mathcal{T}^{A \to B}$ is denoted by a box with an input wire $A$ connected to its bottom and an output wire $B$ connected to its top:

\begin{equation}\label{}
	T  \ \ {\equiv} \ \  %
\InputIfFileExists{Diagrams/transAB.tikz}{}{\input{./figures/Diagrams/transAB.tikz}}.
\end{equation}
For quantum theory, the set of transformations $ \mathcal{T}^{A \to B}$ is the set of quantum operations, which correspond to completely positive trace-non-increasing maps from $V^{A}$ to $V^{B}$.

Of course, we may want to represent not just the transformation itself, but its action on a specific state. This is done by connecting the input wire of the transformation with a state of matching type:

\begin{equation}\label{}
	T(s) \ \ {\equiv} \ \ %
\InputIfFileExists{Diagrams/atransA.tikz}{}{\input{./figures/Diagrams/atransA.tikz}} \,.
\end{equation}
Note that with this, viewing  $T(s)$ as a vector can be expressed in diagrams by ``sliding'' the box representing $T$ until it merges with the box representing s:
\begin{equation}\label{}
	T \circ s = T(s) \ \ {\equiv} \ \ %
\InputIfFileExists{Diagrams/atransA.tikz}{}{\input{./figures/Diagrams/atransA.tikz}} = %
\begin{tikzpicture}
	\begin{pgfonlayer}{nodelayer}
		\node [style=point] (5) at (0, -0.5) {$T(s)$};
		\node [style=none] (6) at (0, 1) {};
		\node [style=right label] (7) at (0, 0.75) {$B$};
	\end{pgfonlayer}
	\begin{pgfonlayer}{edgelayer}
		\draw (5) to (6.center);
	\end{pgfonlayer}
\end{tikzpicture}
}.
\end{equation}
This is a manipulation of diagrams that is used often in this work. The converse operation, where we split a vector into a product where a transformation $T$ is applied on a state $s$, is also a valid manipulation where we split a diagram with only a state into one where a transformation is connected to a different state. Notice that this mirrors exactly linear algebraic operations where an equation like $ s^\prime  = T \circ s$ is used for substitutions. Furthermore, boxes representing transformations can also be connected, when the wire types match, to represent the sequential composition of them. Because a sequence of linear transformations $T$ and $U$, can also be viewed as a single transformation $U\circ T$, the composition of both, the merging of boxes can also be done with transformations that are connected:

\begin{equation}\label{}
	U \circ T \ \ {\equiv} \ \ %
\InputIfFileExists{Diagrams/transAC.tikz}{}{\input{./figures/Diagrams/transAC.tikz}}\ \ =\ \  %
\InputIfFileExists{Diagrams/transAC2.tikz}{}{\input{./figures/Diagrams/transAC2.tikz}}.
\end{equation}
Naturally, one may need to represent transformations that happen in parallel on the parts of a composite system $A \cdot B$. While in linear algebraic notation this is done with the direct product $ \otimes$, so that, for $T$ of type $A \to B$ and $V $ of type $C \to D$, we write $T \otimes V$, in diagrammatic notation we simply put $T$ and $V$ side by side:
\begin{equation}\label{}
	T \otimes V \ \ {\equiv} \ \ %
\InputIfFileExists{Diagrams/transAB.tikz}{}{\input{./figures/Diagrams/transAB.tikz}}%
\InputIfFileExists{Diagrams/transCD.tikz}{}{\input{./figures/Diagrams/transCD.tikz}}\ \ = \ \ %
\InputIfFileExists{Diagrams/transComp.tikz}{}{\input{./figures/Diagrams/transComp.tikz}},
\end{equation}
where the order of the wires (from left to right) matters just like the order of the product $T \otimes V$.

The effects of a system $A$ in a GPT are linear functionals over $V^{A}$ that evaluate to probabilities, i.e., numbers in $[0,1]$, for every valid state. This means that the set $E^{A}$ of effects of a system $A$ is a subset of $(V^{A})^* $, the dual of $V^{A}$, such that $e \in E^{A}$ implies $e(s) \in [0,1]$ for every $s \in \Omega^{A}$. Diagrammatically, the effects are represented as boxes with only inputs, so

\begin{equation}\label{}
	e \ \ {\equiv} \ \ %
\begin{tikzpicture}
	\begin{pgfonlayer}{nodelayer}
		\node [style=none] (0) at (0, -0.75) {};
		\node [style=copoint] (1) at (0, 0.75) {$e$};
		\node [style={right label}] (2) at (0, -0.5) {$A$};
	\end{pgfonlayer}
	\begin{pgfonlayer}{edgelayer}
		\draw (0.center) to (1);
	\end{pgfonlayer}
\end{tikzpicture}
}
\end{equation}
represents the element $e \in E^{A} \subset (V^{A} )^* $. Just like the linear functions, the action of $e$ on a state $s$, that is, $e(s)$, is given by connecting the input (bottom) wire of the effect to the output (top) wire of the state:

\begin{equation}\label{effectnumber}
	e(s)\ \  =\ \  %
\begin{tikzpicture}
	\begin{pgfonlayer}{nodelayer}
		\node [style=point] (0) at (0, -0.75) {$s$};
		\node [style=copoint] (1) at (0, 0.75) {$e$};
		\node [style=right label] (2) at (0, 0) {$A$};
	\end{pgfonlayer}
	\begin{pgfonlayer}{edgelayer}
		\draw [in=90, out=-90] (1) to (0);
	\end{pgfonlayer}
\end{tikzpicture}
}.
\end{equation}
Using quantum theory again as an example, its effects are trace inner products with the elements of a positive operator-valued measure (POVM). That is to say that the POVM elements give the Riesz representation of the effects of the theory. So, if $M$ is associated to the effect $e_{M}$ through the Riesz representation, then $ e_{M} ( \rho ) = \tr(M \rho )$.

Notice that Diagram \eqref{effectnumber}, unlike those in the previous examples, contains no loose wires. This means that $e(s)$ is a real number, and, similarly, any diagram in this formalism without loose wires represents a real number. Moreover, diagrams with only output (top) loose wires are always states, diagrams with only input (bottom) loose wires are always effects, and diagrams with both input and output loose wires are transformations.

Now that we have discussed the effects, we can define what it means for a state to be normalised in a GPT. This is done through a special effect, called the unit effect, which, for a system $A$, we denote by $u^{A} $. We say that a vector $s \in V^{A}$ is normalised if and only if $u^{A} (s) = 1$. Therefore, by our definition of the set of states $ \Omega^{A} $, if $s \in \Omega^{A}$, then $u^{A}(s) = 1$. This special effect is denoted by a special diagram
\begin{equation}\label{theunitefect}
\begin{tikzpicture}
	\begin{pgfonlayer}{nodelayer}
		\node [style=none] (0) at (0, -0.75) {};
		\node [style=none] (1) at (0, 0.25) {};
		\node [style=right label] (2) at (0, -0.25) {$A$};
		\node [style=upground] (3) at (0, 0.5) {};
	\end{pgfonlayer}
	\begin{pgfonlayer}{edgelayer}
		\draw [in=90, out=-90] (1.center) to (0.center);
	\end{pgfonlayer}
\end{tikzpicture}
},
\end{equation}
thus, diagrammatically, state normalisation is captured by the condition
\beq
\begin{tikzpicture}
	\begin{pgfonlayer}{nodelayer}
		\node [style=point] (0) at (0, -0.75) {$s$};
		\node [style=none] (1) at (0, 0.5) {};
		\node [style=right label] (2) at (0, 0) {$A$};
		\node [style=upground] (3) at (0, 0.75) {};
	\end{pgfonlayer}
	\begin{pgfonlayer}{edgelayer}
		\draw [in=90, out=-90] (1.center) to (0);
	\end{pgfonlayer}
\end{tikzpicture}
} \ \ = \ 1.
\eeq
In the example of quantum theory, the unit effect of any system type is the trace operation, or, in other words, the trace inner product with the identity operator.

An important definition to be made, and that appears nicely in diagrammatic notation, is that of separable effects and states. In quantum theory, a separable state is that which can be written as a convex combination of product quantum states. Here, we simply generalise that notion to any GPT state: $ s \in \Omega^{A \cdot B}$ is separable if $ s = \sum_{i} p_{i} r^{A}_{i}  \otimes r^{B}_{i} $, with $p_{i} \in [0,1]$ and $\sum_{i} p_{i} = 1$, $r_i^A \in \Omega^A$, and $r_i^B \in \Omega^B$. In diagrammatic notation, a separable state can be viewed as: 

\begin{equation}\label{}
	s = \sum_{i} p_{i} r^{A}_{i} \otimes r^{B}_{i}  \ \ {\equiv} \ \  \sum_{i} p_{i}\ %
\begin{tikzpicture}
	\begin{pgfonlayer}{nodelayer}
		\node [style=none] (0) at (0, 1) {};
		\node [style=point] (1) at (0, -0.5) {$r_{i}^{A}$};
		\node [style=right label] (2) at (0, 0.5) {$A$};
	\end{pgfonlayer}
	\begin{pgfonlayer}{edgelayer}
		\draw (0.center) to (1);
	\end{pgfonlayer}
\end{tikzpicture}
}%
\begin{tikzpicture}
	\begin{pgfonlayer}{nodelayer}
		\node [style=none] (0) at (0, 1) {};
		\node [style=point] (1) at (0, -0.5) {$r_{i}^{B}$};
		\node [style=right label] (2) at (0, 0.5) {$B$};
	\end{pgfonlayer}
	\begin{pgfonlayer}{edgelayer}
		\draw (0.center) to (1);
	\end{pgfonlayer}
\end{tikzpicture}
} \ \ = \ \ %
\InputIfFileExists{Diagrams/stateAB.tikz}{}{\input{./figures/Diagrams/stateAB.tikz}} .
\end{equation}
Separable effects are defined similarly, with its diagrammatic representation being like the one above but where the loose wires come from the bottom instead of the top.

\subsection{Further definitions}

We can use these fundamental notions to define some concepts that are necessary in this work. These are positive cones, positive vectors, local tomography, the no-restriction hypothesis, the generalised no-restriction hypothesis, the maximal tensor product, trace-preserving transformation, and positive and completely positive transformations.  Some of these are present in quantum theory, but here we need definitions that generalise them to arbitrary GPTs.

\begin{defn}[Positive Cone]
A positive cone $X_{+}$ generated by a subset $X$ of a vector space $V$ is the set of nonnegative multiples of the elements of X. That is,
\begin{equation}\label{}
	X_{+} = \{\lambda x: \lambda \geq 0, x \in X\}.
\end{equation}
\end{defn}

\begin{defn}[Positive Vector]\label{posvector}
A vector $ v$ of a vector space $V^{A}$ associated to a system of $A$ of a GPT is said to be positive, denoted $v \geq 0$, if $ v \in \Omega^{A}_{+}$, the cone generated by the set of states.
\end{defn}
Note that in quantum theory our notion of positive cones recovers the sets of positive operators from the sets of density matrices. This is because, for any quantum system $A$, the density matrices $\rho \in \Omega^{A}$ satisfy $ \rho \geq 0$  and $ \tr(\rho ) = 1$, so by multiplying then by positive numbers $ \lambda$, we are simply dropping the unit trace assumption. Hence, the cone generated is $ \Omega_{+}^{A} = \{ \rho \in V^{A} : \rho \geq 0 \}$, that is, the set of positive operators on the Hilbert space corresponding to $A$. Conversely, using the unit effects, it is always possible to recover the states from the positive cones by restricting them to normalised vectors.

\begin{defn}[Trace Non-Increasing Operation] A transformation $T$ in $ \mathcal{L} (V^{A} , V^{B} )$ is said to be trace non-increasing if $u^{B} (T (s)) \leq u^{A} (s)$ for every $s$ in $\Omega^{A} $.

\end{defn}
Diagrammatically, $T$ is trace non-increasing means, for all $s\in \Omega^A$, that:
\beq
\InputIfFileExists{Diagrams/TNI1.tikz}{}{\input{./figures/Diagrams/TNI1.tikz}}\ \ \leq\ \  %
\begin{tikzpicture}
	\begin{pgfonlayer}{nodelayer}
		\node [style=point] (0) at (0, -0.75) {$s$};
		\node [style=none] (1) at (0, 0.5) {};
		\node [style=upground] (3) at (0, 0.75) {};
		\node [style=right label] (4) at (0, 0) {$A$};
	\end{pgfonlayer}
	\begin{pgfonlayer}{edgelayer}
		\draw (1.center) to (0);
	\end{pgfonlayer}
\end{tikzpicture}
} \, .
\eeq

\begin{defn}[Trace-Preserving Operation] A transformation $T \in \mathcal{L} (V^{A} , V^{B} )$ is said to be trace preserving if $ \forall s \in \Omega^{A}, u^{B} (T(s)) = u^{A} (s) $. 
\end{defn}
Diagrammatically, $T$ is trace preserving means, for all $s\in \Omega^A$, that:
\beq
\InputIfFileExists{Diagrams/TNI1.tikz}{}{\input{./figures/Diagrams/TNI1.tikz}}\ \ =\ \ %
\begin{tikzpicture}
	\begin{pgfonlayer}{nodelayer}
		\node [style=point] (0) at (0, -0.75) {$s$};
		\node [style=none] (1) at (0, 0.5) {};
		\node [style=upground] (3) at (0, 0.75) {};
		\node [style=right label] (4) at (0, 0) {$A$};
	\end{pgfonlayer}
	\begin{pgfonlayer}{edgelayer}
		\draw (1.center) to (0);
	\end{pgfonlayer}
\end{tikzpicture}
} \, . 
\eeq

\begin{defn}[Positive Transformation] A transformation $T \in \mathcal{L} (V^{A} , V^{B} )$ is said to be positive if it takes elements in $ \Omega^{A}_{+} $ to elements in $ \Omega^{B}_{+}$, that is, $ s \in \Omega^{A}_{+} \implies T(s) \in \Omega^{B}_{+}$.
\end{defn}
Diagrammatically, for all $s\in\Omega^A$, a positive transformation satisfies:
 \beq
\InputIfFileExists{Diagrams/Pos1.tikz}{}{\input{./figures/Diagrams/Pos1.tikz}} \ \ \in \ \ \Omega_+^B.
 \eeq

\begin{defn}[Completely Positive Transformation]\label{def:CPdef}
	A transformation $T \in \mathcal{L} (V^{A}, V^{B} )$ is said to be completely positive if it is positive and for any system C, the transformation $T \otimes \mathds{1}_{C} \in \mathcal{L}( V^{A \cdot C} , V^{B \cdot C})$ is positive, where $ \mathds{1}_{C}$ is the identity map on $V^{C}$.
\end{defn}
Diagrammatically, for all systems $C$ and all bipartite states $s\in \Omega^{AC}$, this means that:
\beq
\InputIfFileExists{Diagrams/Pos2.tikz}{}{\input{./figures/Diagrams/Pos2.tikz}} \ \ \in \ \ \Omega^{BC}_+ \, . 
\eeq

Any time that a transformation appears in a diagram, it is implied that it is a completely positive transformation for the corresponding GPT because it is an allowed transformation in said theory. The same applies for states: if they appear in a diagram, they must be positive in the corresponding GPT. So, the diagrams drawn in the beginning of this appendix are examples of positive transformations and states. Moreover, note that these notions recover those of positive, completely positive, trace preserving, and trace non-increasing maps when applied to quantum theory, because the cones are generated by the sets of states, and the unit effects are the trace operations.

\begin{defn}[Local Tomography]
A GPT is said to be locally tomographic if any state $ \rho^{A_{1} \cdot ... \cdot A_{n} }$ of a composite system $ A_{1} \cdot ... \cdot A_{n}$ can be uniquely determined by the information obtained from local effects $ \{ e^{A_{1}} \}$,..., $\{ e^{A_{n}} \}$ on its parts $A_{1}$,..., $A_{n}$.
When this holds, the unit vector for the composite system, $ u^{A_{1} \cdot ... \cdot A_{n} }$, is given by $ u^{A_{1}} \otimes ... \otimes u^{A_{n}}$.
\end{defn}
As an example of a GPT satisfying local tomography we have quantum theory. There, any $ \rho^{A_{1} \cdot ... \cdot A_{n}}$ is completely determined by a set of probabilities $ (e_{i_{1}}^{A_{1}} \otimes ... e_{i_{n}}^{A_{n}})[ \rho^{A_{1} \cdot ... \cdot A_{n}} ] $, where each local effect $ e_{i_{j}}^{A_{j}}$ denotes the inner product of the quantum state with the corresponding POVM element.

\begin{defn}[Maximal Tensor Product]\label{def:max}
The maximal tensor product, $ \otimes_{max}$ is a rule for the combination of two systems into one, say, $A$ and $B$ into $A \cdot B$, that defines the positive cone of the composite system as the largest set of vectors in $V^{A} \otimes V^{B} $ that is consistent (in the sense of producing sensible probabilities) with all the separable effects of $ A \cdot B$. That is 
\begin{equation}\label{}
	\begin{aligned}
		\Omega^{A \cdot B}_{+} = \Omega^{A}_{+} \otimes_{max} \Omega^{B}_{+} := \{ &\rho \in V^{A} \otimes V^{B} : ( e^{A} \otimes e^{B} ) [ \rho ] \geq 0 \\&\forall e^{A} \in E^{A}\,,\, e^{B} \in E^{B} \}.
	\end{aligned}
\end{equation}
\end{defn}
The maximal tensor product $ \otimes_{max}$ is associative \cite{barnum2011information}, hence one can unambiguously write
\begin{equation}\label{}
	\Omega_{+}^{A_{1}} \otimes_{max} ... \otimes_{max} \Omega^{A_{n}}_{+},
\end{equation}
 which has the explicit form:
 \begin{equation}\label{}
	 \left\{ \rho \in \bigotimes_{i=1}^{n} V^{A_{i}} : \bigotimes_{i=1}^{n} e^{A_{i}} [\rho ] \geq 0 \quad \forall e^{A_{i}} \in E^{A_{i}} \right\}.
 \end{equation}
As discussed after Def.\ $\ref{posvector}$, this operation fixes the state spaces for the composite systems, because the cone $ \Omega_{+}^{A \cdot B}$ and the unit effect $ u^{A} \otimes u^{B}$ (which is the unit effect for $ u^{A \cdot B}$ in locally tomographic GPTs \cite{chiribella2010probabilistic}) can be used to construct the set of states $ \Omega^{A \cdot B}$.

\begin{defn}[No-Restriction Hypothesis \cite{chiribella2010probabilistic}]
A theory is said to satisfy the no-restriction hypothesis (NRH) if any element $e \in (V^{A})^*$ that gives $ e( \rho ) \in [0,1]$ for every $ \rho \in \Omega^{A}$ is an element of $ E^{A}$. That is, if
\begin{equation}\label{}
	E^{A} = \{ e \in (V^{A} )^* : \forall \rho \in \Omega^{A}, e( \rho ) \in [0,1]\}.
\end{equation}
\end{defn}
This is to say, the NRH is the statement that, given the set of states, the set of effects is the largest possible that still gives sensible  probabilities for every state. 

\begin{defn}[Generalised No-Restriction Hypothesis]
A theory is said to satisfy the Generalised No-Restriction Hypothesis (GNRH), if it satisfies the NRH and every completely positive trace non-increasing transformation is an allowed transformation. That is, for any two systems $A$ and $B$, every transformation $T \in \mathcal{L} (V^{A} , V^{B} )$ that takes elements in $ \Omega^{A \cdot C}_{+}$ to elements in $ \Omega^{B \cdot C}_{+} $ for any third system C, and satisfies $ u^{B} ( T ( \rho ) ) \leq u^{A} ( \rho )$ for any $ \rho \in V^{A}$, is a valid transformation from $A$ to $B$.
\end{defn}
This is a convenient
 assumption to make about a theory because it simplifies its description, since it implies that the specification
 of the state spaces uniquely fixes both the effects and transformations. Again, we can use quantum theory as an example, as it does satisfy the GNRH.

Our last definition in this section is that of entanglement witnesses in a generic GPT. 
This definition follows closely to that found in the context of quantum theory. 

\begin{defn}[Generic Bipartite Entanglement Witness]\label{bipwitness} The set $W^{A \cdot B}$ of entanglement witnesses of a bipartite system $A\cdot B$ for a generic locally tomographic GPT is given by
	\begin{equation}\label{bipwit}
		\begin{aligned}
			W^{A \cdot B} = \{ w \in V^{A} \otimes V^{B} : &\langle w, s^{A} \otimes s^{B} \rangle \geq 0 \\ &\forall s^{A} \in \Omega^{A}_{+} , s^{B} \in \Omega^{B}_{+} \}\,.
		\end{aligned}
	\end{equation}
Note that this will depend on the choice of inner product. In the quantum case the standard choice will be the Hilbert-Schmidt inner product.
\end{defn} 
That is, an entanglement witness is a vector associated, through the Riesz representation, to a linear functional which evaluates to positive numbers for every product state of the bipartite system. This is a simple generalization of the quantum entanglement witnesses that uses arbitrary GPT states instead of quantum states. The generalization for multipartite entanglement witnesses is straightforward. 

\begin{defn}[Generic Multipartite Entanglement Witness]\label{multwitness} The set $ W^{A_{1} \cdot ... \cdot A_{n}}$ of the entanglement witnesses of a multipartite system $ A_{1} \cdot ... \cdot A_{n}$ for a generic locally tomographic GPT is given by
	\begin{equation}\label{}
		\begin{aligned}
			W^{A_{1} \cdot ... \cdot A_{n}} = \{ &w \in V^{A_{1}} \otimes ... \otimes V^{A_{n}} : \\ &\quad \langle w, s^{A_{1}} \otimes ... \otimes s^{A_{n}} \rangle \geq 0 \\ &\quad \forall s^{A_{1}} \in \Omega^{A_{1}}_{+},..., s^{A_{n}} \in \Omega^{A_{n}}_{+} \} \,.
		\end{aligned}
	\end{equation}
	Like in the bipartite case, this will depend on the choice of inner product. Again, in the quantum case the standard choice will be the Hilbert-Schmidt inner product.
\end{defn}

\subsection{Some useful results}
\label{app:usefulresults}
We now prove (or reprove) various results which are useful later on. 
Firstly we note an important consistency condition for the max tensor product. This is well known in the literature (see, e.g., Ref.~\cite{janotta2013generalized}) but we reproduce it here for completeness.

\begin{proposition}
In a theory where systems compose via the max tensor product, and that satisfies the NRH, one can check that the vectors, $v\in V^A$, that can be steered to from bipartite states in $\Omega^{AB}$ correspond to subnormalised states, i.e., live in $\Omega^A_+$ and satisfy $u^A(v)\leq 1$.\footnote{Note that this is an essential consistency condition for any GPT, but here we see that it is automatically satisfied by GPTs satisfying NRH and composing via the max tensor product and, hence, imposes no further constraints.}
\end{proposition}

\proof
Steered states are of the form $v_A=\mathds{1}_A \otimes e_B (s_{AB})$ for a bipartite state $s_{AB}$ and an effect $e_B$ and the identity transformation $\mathds{1}_A$. Note that a special case of this is the reduced state which is given by taking $e_B= u^B$. We want that these vectors are local (subnormalised) states. Note that, by definition of the max tensor product we have that for all $e_A$ that $e_A\otimes e_B (s_{AB}) \in [0,1]$ and hence that $e_A(v_A) \in [0,1]$. As in our theory the local effects $e_A$ are defined via the NRH this means that $v_A$ must be in the cone $\Omega_+^A$. Moreover, it is easy to compute that it is (sub)normalised as $u^A(v_A) = u^A\otimes e^B (s_{AB}) \in [0,1]$.
\endproof

Another well-known result (see, e.g., Ref.~\cite{barnum2011information}) is the following.
\begin{proposition}\label{thm:NL}
In a GPT composed by the max tensor product, every effect on a composite system is a separable effect \cite{barnum2011information}.
\end{proposition}
\begin{proof}
It follows straightforwardly by noticing that the use of $ \otimes_{max}$ to combine systems implies that the set of effects of the combined system $A \cdot B$ is just the set of separable effects. 
\end{proof}

Next we prove a lemma which is useful for proving a key observation.

\begin{lemma}\label{lemma}
	In a GPT containing systems $A$ and $B$, if the no-restriction hypothesis holds and the map $T \in \mathcal{L} (V^{A} , V^{B}) $ is positive, then 
	\begin{equation}\label{}
\begin{tikzpicture}
	\begin{pgfonlayer}{nodelayer}
		\node [style=none] (0) at (0, -0.5) {};
		\node [style=copoint] (1) at (0, 0.5) {$e$};
		\node [style=right label] (2) at (0, -0.25) {$B$};
	\end{pgfonlayer}
	\begin{pgfonlayer}{edgelayer}
		\draw (0.center) to (1);
	\end{pgfonlayer}
\end{tikzpicture}
} \ \ \in E^{B}_{+}\quad \implies\quad %
\InputIfFileExists{Diagrams/lem1John2.tikz}{}{\input{./figures/Diagrams/lem1John2.tikz}} \ \ \in E^{A}_{+},
	\end{equation}
	\end{lemma} 
\begin{proof}
	Recall that positivity of $ T \in \mathcal{L} ( V^{A} , V^{B} )$ means that
	\beq
\begin{tikzpicture}
	\begin{pgfonlayer}{nodelayer}
		\node [style=none] (0) at (0, 0.75) {};
		\node [style=point] (1) at (0, -0.25) {$s$};
		\node [style=right label] (2) at (0, 0.5) {$A$};
	\end{pgfonlayer}
	\begin{pgfonlayer}{edgelayer}
		\draw (0.center) to (1);
	\end{pgfonlayer}
\end{tikzpicture}
} \ \ \in \Omega^{A}_{+}\quad \implies\quad %
\InputIfFileExists{Diagrams/LEM1jOHN4.tikz}{}{\input{./figures/Diagrams/LEM1jOHN4.tikz}} \ \ \in \Omega^{B}_{+}.\eeq 
	Now, using this positivity and, noting that $e^B \in E^B_+$, we find for all $ s^{A} \in \Omega^{A}_{+} $ that
\beq
\InputIfFileExists{Diagrams/lem1John5.tikz}{}{\input{./figures/Diagrams/lem1John5.tikz}} \ \ \geq 0.
\eeq 
Hence, 
\beq
\InputIfFileExists{Diagrams/lem1John2.tikz}{}{\input{./figures/Diagrams/lem1John2.tikz}} \ \ \in \ (\Omega_+^A)^*\ \stackrel{\textsc{nrh}}{=}\ E^{A}_{+},
\eeq
which completes the proof.
\end{proof}

The lemma above can be used to prove a useful fact for our work.

\begin{thm}\label{thm:PisCP}
	In any GPT that combines systems through $ \otimes_{max}$, the max tensor product, and satisfies the no-restriction hypothesis, if a map $T$ is positive, then it is completely positive.
\end{thm}
\begin{proof}
	For the sake of contradiction, assume that, in a GPT satisfying the NRH where $ \otimes_{max}$ is the combination rule, the map $T \in \mathcal{L} (V^{A} , V^{B} ) $ is positive but not completely positive.  That is, that there exists a system $C$ such that $T \otimes \mathds{1}_{C} \in \mathcal{L} ( V^{A} \otimes V^{C} , V^{B} \otimes V^{C} )$ is not positive. That means that there must exist some bipartite state $s\in \Omega^{A\cdot B} = \Omega^{A}_{+} \otimes_{max} \Omega^{C}_{+}$ such that
	\begin{equation}\label{}
\InputIfFileExists{Diagrams/diagProof2.tikz}{}{\input{./figures/Diagrams/diagProof2.tikz}}\ \  \not\in\ \Omega^{B}_{+} \otimes_{max} \Omega^{C}_{+}.
	\end{equation}
	By the definition of $ \otimes_{max}$, this means exists $t \in E^{B}$ and $v \in E^{C}$ such that
	\begin{equation}\label{eq:negativeCond}
\InputIfFileExists{Diagrams/diagProof3.tikz}{}{\input{./figures/Diagrams/diagProof3.tikz}} < 0.
	\end{equation}
However, as $t \in E^B \subset E^B_+$, we know from lemma \ref{lemma} that:
\beq
\InputIfFileExists{Diagrams/diagProof4.tikz}{}{\input{./figures/Diagrams/diagProof4.tikz}} \ \ \in \ \ E^A_+
\eeq
and hence that there exists $\lambda \geq 0 $ such that 	
\beq
\begin{tikzpicture}
	\begin{pgfonlayer}{nodelayer}
		\node [style=none] (0) at (0, -0.5) {};
		\node [style=none] (12) at (0, -0.5) {};
		\node [style=none] (13) at (0, -0.5) {};
		\node [style=copoint] (14) at (0, 0.5) {$t'$};
		\node [style=right label] (15) at (0, -0.25) {$A$};
	\end{pgfonlayer}
	\begin{pgfonlayer}{edgelayer}
		\draw (14) to (13.center);
	\end{pgfonlayer}
\end{tikzpicture}
} \ \ := \ \ \lambda \  %
\InputIfFileExists{Diagrams/diagProof4.tikz}{}{\input{./figures/Diagrams/diagProof4.tikz}} \ \ \in \ \ E^A.
\eeq
Substituting this into eq.~\eqref{eq:negativeCond} gives us that:
\beq
\InputIfFileExists{Diagrams/diagProof6.tikz}{}{\input{./figures/Diagrams/diagProof6.tikz}} \ \ < \ 0
\eeq
but, as $t' \in E^A$ and $v \in E^C$ this means that $s \not\in \Omega^{A}_{+} \otimes_{max} \Omega^{C}_{+}$ and, hence, we have reached a contradiction. 
\end{proof}

\section{Formal definition  and features of Witworld}\label{app:WW}

In this Appendix we provide the formal definition of Witworld as a GPT, and the proofs that it does indeed possess the features mentioned in the main text. That task amounts to explicitly saying what are the states, effects and transformations of Witworld, following the formalities of the GPT framework mentioned in App.~\ref{ap:GPT}.
For Witworld, this is simplified because we define it to satisfy the GNRH, so by providing just the state spaces for each system type (including the multipartite ones, which requires the combination rule), we determine the complete GPT.

As was said in the main text, Witworld contains systems which we call atomic systems, and systems which we call composite systems. A system $A$ is atomic if it cannot be considered as being the result of combining a system $B$ with another system $C$. That is, there are no $B$ and $C$ such that $A = B \cdot C$. This means that any system type can be built from the atomic systems, so to determine all the types in Witworld, we only need to say what the atomic systems are and what the combination rule $ \cdot $ is. Regarding the latter, we choose $ \cdot $ to be the maximal tensor product $ \otimes_{max}$. Regarding the former, we now formally define the atomic systems:

\begin{defn}[Quantum System $\Q{d}$, $d \in \mathds{N}$]\label{quantumsystem} A quantum system of type $d$ has as its vector space $V^{\Q{d}}$, as positive cone $ \Omega^{\Q{d}}_{+}$, as effects set $ E^{\Q{d}}$, and as unit effect the function $u^{\Q{d}}( \_)$, all defined as follows:
	\begin{itemize}
		\item $ V^{\Q{d}} = \{ A \in \mathcal{L} ( \mathcal{H}_{d}, \mathcal{H}_{d}) : A = A^{\dagger} \}$: the space of Hermitian operators on a Hilbert space of dimension $d$.
		\item $ \Omega_{+}^{\Q{d}} = \{ A \in V^{\Q{d}} : A \geq 0 \}$: the set of positive operators on $ \mathcal{H}_{d}$.
		\item $u^{\Q{d}}(\_) = \tr{(\mathds{1} \_)}$: the trace inner product with the identity operator $ \mathds{1}$ in $ \mathcal{L} ( \mathcal{H}_{d}, \mathcal{H}_{d})$.
		\item $ E^{\Q{d}} = \{ \tr(A \_ ) : A \in V^{\Q{d}} \text{ and } 0 \leq A \leq \mathds{1} \}$: the set of trace inner products with operators in $ V^{\Q{d}}$ that are positive and smaller than or equal to the identity, as required by NRH.
	\end{itemize}
\end{defn}
For atomic systems, the quantum type coincides with those of single systems in traditional quantum theory, and indeed local states and effects of atomic systems coincide for quantum types in both theories. However, as we see later on, this no longer holds for either transformations or composite systems -- for the latter, this follows from the fact that the combination rule in Witworld, the maximal tensor product, is not the same as in quantum theory, so $ \Q{d} \cdot \Q{d'} \neq \Q{dd'}$.

\begin{defn}[Classical system $\C{v}$, $v \in \mathds{N}$] \label{classicalsystem} A classical system of type $v$ has as its vector space $ V^{\C{v}}$, as positive cone $ \Omega_{+}^{\C{v}}$, as effects set $ E^{\C{v}}$, and as unit effect the function $ u^{\C{v}}(\_ )$, all defined as follows: 
	\begin{itemize}
		\item $ V^{\C{v}} = \mathds{R}^{v-1} \oplus \mathds{R}^{1} \cong \mathds{R}^{v} $: The direct sum of a ($v-1$)-dimensional real vector space with the real numbers. We can work with the isomorphic space $ \mathds{R}^{v}$ to simplify notation.
		\item $ \Omega^{\C{v}}_{+}=\{v\in V^{\C{v}}: v= \lambda (q \oplus 1), \lambda \geq 0, q_{i} \geq 0, \sum_{i} q_{i} \leq 1 \}$ : the set of vectors in $ \mathds{R}^{v}$ that are the null vector or have a positive last component and whose first $v-1$ components divided by the last give the probabilities for $ v-1$ outcomes of a measurement with $v$ possible outcomes.
		\item $ u^{\C{v}} (\_ ) = \langle (0,...,0,1)^{T}, \_ \rangle$: the Euclidean inner product with the vector in $ \mathds{R}^{v}$ whose only nonzero component is the last one, which is 1.
		\item $ E^{\C{v}} = \{ \langle e, \_ \rangle: e \in V^{\C{v}}, \langle e, s \rangle \in [0,1] \forall s \in \Omega^{\C{v}} \}$: the set of euclidean inner products with vectors in $V^{\C{v}}$ that evaluate to probabilities for every vector in $ \Omega^{\C{v}}$, as required by the NRH.
	\end{itemize}
\end{defn}
Those are the traditional classical systems -- probability distributions written as vectors. Such vectors can always be seen as a convex combination of deterministic states. Note that, unlike for quantum systems, for classical systems we do have that $\C{d}\cdot\C{d'}=\C{dd'}$\footnote{Hence strictly the atomic systems should be taken to be prime dimensional classical systems, however, we do not worry about this subtlety here.}. 
Writing classical states in this form allows us to further notice that they are just particular cases of the Boxworld type.

\begin{defn}[Boxworld system $ \mathcal{B}_{n,k}$, $ (n,k) \in \mathds{N}^{2}$]\label{boxworldsystem} A Boxworld system of type $(n,k)$\footnote{Strictly we should demand that $n>1$ so as not to duplicate the classical systems, however, we do not worry about this subtlety here either.} has as its vector space $ V^{\B{n,k}}$, as positive cone $ \Omega^{\B{n,k}}$, as effects set $ E^{\B{n,k}}$, and as unit effect the function $ u^{\B{n,k}} ( \cdot )$, all defined as follows:
	\begin{itemize}
		\item $ V^{\B{n,k}} = ( \mathds{R}^{n} \otimes \mathds{R}^{k-1}) \oplus \mathds{R}^{1} \cong \mathds{R}^{n(k-1)+1}$: The direct sum of the real numbers with the direct product between two real vector spaces of dimensions $n$ and $k-1$. We can work with the isomorphic space $ \mathds{R}^{n(k-1) + 1}$ to simplify notation.
		\item $ \Omega_{+}^{\B{n.k}} = \{ v \in V^{\B{n,k}}: v = \lambda \left( \sum_{i=1}^{n} \boldsymbol{m}_{i} \otimes \boldsymbol{q}_{i} \oplus 1 \right), \lambda \geq 0, q_{ij} \geq 0, \sum_{j=1}^{k} q_{ij} \leq 1 , m_{ij} = \delta_{ij} \}$: the vectors in $ \mathds{R}^{n(k-1)+1}$ which are the null vector or that have a positive last component and the first $n(k-1)$ components divided by the last component (if positive) can be viewed as probabilities for the first $k-1$ components of $n$ measurements of $k$ possible outcomes stacked in a list.
		\item $ u^{\B{n,k}}(\_) = \langle \overline{u}, \_ \rangle$ where $ \overline{u} = 0_{ \mathds{R}^{n}} \otimes 0_{ \mathds{R}^{k+1}} \oplus  1$: the inner product with the vector in $ \mathds{R}^{n(k-1)+1}$ with only null components except for the last, which is 1.
		\item $ E^{\B{n,k}} = \{ \langle e, \_ \rangle: e \in V^{ \B{n,k}}, \langle e, s \rangle \in [0,1] \forall s \in \Omega^{\B{n,k}} \}$: the set of inner products with vectors in $ V^{\B{n,k}}$ that satisfies the NRH.
	\end{itemize}	

\end{defn}
The Boxworld systems can be viewed as classical systems that require many measurements to uniquely determine a state, rather than just 1, and the probability distributions for those measurements are independent of each other. A classical system of type $v$, then, can be viewed as a Boxworld system of type $(1,v)$.  Note, however, that unlike classical systems, the composite of two more general Boxworld systems is no longer an atomic Boxworld system, that is $\mathcal{B}_{n,k}\cdot \mathcal{B}_{n',k'}\neq \mathcal{B}_{n'',k''}$. 

When using diagrams, we denote the atomic classical, quantum, and Boxworld systems by different types of wires:
\begin{equation}\label{}
\begin{tikzpicture}
	\begin{pgfonlayer}{nodelayer}
		\node [style=none] (0) at (0, 0.75) {};
		\node [style=none] (1) at (0, -0.75) {};
		\node [style={right label}] (2) at (0, -0.5) {$\C{v}$};
	\end{pgfonlayer}
	\begin{pgfonlayer}{edgelayer}
		\draw [c] (0.center) to (1.center);
	\end{pgfonlayer}
\end{tikzpicture}},\quad %
\begin{tikzpicture}
	\begin{pgfonlayer}{nodelayer}
		\node [style=none] (0) at (0, 0.75) {};
		\node [style=none] (1) at (0, -0.75) {};
		\node [style={right label}] (2) at (0, -0.5) {$\Q{d}$};
	\end{pgfonlayer}
	\begin{pgfonlayer}{edgelayer}
		\draw [q] (0.center) to (1.center);
	\end{pgfonlayer}
\end{tikzpicture}},\quad %
\begin{tikzpicture}
	\begin{pgfonlayer}{nodelayer}
		\node [style=none] (0) at (0, 0.75) {};
		\node [style=none] (1) at (0, -0.75) {};
		\node [style={right label}] (2) at (0, -0.5) {$\B{n,k}$};
	\end{pgfonlayer}
	\begin{pgfonlayer}{edgelayer}
		\draw [b=12] (0.center) to (1.center);
	\end{pgfonlayer}
\end{tikzpicture}
}.	
\end{equation}
When we need to talk about an arbitrary kind of system, the wire we use is the following:
\begin{equation}\label{}
\begin{tikzpicture}
	\begin{pgfonlayer}{nodelayer}
		\node [style=none] (0) at (0, 0.75) {};
		\node [style=none] (1) at (0, -0.75) {};
		\node [style={right label}] (2) at (0, -0.5) {$S$};
	\end{pgfonlayer}
	\begin{pgfonlayer}{edgelayer}
		\draw [g] (0.center) to (1.center);
	\end{pgfonlayer}
\end{tikzpicture}}.
\end{equation}

As stated previously, from the atomic systems any general system in Witworld can be constructed as an arbitrary composite of the three fundamental system types, $\Q{d},\C{v},\B{n,k} $.
For instance, $\Q{d} \cdot \Q{d'}$ and $\Q{d} \cdot \C{v} \cdot \C{v'} \cdot \B{n,k} \cdot \Q{d'}$ would both be systems within our theory.
More generally, systems correspond to arbitrary strings of elements from the set $\{ \Q{d},\C{v},\B{n,k}\}_{d,v,n,k \in \mathds{N}}$.
The positive cones for these composite systems are obtained through the max tensor product of the cones of the atomic types, and from those we can obtain the set of states by taking the intersection of the cone with the set of normalised vectors in the product vector space. Here, by normalised vector we mean vectors for which the unit effect evaluates to 1. Since Witworld is a locally tomographic GPT, the unit effect for $ A_{1} \cdot ... \cdot A_{n}$ is simply $ u^{A_{1}} \otimes ... \otimes u^{A_{n}}$.

Since Witworld, by definition, satisfies the NRH, once we establish what the states for every type of system are, the effects are also determined. Given any system like $\Q{d} \cdot \B{n,k} \cdot ...$, every linear functional on $ V^{\Q{d}} \otimes V^{\B{n,k}} \otimes ...$ that gives probabilities for every vector in {$ \Omega^{\Q{d} \cdot \B{n,k} \cdot ...}$ is a valid effect. 

The definition of the transformations in Witworld is similar to that of the effects. Here, we require the theory to satisfy the GNRH, so when we determine the states, the transformations are fixed. In Witworld, any completely positive transformation is allowed. We prove later that this, together with the fact that the combination rule is $ \otimes_{max}$, implies that any positive transformation is an allowed transformation for arbitrary systems in Witworld.

Formally, and concisely, Witworld is therefore defined as follows:

\begin{defn}[Witworld] Witworld is the locally-tomographic GPT that satisfies the generalised no-restriction hypothesis, and whose systems are arbitrary combinations under the max tensor product $ \otimes_{max}$ of the atomic system types described in definitions \ref{quantumsystem}, \ref{classicalsystem}, and \ref{boxworldsystem}.
\end{defn}

Now that we presented the definition of the theory, we can move on to observing or proving the various features of Witworld which we used
in the main text. Note that because Witworld composes via the max tensor product and satisfies the GNRH, that all of the results of App.~\ref{app:usefulresults} hold.

Firstly, Proposition~\ref{thm:NL}, tells us that in Witworld there are only separable effects.
Therefore, for systems that are the combination of atomic quantum systems, there are fewer effects in Witworld than in quantum theory: effects from measurements in an entangled basis are not present in Witworld. For Boxworld and classical systems, such a difference does not exist: that is, Boxworld and classical system types feature separable-only effects both in Witworld and in their respective traditional frameworks.

Theorem~\ref{thm:PisCP}, together with the GNRH means that local transformations which in quantum theory are positive but not completely positive maps are indeed valid transformations in Witworld. This is important for constructing many examples of post-quantum assemblages. Again for Boxworld and classical systems this distinction does not exist as for them the notion of positivity and complete positivity already coincide in their respective traditional frameworks.

Next we show that bipartite states for quantum systems within Witworld correspond to quantum entanglement witnesses. 

\begin{thm}\label{bipgleason} In Witworld, 
the bipartite system resulting from the combination of atomic 
quantum 
systems $\Q{d}$ and $\Q{d'}$ 
contains every bipartite entanglement witness in its positive cone.
\end{thm}

\begin{proof} Take the set of effects $ E^{\Q{d}}$. To each effect $ \varphi_{e} $ in it, there is a vector $ e \in V^{\Q{d}} $ associated to it by its Riesz representation with the Hilbert-Schmidt inner product. Let us call the set of all $ e \in V^{\Q{d}}$ associated to some $ \varphi_{e} \in E^{\Q{d}} $ by $ \tilde{E}^{\Q{d}}$. Do the same to define $ \tilde{E}^{\Q{d'}}$. In quantum theory, $ \tilde{E}_{+} = \Omega_{+}$ because the effects are inner products with positive operators. 
	Now, note that by changing $ E^{A} \to E^{A}_{+}$ and $ E^{B} \to E^{B}_{+}$ we do not change the set defined by equation \ref{bipwitness}. Hence, we can write it, already using the action of the effect as an inner product, as
	\begin{equation}\label{}
		\begin{aligned}
			\Omega^{A}_{+} \otimes_{max} \Omega^{B}_{+} = \{ v &\in V^{A} \otimes V^{B} : \langle v, e^{A} \otimes e^{B} \rangle \geq 0 \\ &\quad \forall e^{A} \in \tilde{E}_{+}^{A}, e^{B} \in \tilde{E}^{B}_{+} \}.
		\end{aligned}
	\end{equation}
	If $A = \Q{d}$ and $B = \Q{d'}$, then $ \tilde{E}^{A}_{+} = \Omega^{\Q{d}}_{+} $ and $ \tilde{E}^{B}_{+} = \Omega^{\Q{d}}_{+}$, so the equation above is by definition the set $ W^{\Q{d} \cdot \Q{d'}} $ of entanglement witnesses of $\Q{d'} \cdot \Q{d'}$. 
\end{proof}
This theorem means that if we compare the positive cones in quantum theory with the positive cones in Witworld when combining two atomic quantum systems, we see that the positive cone in Witworld is larger than the positive cone in quantum theory.
To see this more explicitly, suppose $A$ and $B$ are two atomic quantum systems, then refer to $A$ and $B$ combined as prescribed by quantum theory by $A \otimes B$, and by $A \cdot B$ when combined as prescribed by Witworld.
Then, since definition \ref{bipwitness} is independent of the combination rule and is equivalent to what is used in quantum theory, theorem \ref{bipgleason} implies $ \Omega^{A \cdot B}_{+} = W^{A \cdot B} = W^{A \otimes B} \supset \Omega^{A \otimes B}_{+}$, where the last inclusion is given from quantum theory. Nevertheless, since classical and Boxworld systems originally combine through $ \otimes_{max}$, in Witworld the combination of atomic systems of said types do not build more states than what we would normally have. Finally, states that are the combination of different types of atomic systems are incomparable to states in quantum, Boxworld or classical systems. 

The fact that we can view quantum entanglement witnesses as valid states in Witworld is a key feature which underpins many of our realisations of post-quantum assemblages. This is also true in the multipartite generalisation which we now prove.

\begin{thm}\label{thm:B7}
	In Witworld, the multipartite system resulting from the combination of atomic quantum systems $\Q{d_{1} } = A_{1}, ... , \Q{d_{n} } = A_{n}$ contains every multipartite entanglement witness in its positive cone.
\end{thm}

\begin{proof}
This is a straightfoward generalization of the bipartite case:
	\begin{equation}\label{}
		\begin{aligned}
			\Omega_{+}^{A_{1}}&\otimes_{max} \cdots \otimes_{max} \Omega_+^{A_n} = \\ &= \left\{ s \in \bigotimes_{i=1}^{n} V^{A_{i}} : \bigotimes_{i=1}^{n} \varphi_{e}^{A_{i}} [s] \geq 0, \quad \forall \varphi_{e}^{A_{i}} \in E^{A_{i}} \right\}\\\
			&= \left\{ s \in \bigotimes_{i=1}^{n} V^{A_{i}} : \left\langle \bigotimes_{i=1}^{n} e^{A_{i}} , s \right\rangle \geq 0 \quad \forall e^{A_{i}} \in \tilde{E}^{A_{i}}_{+} \right\}\\
			&= \left\{ s \in \bigotimes_{i=1}^{n} V^{A_{i}} : \left\langle \bigotimes_{i=1}^{n} e^{A_{i}} , s \right\rangle \geq 0 \quad \forall e^{A_{i}} \in \Omega^{A_{i}}_{+} \right\} \\
			&= W^{A_{i} \otimes ... \otimes A_{n}}
		\end{aligned}
	\end{equation}
	where $e^{A_{i}} \in V^{A_{i}}$ is associated to $ \varphi_{e}^{A_{i}} \in (V^{A_{i}})^*$ by the Riesz representation and $ \tilde{E}^{A_{i}}_{+} = \Omega^{A_{i}}_{+}$ because the $A_{i}$ are atomic quantum systems.
\end{proof}

\section{How to realise a PR-box in Boxworld and Witworld}\label{ap:PRbox}

To make explicit that Witworld can, in fact, realise a PR-box, we explicitly write down the elements in the diagram
\beq
\InputIfFileExists{Diagrams/prCorr.tikz}{}{\input{./figures/Diagrams/prCorr.tikz}}
\eeq
and thereby show that the following holds: 
\begin{equation}
p_{\mathrm{PR}}(ab|xy)  := \frac{1}{2} \delta_{a\oplus b = xy} = %
\InputIfFileExists{Diagrams/prCorr.tikz}{}{\input{./figures/Diagrams/prCorr.tikz}},
\end{equation}
where $ \oplus$ is addition modulo 2, and $a,b,x,y \in \{0,1\}$.

The first step is to define the measurements $M_{\mathrm{PR}}$ and $M'_{\mathrm{PR}}$. We characterise these by their associated set of effects. For example, $M_{\mathrm{PR}}$ by $\{e_{a|x}^{A}\}$, for the outcome $a$ when measurement $x$ is performed, diagrammatically these are defined as:
\beq\label{eq:40}
\begin{tikzpicture}
	\begin{pgfonlayer}{nodelayer}
		\node [style=none] (4) at (-0.5, -1.5) {};
		\node [style=none] (26) at (-0.5, -1.5) {};
		\node [style=copoint] (27) at (-0.5, 0) {$e_{a|x}$};
	\end{pgfonlayer}
	\begin{pgfonlayer}{edgelayer}
		\draw [b=6, in=90, out=-90] (27) to (26.center);
	\end{pgfonlayer}
\end{tikzpicture}
} \ \ := \ \ %
\InputIfFileExists{Diagrams/PREffect1.tikz}{}{\input{./figures/Diagrams/PREffect1.tikz}}.
\eeq
Similarly we can characterise the measurement $M'_{\mathrm{PR}}$ by the set of effects $\{e^{B}_{b|y}\}$ which are defined analogously. We therefore want to verify that there exists measurements and states such that:
\beq %
\InputIfFileExists{Diagrams/PREffect3.tikz}{}{\input{./figures/Diagrams/PREffect3.tikz}} \ \ = \ \ \frac{1}{2} \delta_{a\oplus b = xy},\eeq
which, to make this more explicit, can be rewritten symbolically as:
\beq 
 e_{a|x} \otimes e_{b|y} ( s^{AB}) = \frac{1}{2} \delta_{a\oplus b = xy}.
\eeq

Following definition \ref{boxworldsystem}, the states of $\B{2,2}$ are three component real vectors, whose first component is $p(0|0)$, second component is $p(0|1)$ and last component is $1$.
Therefore, for $e_{a|x} (s) = p(a|x)$ to hold, we need the vectors $ \tilde{e}_{a|x} $, associated to $e_{a|x}$ by the Riesz representation, to be given by
\begin{equation}\label{eq:thePReffectsp}
	\begin{aligned}
		\tilde{e}_{0|0} &=\begin{pmatrix} 1 \\ 0 \\ 0 \end{pmatrix},\quad \tilde{e}_{1|0} = \begin{pmatrix} -1 \\0 \\ 1 \end{pmatrix} \\ \tilde{e}_{0|1} &= \begin{pmatrix} 0 \\ 1 \\ 0 \end{pmatrix}, \quad \tilde{e}_{1|1} = \begin{pmatrix} 0 \\ -1 \\ 1 \end{pmatrix}.
	\end{aligned}
\end{equation}
Also by definition \ref{boxworldsystem}, the Riesz representation of the unit effect $u^{A}$ is given by
\begin{equation}\label{}
	\tilde{u}^{A} = \begin{pmatrix} 0 \\ 0 \\ 1 \end{pmatrix},
\end{equation}
so that $ \tilde{e}_{0|x} + \tilde{e}_{1|x } = \tilde{u}^{A} $, which makes $\{ e_{a|x} \}$ a valid measurement in $\B{2,2}$ for each $x$.

From the vectors above, we can, using the Kronecker product, write $ \tilde{e}_{a|x} \otimes  \tilde{e}_{b|y}$, which are associated to the product effects of the composite system $ \B{2,2} \cdot \B{2,2}$.
Now, any vector $s^{AB}$  in $V^{\B{2,2}} \otimes V^{\B{2,2}} \cong \mathds{R}^{9}$ such that $ e_{a|x} \otimes e_{b|y} ( s^{AB}) = \langle \tilde{e}_{a|x} \otimes \tilde{e}_{b|y} , s^{AB} \rangle \geq 0$ for all $a,b,x,y \in \{0,1\}$ is in the positive cone $ \Omega^{\B{2,2}}_{+} \otimes_{max} \Omega^{\B{2,2}}_{+} $.

We now show that the following vector describes a normalised state in the positive cone of bipartite states, and, moreover, reproduces the statistics of the PR box as we desire. That is, consider the vector:
\begin{equation}\label{eq:sPR}
s_{\mathrm{PR}} = \begin{pmatrix} 1/2 \\ 1/2 \\ 1/2 \\ 1/2 \\ 0 \\ 1/2 \\ 1/2 \\ 1/2 \\ 1 \end{pmatrix}.
\end{equation}
and note that it satisfies
\begin{equation}\label{}
	\langle \tilde{e}_{a|x} \otimes \tilde{e}_{b|y} , s_{\mathrm{PR}} \rangle = \frac{1}{2} \delta_{a \oplus b, xy},
\end{equation}
which can be verified by direct calculation.
Therefore $s_{\mathrm{PR}}$ is in the positive cone of $ \B{2,2} \cdot \B{2,2}$ and moreover reproduces the PR box statistics. Finally it is also normalised as:
\begin{equation}\label{}
	u^{A} \otimes u^{B} (s_{\mathrm{PR}} ) = 1.
\end{equation}
Hence, $s_{\mathrm{PR}}$ is a valid state of $\B{2,2} \cdot \B{2,2}$ which recovers the PR-box under the separable measurements $\{ e_{a|x}^{A} \otimes e_{b|y}^{B}\}$, which proves that Witworld can indeed realise a PR-box.

\section{Formalities of steering scenarios}\label{ap:stee}

Here we introduce the basic concepts in steering, starting from the simple case of a bipartite scenario and building up to more general cases.
We present steering scenarios by comparison with Bell scenarios, so that the former can be viewed as a modification of the latter where one or more of the parties does not perform a measurement. After the transition from Bell scenarios to a simple steering scenario, we introduce more general ones and proceed to define some important types of assemblages.
In particular, we focus on generalised steering scenarios which display post-quantum features.

Consider the pictorial description\footnote{For now, these diagrams are not, strictly speaking, the same kind of mathematical diagrams from App.~\ref{ap:GPT}, because drawing such diagrams presupposes that all the parts of them are objects existing in some GPT, and we don't know if there exists a GPT capable of realizing all assemblages to lend us its diagrams.} of  a no-signalling box in a Bell scenario where Alice (Bob) measures $x \in \mathds{X} $ ($y \in \mathds{Y}$) and obtains the outcome $a \in \mathds{A} $ ($b \in \mathds{B}$):
\begin{equation}\label{}
\InputIfFileExists{Diagrams/bellScenario.tikz}{}{\input{./figures/Diagrams/bellScenario.tikz}} = p(ab|xy).
\end{equation}
Now, suppose instead that Bob decides not to perform the measurement $y$ or indeed any other measurement, and, instead, merely keeps his system -- by assumption, a quantum one. Then Bob has the subnormalised states $ \sigma_{a|x}$ which are given by the following scheme:
\begin{equation}\label{}
\InputIfFileExists{Diagrams/steeringScenario-John.tikz}{}{\input{./figures/Diagrams/steeringScenario-John.tikz}} = %
\InputIfFileExists{Diagrams/sigmaax.tikz}{}{\input{./figures/Diagrams/sigmaax.tikz}} \equiv \sigma_{a|x} = p(a|x) \rho_{a|x},
\end{equation}
where $ \rho_{a|x}$ is the normalised state in possession of Bob when Alice obtains outcome $a$ upon measuring $x$ with probability $p(a|x)$.
The complete description of this scenario is specified by the the set $ \Sigma_{\mathds{A}|\mathds{X}} = \{ \sigma_{a|x} \}_{a \in \mathds{A},x \in \mathds{X}}$ of subnormalised quantum states, which contains the information about the states that Bob can have by the end of Alice's measurement, each of them conditioned on some measurement $x$ and outcome $a$ on Alice's side, together with the probability of the outcome $a$ happening.
This set of subnormalised quantum states, $ \Sigma_{\mathds{A}|\mathds{X}}$, is known as an assemblage \cite{pusey2013negativity}.
Note that the assemblage elements $ \sigma_{a|x}$ indeed contain the complete information about the scenario because $p(a|x) = \tr( \sigma_{a|x}) $ and $ \rho_{a|x} = \sigma_{a|x} / \tr( \sigma_{a|x} )$. 

Of course, if signalling is permitted between Alice and Bob, then (within quantum theory) any assemblage can be trivially prepared, so we restrict our discussion to the non-signalling scenarios. The assemblages that can possibly be produced in this case are called non-signalling assemblages. We define them as being those satisfying conditions analogous to those that define non-signalling boxes in Bell scenarios:
\begin{defn}[No signalling bipartite assemblages] An assemblage $\Sigma_{\A|\X}$ is no signalling iff
\begin{equation}\label{}
	\sigma_{a|x} \in \Omega^{B}_{+},
\end{equation}
\begin{equation}\label{}
	\sum_{a} \sigma_{a|x} = \sum_{a} \sigma_{a|x'} = \rho^{B} \quad \forall x,x',
\end{equation}
\begin{equation}\label{}
	u^{B} ( \rho^{B} ) = 1.
\end{equation}
\end{defn}
In the channel based picture, these constraints can be captured diagrammatically by the condition that $\Sigma$ is a `causal' channel \cite{beckman2001causal}:
\begin{equation}\label{}
\InputIfFileExists{Diagrams/NS1-John.tikz}{}{\input{./figures/Diagrams/NS1-John.tikz}} \ \ =\ \ %
\InputIfFileExists{Diagrams/NS2-John.tikz}{}{\input{./figures/Diagrams/NS2-John.tikz}}.
\end{equation}
It can be seen that this is equivalent to the standard no-signalling condition by noting that whatever input $x$ is chosen for the classical system $\X$ can have no influence over the quantum system, as it is in a fixed normalised state $\rho^B$.
The study of steering is the study of the properties of assemblages  -- or equivalently, the study of the properties of `causal' classical-quantum channels. We use this to make a classification of types of assemblages in a meaningful way. Notice that while the set $ \Sigma_{\mathds{A}|\mathds{X}}$ is a set of (subnormalised) quantum states, the fact that the system that Alice measures
is not specified opens up the possibility for post-quantumness in the joint scenario, while keeping the quantum theoretical description valid for the local state of Bob.

Of course, like in the study of Bell non-locality, this scenario can be generalised. The two generalisations that we consider here are: i) adding more parties that are steering Bob; and, ii) allowing for Bob to have a setting variable $y\in\mathds{Y}$.
For the purpose of investigating post-quantumness, this is not only possible, but necessary, as it has been proven \cite{gisin1989stochastic,hughston1993complete} that every assemblage in the standard bipartite scenario can be realised in quantum theory. That is, any standard bipartite no signalling assemblage can be constructed by Alice performing a controlled measurement on one half of a bipartite quantum state shared with Bob.
We call the assemblages which are beyond the powers of quantum theory in non-signalling scenarios, that is, which cannot be realised in this way, \emph{post-quantum assemblages}. 

In multipartite scenarios, the assemblage elements now carry the labels for the outcomes $ a_{i} \in \mathds{A}_{i}$ for multiple parties $i \leq N$, and similarly for the measurement choices $x_{i} \in \mathds{X}_{i}$. They form assemblages $ \Sigma_{\mathds{A}_{1} ... \mathds{A}_{N}| \mathds{X}_{1} ... \mathds{X}_{N }}= \{\sigma_{a_{1} ... a_{N} | x_{1} ... x_{N}}\} $ and are given in diagrammatic notation by 
\beq
\InputIfFileExists{Diagrams/suggestion5.tikz}{}{\input{./figures/Diagrams/suggestion5.tikz}}
\eeq

\begin{defn}[No-signalling multipartite assemblages]\label{def:NSassem}
A multipartite assemblage, $\Sigma_{\mathds{A}_{1} ... \mathds{A}_{N}| \mathds{X}_{1} ... \mathds{X}_{N }}$, is said to be no-signalling if it satisfies the following non-signalling constraints:\\
a)
\begin{equation}\label{}
	\sigma_{a_{1} ... a_{N} | x_{1} ..._{N}} \in \Omega^{B}_{+},
\end{equation}
b)
\begin{equation}\label{mcausal}
	\sum_{a_{1}... a_{N}} \sigma_{a_{1} ... a_{N} | x_{1}... x_{N} } = \rho^{B},
\end{equation}
c)
\begin{equation}\label{eq1}
	u( \rho^{B} ) =1,
\end{equation}
\bel{d) Let $S = \{s_1, \ldots, s_r\} \subseteq \{1,...,N\}$ be an arbitrary set of $r$ parties, with $1 \leq r \leq N$, and denote by $\{t_1,...,t_{N-r}\} = \{1,...,N\} \setminus S$. Then, for all such $S$,}
\begin{align}\label{eq2}
	\sum_{a_{s}\,:\,s \in S} \sigma_{a_{1}... a_{N}| x_{1} ... x_{N}} =& \sigma_{a_{t_{1}}... t_{j_{N-r}}|x_{t_{1}} ... x_{t_{N-r}}}.
\end{align}
\end{defn}
\bel{The no-signalling constraints of Def.~\ref{def:NSassem} can alternatively be expressed in diagrammatic notation in a simple way. For each arbitrary partitioning of $\{1,..N\}=\{s_1,...,s_r\}\sqcup\{t_1,...,t_{N-r}\}$, with $S=\{s_1,...,s_r\}$ and now $0 \leq r \leq N$, describe this partitioning via a physical splitting of the wires into a left hand group (the $s_i$) and a right hand group (the $t_j$) depicted by the process $\mathsf{Part}_S$. Then diagrammatically, Eqs.~\eqref{mcausal}, \eqref{eq1}, and \eqref{eq2} read:}
\beq
\InputIfFileExists{Diagrams/multiNS-John.tikz}{}{\input{./figures/Diagrams/multiNS-John.tikz}} \ \ = \ \ %
\InputIfFileExists{Diagrams/multiNS2-John.tikz}{}{\input{./figures/Diagrams/multiNS2-John.tikz}}.
\eeq

The other kind of generalization of the bipartite scenario that we consider is to allow Bob to, instead of staying passive, perform a local transformation labelled by $y \in \mathds{Y}$ to his share of the system. Here, we require that Bob's system is locally a quantum one only after his transformation. In this case, the assemblage is denoted by $ \Sigma_{\mathds{A}|\mathds{X}\mathds{Y}}=\{\sigma_{a|xy}\}$, and pictorially is represented by
\begin{equation}\label{}
\InputIfFileExists{Diagrams/suggestion3.tikz}{}{\input{./figures/Diagrams/suggestion3.tikz}}
\end{equation}

\begin{defn}[No-signalling Bob-with-input assemblages]\label{def:BWI} A Bob-with-input assemblage $\Sigma_{\A|\X\Y}$ is  no-signalling iff the following no-signalling constraints are satisfied:
\begin{equation}\label{}
	\sigma_{a|xy} \in \Omega^{B}_{+} \hspace{0.5cm} \forall a,x,y,
\end{equation}
\begin{equation}\label{}
	\sum_{a} \sigma_{a|xy} = \sum_{a} \sigma_{a|x'y} \hspace{0.5cm} \forall x,x',y,
\end{equation}
\begin{equation}\label{}
	u^{B} (\sigma_{a|xy}) = p(a|x) \hspace{0.5cm} \forall a,x,y.
\end{equation}
\end{defn}
These can be pictorially represented by:

\begin{align}
\InputIfFileExists{Diagrams/BWINS-John.tikz}{}{\input{./figures/Diagrams/BWINS-John.tikz}} \ \ &= \ \ %
\InputIfFileExists{Diagrams/BWINS2-John.tikz}{}{\input{./figures/Diagrams/BWINS2-John.tikz}}, \text{ and} \\
\InputIfFileExists{Diagrams/BWINS3-John.tikz}{}{\input{./figures/Diagrams/BWINS3-John.tikz}} \ \ &= \ \ %
\InputIfFileExists{Diagrams/BWINS4-John.tikz}{}{\input{./figures/Diagrams/BWINS4-John.tikz}}.
\end{align}

The last scenario which is important to us is the instrumental steering scenario. This can be seen as a Bob-with-input scenario where Bob's input $y$ is completely determined by Alice's output $a$. The assemblage $ \Sigma^I_{\mathds{A}|\mathds{X}}=\{\sigma_{a|x}\}$ in the instrumental scenario
can be diagrammatically represented as a wiring of a Bob-with-input assemblage $\Sigma$:
\begin{equation}\label{eq:instdiag}
\InputIfFileExists{Diagrams/instr2-John.tikz}{}{\input{./figures/Diagrams/instr2-John.tikz}}
	\ \ = \ \ 
\InputIfFileExists{Diagrams/instr-John.tikz}{}{\input{./figures/Diagrams/instr-John.tikz}} \,, 
\end{equation}
where the small circle in the bifurcation represents the copy operation, which is available for classical systems (which are the types of systems carrying measurement inputs and outputs). The class of instrumental assemblages of interest are known as general instrumental assemblages, in contrast to the other cases in which they were no-signalling assemblages. The reason for this is that in this scenario there is explicit signalling from Alice to Bob so the term no-signalling would be inappropriate.
\begin{defn}[General instrumental assemblage]\label{ap:Instapp}
An instrumental assemblage $\Sigma^I_{\A|\X}$ is said to be a general instrumental assemblage iff it is a wiring (as per Eq.~\eqref{eq:instdiag}) of a no-signalling Bob-with-input assemblage (Def.~\ref{def:BWI}).
\end{defn}

As was said previously, not only the scenarios are important but also some types of assemblages in each scenario should be defined. By doing so, we become able to talk precisely about what post-quantumness means for steering. These types are defined as follows. 

\begin{defn}[Local Hidden State (LHS) (N+1)-Partite Assemblage]\label{def:LHSapp}
	An assemblage $ \Sigma_{\mathds{A}_{1} ... \mathds{A}_{N} | \mathds{X}_{1} ... \mathds{X}_{N}}$ in the (N+1)-partite steering scenario has a local hidden state model iff it can be prepared by the parties $A_{i}$ performing local measurements on a shared classical random variable $\lambda$, whilst $B$ prepares a quantum state conditioned on this random variable. That is,
	\begin{equation}\label{}
			\sigma_{a_{1} ... a_{n} | x_{1} ... x_{N} } 
			= \sum_{\lambda} p(\lambda ) p_{\lambda}(a_{1} | x_{1} ) ... p_{\lambda }(a_{N} | x_{N}) \rho_{\lambda}^{B}  
	\end{equation}
		for some 
		probability distribution $p(\lambda )$ over $\lambda$, and quantum states $ \rho_{\lambda}^{B}$.
\end{defn}

\begin{defn}[Quantum (N+1)-Partite Assemblage]\label{def:Qapp}
	An assemblage $ \Sigma_{\mathds{A}_{1} ... \mathds{A}_{N} | \mathds{X}_{1} ... \mathds{X}_{N}}$ in the (N+1)-partite steering scenario has a quantum realization iff it can be prepared by the parties $A_{i}$ performing local quantum measurements on a shared quantum system. That is,

	\begin{align}\label{}
		&\sigma_{a_{1} ..._{N} | x_{1} ... x_{N}} \nonumber \\
		& \qquad = \tr_{A_{1}... A_{N}} \left( M^{x_{1}}_{a_{1} } \otimes ... \otimes M^{x_{N}}_{a_{N}} \otimes \mathds{1}_{B} \cdot \rho^{A_{1}  ...  A_{N}  B }\right)\,,
	\end{align}
	for some local POVMs $\{M_{a_{i}}^{x_{i}} \}$ and joint quantum state $ \rho^{A_{1} ...  A_{N}  B}$.
\end{defn}

\begin{defn}[Local Hidden State (LHS) Bob-with-input Assemblage]
	An assemblage $ \Sigma_{\mathds{A}|\mathds{X}\mathds{Y}}$ in the Bob-with-input steering scenario has a local hidden state model if and only if it can be prepared by the parties performing local operations on a shared classical system. That is,

\begin{equation}\label{}
		\rho_{a|xy} 
		= \sum_{\lambda} p(\lambda) p_{\lambda} (a|x) \, \rho_{\lambda,y}^{B}\,,
\end{equation}
where $ \rho_{\lambda,y}^{B}$ are local quantum states. 
\end{defn}

\begin{defn}[Quantum Bob-with-input Assemblage]
	An assemblage $\Sigma_{\mathds{A}|\mathds{X}\mathds{Y}}$ in the Bob-with-input steering scenario has a quantum realization if and only if it can be prepared by the parties performing local operations on a shared quantum system. That is,
	\begin{equation}\label{}
		\sigma_{a|xy} = T_{y}\left(\tr_{A} \left[ M_{a}^{x} \otimes \mathds{1}_{B} \cdot \rho^{AB} \right]\right)
	\end{equation}
	for some local POVM $\{ M_{a}^{x}\}$, joint quantum state $ \rho^{AB}$, and quantum operations $\{ T_{y}\}$.
\end{defn}

The LHS and quantum assemblages in the instrumental scenario are defined just like in the Bob-with-input scenario, but with the constraint that $y=a$. 

With the steering scenarios and assemblages defined in a general way, we can proceed to describe how these appear within the GPT framework. We follow a similar path, starting from Bell nonlocality scenarios and legitimately use the GPT diagrams for each case.

In a GPT, a Bell scenario where the no-signalling condition is satisfied is produced when Alice (Bob) makes local measurements $ M_{A}$ ($M^{\prime}_{B}$) with input $x$ ($y$) on a shared state $s$. The set of no-signalling boxes that can be realised in such a way are equivalent to the set of `causal' classical channels, $N$, that can be realised by:

\begin{equation}\label{}
\InputIfFileExists{Diagrams/bellScenario-John.tikz}{}{\input{./figures/Diagrams/bellScenario-John.tikz}}\ \  =\ \  %
\InputIfFileExists{Diagrams/bellScenario0-John.tikz}{}{\input{./figures/Diagrams/bellScenario0-John.tikz}}\,.
\end{equation}

Again, we now let Bob be passive and perform no measurement, under the assumption that his system is a quantum one, and the resulting diagram is an assemblage element in the bipartite steering scenario:
\begin{defn}[GPT realisable assemblages] 
i) A bipartite assemblage $\Sigma_{\A|\X}$ is GPT realisable for a given GPT iff the channel associated to it can be written as:

\beq
\InputIfFileExists{Diagrams/suggestion1.tikz}{}{\input{./figures/Diagrams/suggestion1.tikz}}\ \ = \ \ %
\InputIfFileExists{Diagrams/suggestion2.tikz}{}{\input{./figures/Diagrams/suggestion2.tikz}}\,.
\eeq
Note that the transformation $T$ can be viewed as the process by which Bob characterises his system, which could be a post-quantum system, as a quantum system. This could be incorporated into the state $s$ and we could view Bob as being given a quantum system to start with, this picture, however, is useful for later generalisations.

\noindent ii) A multipartite assemblage $\Sigma_{\A_1...\A_N|\X_1...\X_N}$ is GPT realisable if and only if its associated causal channel can be written as:
\beq\label{nsteering}
\InputIfFileExists{Diagrams/suggestion5.tikz}{}{\input{./figures/Diagrams/suggestion5.tikz}}\ \ = \ \ %
\InputIfFileExists{Diagrams/suggestion6.tikz}{}{\input{./figures/Diagrams/suggestion6.tikz}}\,.
\eeq

\noindent iii) A Bob-with-input assemblage $\Sigma_{\A|\X\Y}$ is GPT realisable iff its associated channel can be written as: 
\beq\label{bobin}
\InputIfFileExists{Diagrams/suggestion3.tikz}{}{\input{./figures/Diagrams/suggestion3.tikz}}\ \ = \ \ %
\InputIfFileExists{Diagrams/suggestion4.tikz}{}{\input{./figures/Diagrams/suggestion4.tikz}}\,.
\eeq

\noindent iv) An Instrumental assemblage $\Sigma^I_{\A|\X}$ is GPT realisable if and only if it is a wiring of a GPT realisable assemblage in the Bob-with-input scenario: 

\begin{equation}\label{instr}
\InputIfFileExists{Diagrams/instr2-John.tikz}{}{\input{./figures/Diagrams/instr2-John.tikz}} \ \ = \ \ %
\InputIfFileExists{Diagrams/suggestion10.tikz}{}{\input{./figures/Diagrams/suggestion10.tikz}}\,.
\end{equation}

\end{defn}

With the definition of GPT realisable assemblages in place, we can revisit the LHS and quantum assemblages, and see how these amount to restrictions on the shared state $s$ and the GPT to which it belongs. That is, if we consider the assemblages that are realisable in the GPT of quantum theory, then, within this GPT, the GPT realisable assemblages are exactly the Quantum assemblages. If we moreover restrict to the state $s$ being a separable state, then we recover the LHS assemblages.

\section{Details on the Remote State Preparation protocol}\label{app:rsp}

In this section, we prove that 
Eq.~\eqref{eq:RSPwit} is true. Namely, that for all $\psi$ we have:
\beq
\InputIfFileExists{Diagrams/RSP.tikz}{}{\input{./figures/Diagrams/RSP.tikz}} \quad = \quad %
\begin{tikzpicture}
	\begin{pgfonlayer}{nodelayer}
		\node [style=point] (5) at (0.75, -0.75) {$\psi$};
		\node [style=none] (7) at (0.75, 1) {};
		\node [style=right label] (9) at (0.75, 0) {$\Q{2}$};
	\end{pgfonlayer}
	\begin{pgfonlayer}{edgelayer}
		\draw [q] (7.center) to (5);
	\end{pgfonlayer}
\end{tikzpicture}
}.
\eeq

To see this we note some basic results regarding the various components of the diagram on the left hand side. First note that, as we are considering a qubit system $\Q{2}$, then for each pure state $\psi$ there is a unique orthogonal state, which we denote as $\psi^\perp$. This uniqueness, in particular, means that:
\beq\label{eq:doublePerp}
\begin{tikzpicture}
	\begin{pgfonlayer}{nodelayer}
		\node [style=point] (5) at (0, -0.75) {$\psi^{\perp^\perp}$};
		\node [style=none] (7) at (0, 0.75) {};
		\node [style=right label] (9) at (0, 0.25) {$\Q{2}$};
	\end{pgfonlayer}
	\begin{pgfonlayer}{edgelayer}
		\draw [q] (7.center) to (5);
	\end{pgfonlayer}
\end{tikzpicture}
} \ \ = \ \ %
\begin{tikzpicture}
	\begin{pgfonlayer}{nodelayer}
		\node [style=point] (5) at (0, -0.5) {$\psi$};
		\node [style=none] (7) at (0, 0.75) {};
		\node [style=right label] (9) at (0, 0.25) {$\Q{2}$};
	\end{pgfonlayer}
	\begin{pgfonlayer}{edgelayer}
		\draw [q] (7.center) to (5);
	\end{pgfonlayer}
\end{tikzpicture}
}.
\eeq

Now, turning to basic properties of the diagrammatic elements we have:
\bit
\item[i.] the singlet state
\beq
\InputIfFileExists{Diagrams/RSP-State.tikz}{}{\input{./figures/Diagrams/RSP-State.tikz}}
\eeq
satisfies
\beq\label{eq:RSPState}
\InputIfFileExists{Diagrams/RSP-State2.tikz}{}{\input{./figures/Diagrams/RSP-State2.tikz}} \ \ = \ \ \frac{1}{2} \  %
\begin{tikzpicture}
	\begin{pgfonlayer}{nodelayer}
		\node [style=none] (10) at (0, 1.25) {};
		\node [style=point] (11) at (0, -0.25) {$\psi^\perp$};
		\node [style=right label] (12) at (0, 1) {$\Q{2}$};
	\end{pgfonlayer}
	\begin{pgfonlayer}{edgelayer}
		\draw [q, in=270, out=90] (11) to (10.center);
	\end{pgfonlayer}
\end{tikzpicture}
} \quad \forall \psi;
\eeq
\item[ii.] the measurement
\beq
\InputIfFileExists{Diagrams/RSP-Measurement.tikz}{}{\input{./figures/Diagrams/RSP-Measurement.tikz}}
\eeq
 is the computational basis measurement satisfying
\beq\label{eq:compMeas}
\InputIfFileExists{Diagrams/RSP-Measurement2.tikz}{}{\input{./figures/Diagrams/RSP-Measurement2.tikz}} \ \ = \ \ %
\begin{tikzpicture}
	\begin{pgfonlayer}{nodelayer}
		\node [style=copoint] (10) at (0, 0.5) {$i$};
		\node [style=none] (12) at (0, -1) {};
		\node [style=right label] (13) at (0, -0.5) {$\Q{2}$};
	\end{pgfonlayer}
	\begin{pgfonlayer}{edgelayer}
		\draw [q, in=-90, out=90] (12.center) to (10);
	\end{pgfonlayer}
\end{tikzpicture}
} \quad \forall \ i\in\{0,1\};
\eeq
\item[iii.] the family  of unitaries
\beq
\InputIfFileExists{Diagrams/RSP-Unitary.tikz}{}{\input{./figures/Diagrams/RSP-Unitary.tikz}}
\eeq
satisfies (for all $\bm{\uppsi}$)
\beq\label{eq:RSPUnitary}
\InputIfFileExists{Diagrams/RSP-Unitary2.tikz}{}{\input{./figures/Diagrams/RSP-Unitary2.tikz}} \ \ = \ \ %
\begin{tikzpicture}
	\begin{pgfonlayer}{nodelayer}
		\node [style=none] (4) at (0, -1.25) {};
		\node [style=copoint] (6) at (0, 0.25) {$\psi^\perp$};
		\node [style=right label] (8) at (0, -0.75) {$\Q{2}$};
	\end{pgfonlayer}
	\begin{pgfonlayer}{edgelayer}
		\draw [q] (4.center) to (6);
	\end{pgfonlayer}
\end{tikzpicture}
} \quad \text{and} \quad %
\InputIfFileExists{Diagrams/RSP-Unitary3.tikz}{}{\input{./figures/Diagrams/RSP-Unitary3.tikz}} \ \ = \ \ %
\begin{tikzpicture}
	\begin{pgfonlayer}{nodelayer}
		\node [style=none] (4) at (0, -1.25) {};
		\node [style=copoint] (6) at (0, 0.25) {$\psi$};
		\node [style=right label] (8) at (0, -0.75) {$\Q{2}$};
	\end{pgfonlayer}
	\begin{pgfonlayer}{edgelayer}
		\draw [q] (4.center) to (6);
	\end{pgfonlayer}
\end{tikzpicture}
};
\eeq
\item[iv.] the controlled transformation
\beq
\InputIfFileExists{Diagrams/RSP-Correction.tikz}{}{\input{./figures/Diagrams/RSP-Correction.tikz}}
\eeq
satisfies
\beq\label{eq:RSPCorr}
\InputIfFileExists{Diagrams/RSP-Correction2.tikz}{}{\input{./figures/Diagrams/RSP-Correction2.tikz}} \ \ = \ \ %
\begin{tikzpicture}
	\begin{pgfonlayer}{nodelayer}
		\node [style=none] (5) at (0, -1.5) {};
		\node [style=none] (7) at (0, 1.5) {};
		\node [style=right label] (9) at (0, 0) {$\Q{2}$};
	\end{pgfonlayer}
	\begin{pgfonlayer}{edgelayer}
		\draw [q] (7.center) to (5.center);
	\end{pgfonlayer}
\end{tikzpicture}
} \quad \text{and} \quad %
\InputIfFileExists{Diagrams/RSP-Correction3.tikz}{}{\input{./figures/Diagrams/RSP-Correction3.tikz}} \ \ = \ \ %
\InputIfFileExists{Diagrams/RSP-Correction4.tikz}{}{\input{./figures/Diagrams/RSP-Correction4.tikz}};
\eeq
\item[v.] the universal not gate
\beq
\InputIfFileExists{Diagrams/RSP-UNot.tikz}{}{\input{./figures/Diagrams/RSP-UNot.tikz}}
\eeq
satisfies
\beq\label{eq:RSPUnot}
\InputIfFileExists{Diagrams/RSP-UNot2.tikz}{}{\input{./figures/Diagrams/RSP-UNot2.tikz}} \ \ = \ \ %
\begin{tikzpicture}
	\begin{pgfonlayer}{nodelayer}
		\node [style=point] (5) at (0, -0.5) {$\psi^\perp$};
		\node [style=none] (7) at (0, 0.75) {};
		\node [style=right label] (9) at (0, 0.25) {$\Q{2}$};
	\end{pgfonlayer}
	\begin{pgfonlayer}{edgelayer}
		\draw [q] (7.center) to (5);
	\end{pgfonlayer}
\end{tikzpicture}
} \quad \forall \psi.
\eeq
\eit
Finally, note that we can decompose the classical system as a sum of projectors. Consider the basis for classical system $\C{v}$ labelled by $i \in \{1,...,v\}$, then
\begin{equation}\label{eq:classDecomp}
	\vspace{7cm}%
\begin{tikzpicture}
	\begin{pgfonlayer}{nodelayer}
		\node [style=none] (0) at (0, -1) {};
		\node [style=none] (1) at (0, 1) {};
		\node [style=right label] (2) at (0, 0) {$\C{v}$};
	\end{pgfonlayer}
	\begin{pgfonlayer}{edgelayer}
		\draw [c, in=270, out=90] (0.center) to (1.center);
	\end{pgfonlayer}
\end{tikzpicture}
} = \sum_{i=1}^{v}%
\begin{tikzpicture}
	\begin{pgfonlayer}{nodelayer}
		\node [style=none] (0) at (0, -2) {};
		\node [style=copoint] (1) at (0, -0.75) {$i$};
		\node [style=point] (2) at (0, 0.75) {$i$};
		\node [style=none] (3) at (0, 2) {};
	\end{pgfonlayer}
	\begin{pgfonlayer}{edgelayer}
		\draw [c] (3.center) to (2);
		\draw [c, in=270, out=90] (0.center) to (1);
	\end{pgfonlayer}
\end{tikzpicture}
}.
\end{equation}

From these, together with linearity of GPTs, we can conclude that the protocol does indeed work as we want:

\begin{align}
\InputIfFileExists{Diagrams/RSP.tikz}{}{\input{./figures/Diagrams/RSP.tikz}} \ \ &\stackrel{\ref{eq:classDecomp}}{=} \ \ \sum_i \  %
\InputIfFileExists{Diagrams/RSP-Proof1.tikz}{}{\input{./figures/Diagrams/RSP-Proof1.tikz}} \\
&\stackrel{\ref{eq:compMeas}}{=} \ \ \sum_i \  %
\InputIfFileExists{Diagrams/RSP-Proof2.tikz}{}{\input{./figures/Diagrams/RSP-Proof2.tikz}} \\
&\stackrel{\ref{eq:RSPUnitary}}{=} \ \ %
\InputIfFileExists{Diagrams/RSP-Proof3.tikz}{}{\input{./figures/Diagrams/RSP-Proof3.tikz}} \ + \ %
\InputIfFileExists{Diagrams/RSP-Proof4.tikz}{}{\input{./figures/Diagrams/RSP-Proof4.tikz}} \\
&\stackrel{\ref{eq:RSPCorr}}{=} \ \ %
\InputIfFileExists{Diagrams/RSP-Proof5.tikz}{}{\input{./figures/Diagrams/RSP-Proof5.tikz}} \ + \ %
\InputIfFileExists{Diagrams/RSP-Proof6.tikz}{}{\input{./figures/Diagrams/RSP-Proof6.tikz}} \\
&\stackrel{\ref{eq:RSPState}}{=} \ \ \frac{1}{2} \  %
\begin{tikzpicture}
	\begin{pgfonlayer}{nodelayer}
		\node [style=point] (5) at (0.75, -1) {$\psi^{\perp^\perp}$};
		\node [style=none] (7) at (0.75, 1) {};
		\node [style=right label] (21) at (0.75, 0.75) {$\Q{2}$};
	\end{pgfonlayer}
	\begin{pgfonlayer}{edgelayer}
		\draw [q] (7.center) to (5);
	\end{pgfonlayer}
\end{tikzpicture}
} \ + \ \frac{1}{2} \ %
\InputIfFileExists{Diagrams/RSP-Proof8.tikz}{}{\input{./figures/Diagrams/RSP-Proof8.tikz}} \\
&\stackrel{\ref{eq:RSPUnot}}{=} \ \ \frac{1}{2} \ %
\begin{tikzpicture}
	\begin{pgfonlayer}{nodelayer}
		\node [style=point] (5) at (0.75, -1) {$\psi^{\perp^\perp}$};
		\node [style=none] (7) at (0.75, 1) {};
		\node [style=right label] (9) at (0.75, 0.25) {$\Q{2}$};
	\end{pgfonlayer}
	\begin{pgfonlayer}{edgelayer}
		\draw [q] (7.center) to (5);
	\end{pgfonlayer}
\end{tikzpicture}
} \ + \ \frac{1}{2} \ %
\begin{tikzpicture}
	\begin{pgfonlayer}{nodelayer}
		\node [style=point] (5) at (0.75, -1) {$\psi^{\perp^\perp}$};
		\node [style=none] (7) at (0.75, 1) {};
		\node [style=right label] (9) at (0.75, 0.25) {$\Q{2}$};
	\end{pgfonlayer}
	\begin{pgfonlayer}{edgelayer}
		\draw [q] (7.center) to (5);
	\end{pgfonlayer}
\end{tikzpicture}
} \\
&\stackrel{\ref{eq:doublePerp}}{=} %
\begin{tikzpicture}
	\begin{pgfonlayer}{nodelayer}
		\node [style=point] (5) at (0.75, -0.75) {$\psi$};
		\node [style=none] (7) at (0.75, 1) {};
		\node [style=right label] (9) at (0.75, 0) {$\Q{2}$};
	\end{pgfonlayer}
	\begin{pgfonlayer}{edgelayer}
		\draw [q] (7.center) to (5);
	\end{pgfonlayer}
\end{tikzpicture}
},
\end{align}
which is the state $ \ket{\psi}\bra{\psi}$ chosen by Alice but prepared at Bob's lab, as required.

\end{document}